\documentclass[format=acmsmall, review=false, screen=true]{acmart}
\usepackage{booktabs} % For formal tables

\usepackage{hyperref}
\usepackage{epsfig}
\usepackage{bbm}
\usepackage{graphicx}
\usepackage{amsmath, amssymb}
\usepackage[ruled]{algorithm2e}
\usepackage{empheq}
\usepackage{dsfont}
\usepackage{thmtools,thm-restate}
\usepackage{framed}

\usepackage{enumitem}
\usepackage{letltxmacro}
\usepackage{subfig}
\usepackage{multirow}
\newtheorem{prop}{Proposition}
\newtheorem{lemma}{Lemma}
\newtheorem{thm}{Theorem}

\newcommand{\prob}{\mathbb P}
\newcommand{\E}{\mathbb E}
\newcommand{\algo}{dandelion}
\newcommand{\Algo}{Dandelion}
\newcommand{\Algopp}{{\sc Dandelion++}}
\DeclareMathOperator*{\argmin}{arg\,min}
\DeclareMathOperator*{\argmax}{arg\,max}
\DeclareMathOperator*{\detec}{\hat v = v}

\makeatletter
\newcommand{\removelatexerror}{\let\@latex@error\@gobble}
\makeatother

\begin{document}

\setcopyright{acmcopyright}
\acmJournal{POMACS}
\acmYear{2018} \acmVolume{2} \acmNumber{2} \acmArticle{29} \acmMonth{6} \acmPrice{15.00}\acmDOI{10.1145/3224424}

\title{\Algopp: Lightweight Cryptocurrency Networking with Formal Anonymity Guarantees}
\thanks{This work was supported in part by NSF grant CIF-1705007, Input Output Hong Kong (IOHK), Jump Trading, CME Group, and the Distributed Technologies Research Foundation}
\author{Giulia Fanti}
\affiliation{%
  \institution{Carnegie Mellon University}
}
\email{gfanti@andrew.cmu.edu}

\author{Shaileshh Bojja Venkatakrishnan} 
\affiliation{%
  \institution{Massachusetts Institute of Technlogy}
}
\email{shaileshh.bv@gmail.com}

\author{Surya Bakshi}
\affiliation{%
  \institution{University of Illinois at Urbana-Champaign}
 }
\email{sbakshi3@illinois.edu }

\author{Bradley Denby}
\affiliation{%
  \institution{Carnegie Mellon University}
}
\email{bdenby@andrew.cmu.edu}

\author{Shruti Bhargava}
\affiliation{%
  \institution{University of Illinois at Urbana-Champaign}
}
\email{bharshruti@gmail.com}

\author{Andrew Miller}
\affiliation{%
  \institution{University of Illinois at Urbana-Champaign}
}
\email{soc1024@illinois.edu }

\author{Pramod Viswanath}
\affiliation{%
  \institution{University of Illinois at Urbana-Champaign}
}
\email{pramodv@illinois.edu}
%\titlenote{Produces the permission block, and
%  copyright information}
%\subtitle{Extended Abstract}
%\subtitlenote{The full version of the author's guide is available as
%  \texttt{acmart.pdf} document}

% The default list of authors is too long for headers}
%\renewcommand{\shortauthors}{B. Trovato et al.}

%\ccsdesc{Security and privacy~Use https://dl.acm.org/ccs.cfm to generate actual concepts section for your paper}
% -- end of section to replace with generated code

%\keywords{template, formatting, pickling} % TODO: replace with your keywords

% TODO: replace this section with code generated by the tool at https://dl.acm.org/ccs.cfm



%\ccsdesc{Security and privacy~Use https://dl.acm.org/ccs.cfm to generate actual concepts section for your paper}
% -- end of section to replace with generated code

 % TODO: replace with your keywords

\begin{abstract}
Recent work has demonstrated significant anonymity vulnerabilities in Bitcoin's networking stack.
In particular, the current mechanism for broadcasting Bitcoin transactions allows third-party observers to link transactions to the IP addresses that originated them.
This lays the groundwork for low-cost, large-scale deanonymization attacks. 
In this work, we present \Algopp, a first-principles defense against large-scale deanonymization attacks with near-optimal information-theoretic guarantees.
\Algopp~builds upon a recent proposal called \Algo~that exhibited similar goals. 
However, in this paper, we highlight some simplifying assumptions made in \Algo, and show how they can lead to serious deanonymization attacks when violated. 
In contrast, \Algopp~defends against stronger adversaries that are allowed to disobey protocol.
\Algopp~ is lightweight, scalable, and completely interoperable with the existing Bitcoin network.
We evaluate it through experiments on Bitcoin's mainnet (i.e., the live Bitcoin network) to demonstrate its interoperability and low broadcast latency overhead. 
%Recent work has demonstrated significant anonymity vulnerabilities in Bitcoin's networking stack.
%In particular, the current mechanism for broadcasting Bitcoin transactions allows third-party observers to link transactions to the IP addresses that originated them.
%This lays the groundwork for low-cost, large-scale deanonymization attacks. 
%In this work, we present \Algopp, a first-principles defense against large-scale deanonymization attacks with near-optimal information-theoretic guarantees.
%\Algopp~ is lightweight, scalable, and completely interoperable with the existing Bitcoin network.
%We evaluate \Algopp~ through experiments on Bitcoin's mainnet (i.e., the live Bitcoin network) to demonstrate its interoperability with the current network, as well as low broadcast latency overhead. 
%%\Algopp~ is currently being considered for integration in Bitcoin Core, the most widely-used Bitcoin client. 
\end{abstract}

\maketitle

\section{Introduction}
%Cryptocurrencies like Bitcoin were originally designed to provide transparency \cite{}, often at the expense of other desirable properties.
%For example, \emph{blockchains} are public, append-only ledgers of a cryptocurrency's entire transaction history. 
%On one hand, they enable users to verify the validity of transactions; 
%on the other, their existence can hurt a cryptocurrency's scalability and anonymity \cite{}.

Anonymity is an important property for a financial system, especially given the often-sensitive nature of transactions \cite{chellappa2005personalization}. %, particularly with respect to other users.
Unfortunately, the anonymity protections in Bitcoin and similar cryptocurrencies can be fragile.  
This is largely because Bitcoin users are identified by cryptographic pseudonyms (a user can have multiple pseudonyms). 
When a user Alice wishes to transfer funds to another user Bob, she generates a \emph{transaction message} %, which specifies the quantity of Bitcoins to that should be sent to Bob's pseudonym \cite{bitcoin}. 
that includes Alice's pseudonym, the quantity of funds transferred, 
the prior transaction from which these funds are drawn, and a reference to Bob's pseudonym \cite{bitcoin}.
%The funds in this transaction are drawn from a prior transaction that transferred funds to Alice. 
%So each transaction message must also specify the previous transaction from which it draws; 
%the whole transaction is then signed by Alice to ensure that she approved the transfer of funds.
The system-wide sequence of transactions is recorded in a public, append-only ledger known as the \emph{blockchain}.
The public blockchain means that users only remain anonymous as long as their pseudonyms cannot be linked to their true identities. 

This mandate has proved challenging to uphold in practice. 
Several vulnerabilities have enabled researchers, law enforcement, and possibly others to partially deanonymize users \cite{tracking}. 
Publicized attacks have so far included: (1) linking different public keys that belong to the same user \cite{fistful}, (2) associating users' public keys with their IP addresses \cite{biryukov,koshy2014analysis}, and in some cases, (3) linking public keys to human identities \cite{bitcoin_track}. 
Such deanonymization exploits tend to be cheap, easy, and scalable \cite{fistful,biryukov,koshy2014analysis}. % accessible to a wide swath of the population.

%Some of Bitcoin's anonymity vulnerabilities are better-known than others. 
%For instance, blockchain-based vulnerabilities have been widely-publicized;
%relevant attacks typically use the public blockchain (as well as side information) to learn information about Bitcoin users.
%Significant work has gone towards mitigating these vulnerabilities \cite{mimblewimble,zcash,tumblebit,monero}.
%However, patches to the blockchain ignore a lower level source of vulnerabilities:  the peer-to-peer (P2P) network on which Bitcoin runs. 

%Design choices at the P2P layer have ramifications for the robustness, consistency, and anonymity of the system.
Although researchers have traditionally focused on the privacy implications of the blockchain \cite{ober2013structure,ron2013quantitative,fistful}, 
we are interested in lower-layer vulnerabilities that emerge from Bitcoin's peer-to-peer (P2P) network. 
Recent work has demonstrated P2P-layer anonymity vulnerabilities that allow transactions to to be linked to users' IP addresses with accuracies over 30\% \cite{biryukov,koshy2014analysis}.
%In principle, anonymity vulnerabilities at the network layer should be solvable.
%This is because networking protocols are not inherent to the functionality of Bitcoin,
%and  can therefore be changed without forcing users to adopt a new cryptocurrency or otherwise change their behavior.
%Although the authors of these attack papers proposed high-level solutions, 
Understanding how to patch these vulnerabilities without harming utility remains an open question.
The goal of our work is to propose a practical, lightweight modification to Bitcoin's networking stack that provides theoretical anonymity guarantees against the types of attacks demonstrated in \cite{biryukov,koshy2014analysis}, and others.
We begin with an overview of Bitcoin's P2P network, and explain why it enables deanonymization attacks.

\vspace{0.05in}
\noindent \textbf{Bitcoin's P2P Network.}
Bitcoin nodes are connected over a P2P network of TCP links.
This network is used to communicate transactions, the blockchain, and control packets, and it plays a crucial role in maintaining the network's consistency. 
Each peer is identified by its (IP address, port) combination.
Whenever a node generates a transaction, it broadcasts a record of the transaction over the P2P network; 
critically, transaction messages do not include the sender's IP address---only their pseudonym.
Since the network is not fully-connected, transactions are relayed according to epidemic flooding \cite{peterson2007computer}.
This ensures that all nodes receive the transaction and can add it to the blockchain.
Hence, transaction broadcasting enables the network to learn about them quickly and reliably.

\begin{figure}[t]
    \centering
  \includegraphics[width=.2\textwidth]{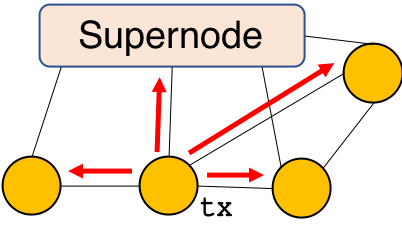}
  \caption{Supernodes can observe relayed transaction propagation metadata to infer which node was the source of a transaction message (\texttt{tx}).}
  \label{fig:supernode}
\end{figure}

However, the broadcasting of transactions can also have negative anonymity repercussions.
Bitcoin's current broadcast mechanism spreads content isotropically over the graph; this allows adversarial peers who observe the spreading dynamics of a given transaction to infer the source IP of each transaction.
For example, in recent attacks \cite{biryukov,koshy2014analysis}, researchers launched a supernode (disguised as a regular node) that connected to all P2P nodes (Figure \ref{fig:supernode}) and logged their relayed traffic.
%Since the supernode was indistinguishable from other nodes, ``honest" peers would relay transactions to it.
This allowed the supernode to observe the spread of each transaction over the network over time, and ultimately infer the source IP. 
Since transaction messages include the sender's pseudonym,  the supernodes were  able to \emph{deanonymize} users, 
or link their pseudonyms to an IP address \cite{biryukov,koshy2014analysis}.
Such deanonymization attacks are problematic because of Bitcoin's transparency: once a user is deanonymized, her other transactions can often be linked, even if she creates fresh pseudonyms for each transaction \cite{fistful}.

%The efficacy of such attacks depends on two key aspects that are unrelated to the adversary's capabilities: the underlying graph structure, and the broadcasting protocol.
%Each node in this P2P network is identified by its (IP address, port) combination, and maintains an address book of other nodes;
%the address book is initially seeded from a DNS server and subsequently updated through the node's interactions with other peers.
%Although links are technically bidirectional, the Bitcoin network treats them as directed, where an \emph{outbound} link from Alice to Bob is one that Alice initiated (and vice versa for inbound links).
%Each node establishes up to eight outbound connections, and maintains up to 125 total connections \cite{}.

%In particular, researchers showed that the flooding mechanisms used to broadcast transactions can reveal the originating IP address of  a particular transaction with accuracies up to 30\% \cite{biryukov,koshy2014analysis}.

%Hence, this is a serious problem for the cryptocurrency community \cite{}---especially since several cryptocurrencies use similar networking code to Bitcoin's \cite{}.
%The goal of this paper is to provide a practical transaction broadcasting protocol that can provide meaningful anonymity under realistic adversarial conditions.

%cryptocurrencies need a network to be useful, and patching vulnerabilities at the blockchain level is meaningless if the underlying network is insecure. 
%The proposed exploits use a supernode to monitor all traffic, and use the observed metadata to link transactions to IP addresses.

There have been recent proposals for mitigating these vulnerabilities, included broadcasting protocols that reduce the symmetry of epidemic flooding. 
%For example, 
%However, several proposed approaches---including one adopted by 
Bitcoin Core \cite{bitcoindTrickleDiffusion}, the most popular Bitcoin implementation, adopted a protocol called \emph{diffusion},  where each node spreads transactions with independent, exponential delays to its neighbors on the P2P graph.
Diffusion is still in use today.
However, proposed solutions (including diffusion) tend to be heuristic, and recent work shows that they do not provide sufficient anonymity protection~\cite{fanti2017}.
Other proposed solutions, such as \Algo~\cite{dand}, offer theoretical anonymity guarantees, but do so under idealistic assumptions that are unlikely to hold in practice.
The aim of this work is to propose a broadcasting mechanism that (a) provides provable anonymity guarantees under \emph{realistic} adversarial and network assumptions, and (b) does not harm the network's broadcasting robustness or latency.
We do this by revisiting the \Algo~system and redesigning it to withstand a variety of practical threats.

\vspace{0.05in}
\noindent \textbf{Contributions.}
The main contributions of this paper are threefold:

\noindent \emph{(1)} 
 We identify key idealistic assumptions made by \Algo~\cite{dand}, and show how anonymity is degraded when those assumptions are violated. 
In particular, \cite{dand} assumes an honest-but-curious adversary that has limited knowledge of the P2P graph topology and \emph{only observes one transaction per node}. 
If adversaries are instead malicious and collect more information over time, we show that they are able to   weaken the anonymity guarantees of \cite{dand} through a combination of attacks, including side information, graph manipulation, black hole, and intersection attacks.

\noindent \emph{(2)} 
We propose a modified protocol called \Algopp~ %\footnote{The name \Algopp~is derived from a twisted pair cable, which resembles the proposed spreading pattern.} 
that subtly changes most of the implementation choices of \Algo, from the graph topology to the randomization mechanisms for message forwarding. 
Mathematically, these (relatively small) algorithmic changes completely change the anonymity analysis by exponentially augmenting the problem state space. 
%analysis lead to exponential growth in the size of the state space for the corresponding anonymity analysis
%we demonstrate that \Algopp~ significantly strengthens the related anonymity guarantees, both theoretically and in simulation.
%\Algopp~
Using analytical tools from Galton-Watson trees and random processes over graphs, we evaluate the anonymity tradeoffs of \Algopp, both theoretically and in simulation,
against stronger adversaries.
The main algorithmic changes in \Algopp~rely on increasing the amount of information the adversary must learn to deanonymize users. 
Technically, the anonymity proofs require us to bound the amount of information a single node can pass to the adversary; on a line graph, this is easy to quantify, but on more complicated expander graphs, this requires reasoning about the random routing decisions of each node in a local neighborhood.
Exploiting the locally-tree-like properties of such graphs allows us to analyze these settings with tools from the branching process literature.
%This is the main technical challenge.

%when adversaries have stronger capabilities. 
%The resulting statistical anonymity guarantees are weaker than commonly-studied cryptographic guarantees, but enable the use of lightweight, scalable solutions.
%The key algorithmic ideas are to use more complex graph structures and non-interactive protocols wherever possible.
%\Algopp~is  being considered for integration in Bitcoin Core.\footnote{Bitcoin Core is a Bitcoin wallet used by over 70\% of active nodes. For anonymous review, the name \Algopp~ has been altered, and we have redacted links to Bitcoin Core forums discussing its potential adoption.}

\noindent \emph{(3)} 
 We demonstrate the practical feasibility of \Algopp~ by evaluating an implementation on Bitcoin's mainnet (i.e., the live Bitcoin network).
 We show that \Algopp~does not increase latency significantly compared to current methods for broadcasting transactions, and it is robust to node failures and misbehavior.

The paper is structured as follows: 
in \S\ref{sec: dand recap}, we discuss relevant work on anonymity in cryptocurrencies and P2P networks.
%In \S\ref{sec:attack}, we present a simple attack on the P2P flooding mechanism that increases the probability of deanonymization compared to prior attacks. 
In \S\ref{sec:model}, we present our adversarial model, which is based on prior attacks in the literature. % and the new attack proposed in \S\ref{sec:attack}.
\S\ref{sec: dand recap} presents \Algo~ in more detail;
 in \S\ref{sec: Dandelionpp}, we analyze \Algo's weaknesses, and propose \Algopp~ as an alternative. 
We present experimental evaluation results in \S\ref{sec:evaluation},
and discuss the implications of these results in \S\ref{sec:conclusion}.
%To this end, we abstract \Algo~into three components: its high-level policy, its implementation details, and its integration into the Bitcoin ecosystem (Figure \ref{fig:breakdown}).
%We identify gaps in \cite{dand}'s coverage of all three categories, and propose lightweight solutions . 

%\input{baselines}
%%\input{attack}
\section{Related Work} \label{sec: dand recap}
The anonymity properties of cryptocurrencies have been studied extensively. 
%Most work has focused on the blockchain. 
Several papers have exploited anonymity vulnerabilities in the blockchain \cite{androulaki2013evaluating,reid2013analysis,ober2013structure,ron2013quantitative,fistful}, suggesting that transactions by the same user can be linked, even if the user adopts different addresses \cite{fistful}.
In response, researchers proposed alternative cryptocurrencies and/or tumblers that provide anonymity at the blockchain level \cite{maxwell2013coinjoin,monero,sasson2014zerocash,coinshuffle,heilman2016tumblebit,mimblewimble}.
In 2014,  researchers turned to the P2P network, showing that regardless of blockchain implementation, users can be deanonymized by network attackers \cite{koshy2014analysis,biryukov,biryukov2015bitcoin,fanti2017}. 
Researchers were able to link transactions to IP addresses with accuracies over 30\% \cite{biryukov}. 
These attacks proceeded by connecting a supernode to most all active Bitcoin server nodes.
More recently, \cite{apostolaki2016hijacking} demonstrated the serious anonymity and routing risks posed by an AS-level attacker. 
These papers suggest a need for networking protocols that defend against deanonymization attacks.

Anonymous communication for P2P/overlay networks has been an active research topic for decades. 
Most work  relies on two ideas: randomized routing (e.g., onion routing, Chaumian mixes) and/or dining cryptographer (DC) networks.
Systems that use DC nets \cite{chaum88} are typically designed for broadcast communication, which is our application of interest.
However, DC nets are  known to be inefficient and brittle  \cite{golle2004dining}.
Proposed systems \cite{goel2003herbivore,zamani2013towards,corrigan2010dissent,anonymousscale} have improved these properties significantly, but DC networks never became scalable enough to enjoy widespread adoption in practice. 

Systems based on randomized routing are generally more efficient, but focus on point-to-point communication (though these tools can be adapted for broadcast communication). 
Early works like Crowds \cite{reiter1998crowds}, Tarzan \cite{tarzan}, and P5 \cite{sherwood2005p5}  paved the way for later practical systems, such as Tor \cite{tor} and I2P \cite{i2p}, as well as recent proposals like Drac \cite{danezis2010drac}, Pisces \cite{mittal2013pisces}, and Vuvuzela \cite{vuvuzela1}.
Our work differs from this body of work along two principal axes: (1) usage goals, and (2) analysis metrics/results.

\noindent \emph{(1) Usage goals.} 
Among tools with real-world adoption, Tor \cite{tor} is the most prominent; privacy-conscious Bitcoin users frequently use it to anonymize transmissions.%
\footnote{Network crawls show 300 hidden Bitcoin services available \url{
    https://web.archive.org/web/20180211193652/https://bitnodes.earn.com/nodes/?q=Tor\%20network}.
  Clients may alternatively use Tor exit nodes, but prior work has shown this is unsustainable and poses privacy risks~\cite{biryukov2015bitcoin}.}
%This is also discouraged because it would impose an externality on the Tor network, using up resources that are currently available for communications.
%    A potential direction for future work would be conduct the “stem phase” using Tor-like onion routing over the Bitcoin network. 
%However, Bitcoin over Tor is susceptible to eclipsing attacks \cite{biryukov2015bitcoin}.
However, expecting Bitcoin users to route their traffic through Tor (or a similar service) poses several challenges, depending on the mode of integration. % requires users to route their network traffic through a third-party service.
%If users are to manage these connections indvidually, this solution cannot help users who are unaware of Bitcoin's privacy vulnerabilities.
One option would be to hard-code Tor-like functionality into the cryptocurrency's networking stack;
for instance, Monero is currently integrating onion routing into its network~\cite{kovri}. 
However, this requires significant engineering effort; Monero's development effort is still incomplete after four years~\cite{reddit_i2p}, and no other major cryptocurrencies (Bitcoin, Ethereum, Ripple) have announced plans to integrate anonymized routing.
Principal challenges include the difficulty of implementing cryptographic protocols correctly, as well as the fact that onion routing clients need global, current network information to determine transaction paths,
whereas existing systems (and \Algopp) make local connectivity and routing decisions.
Another option would be to have users route their transactions through Tor. However, many Bitcoin users are unaware of Bitcoin's privacy vulnerabilities and/or may lack the technical expertise to route their transactions through Tor. 
%Regardlesss, this may be a viable long-term solution, which we consider interesting future work.
%Bitcoin Core developers have shown no indication that they are willing to do so.
Our goal in \Algopp is instead to propose \emph{simple, lightweight solutions} that can easily be implemented within existing cryptocurrencies, with privacy benefits for \emph{all} users.

\noindent \emph{(2) Analysis.} The differences in analysis are more subtle. 
Many of the above systems include theoretical analysis, but none provide optimality guarantees under the metrics discussed in this paper. 
Prior work in this space has mainly analyzed per-user metrics, such as probability of linkage \cite{reiter1998crowds,tarzan,vuvuzela1}. 
For example, Crowds provides basic linkability analysis \cite{reiter1998crowds}, and Danezis \emph{et al.} show that Crowds has an optimally-low probability of detection under a simple \emph{first-spy} estimator that assigns each transaction 
to the first honest node to deliver the transaction to the adversary \cite{danezis2009wisdom}. % and plays an important role in this paper.
%The analysis further assumes a complete graph topology.
Such analysis overlooks the fact that adversaries can use more sophisticated estimators based on data from many users to execute \emph{joint deanonymization}.
We consider the (more complex) problem of population-level joint deanonymization over realistic graph topologies, using more nuanced estimators and anonymity metrics. 
Indeed under the metrics we study, Crowds is provably sub-optimal~\cite{dand}. 
We maintain that joint deanonymization is a realistic adversarial model as companies are being built on the premise of providing network-wide identity analytics (e.g. Chainalysis~\cite{chainalysis}). % while also giving the adversary greater deanonymization power.
In addition to using a different anonymity metric, the protocols of prior work exhibit subtle differences that significantly change the corresponding anonymity analyses and guarantees. 
We will highlight these differences when the protocols and analysis are introduced. 

% in prior work are slightly different, which significantly changes the associated analysis. 

%We begin with a description of the \Algo~ protocol proposed in~\cite{dand}. 
The most relevant solution to our problem is a recent proposal called \Algo~\cite{dand}, which uses statistical obfuscation to provide anonymity against distributed, resource-limited adversaries (pseudocode in Appendix~\ref{app:algo}).
%The \Algo~networking policy has three components: a graph topology, a spreading mechanism that describes how to propagate over a given graph, and a dynamicity level that specifies how often the graph should change \cite{dand} 
%\noindent {\bf Topology:} \Algo~uses two graphs: the regular P2P graph $G$, and the anonymity graph $H$. 
%The anonymity graph $H$ is a randomly-selected, directed line (or cycle) graph, built with a distributed protocol over all nodes in $V$.
%The P2P graph $G$ remains unchanged with respect to its current construction. 
%\noindent {\bf Spreading mechanism:} 
\Algo~propagates transactions in two phases: (i) an anonymity (or \emph{stem}) phase, and (ii) a spreading (or \emph{fluff}) phase. 
In the anonymity phase, each message is passed to a single, randomly-chosen neighbor in an \emph{anonymity graph} $H$ (this graph can be an overlay of the P2P graph $G$).
This propagation continues for a geometric number of hops with parameter $q$. %, similar to Crowds \cite{reiter1998crowds}.
However, unlike related prior work 
(e.g., Crowds \cite{reiter1998crowds}),  different users forward their transactions along the \emph{same} path in the anonymity graph $H$, which is chosen as a directed cycle in \cite{dand}; this  small difference causes Crowds to be sub-optimal under the metrics studied here and in \cite{dand}, 
and significantly affects the resulting anonymity guarantees. 
In the spreading phase, messages are flooded over the P2P network $G$ via diffusion, just as in today's Bitcoin network. 
%\Algo~uses a directed line as its anonymity graph, so each node has only one successor on $H$.
%if the anonymity graph were not a line, the successor would be chosen randomly among outbound edges.   %is passed along the line graph $H$.
%At each hop, with probability $1-q$,
%the message recipient passes the message along to the next relay in $H$. 
%otherwise, the node transitions to the \emph{spreading phase}. 
%This ensures that a maximal number of messages take the same path in the graph before reaching any adversary, thereby creating the desired strong mixing.  
%\noindent {\bf Spreading phase:} 
%In principle, messages could be forwarded along the line graph until all nodes receive it; however, the latency of this approach would be high. 
%On the other hand, it can be shown that optimal source hiding occurs after only a few hops of forwarding. 
%As such, \Algo~also proposes a second spreading phase, in which 
\Algo~periodically re-randomizes the line graph, so the adversaries' knowledge of the graph is  assumed to be limited to their immediate neighborhood. 

Under  restrictive adversarial assumptions, \Algo~ exhibits near-optimal anonymity guarantees under a joint-deanonymization model \cite{dand}. 
%As mentioned earlier, \Algo's guarantees are based on very restrictive adversarial assumptions. 
Our work illustrates \Algo's fragility to basic Byzantine attacks and proposes a scheme that is robust to Byzantine intersection attacks.
This relaxation requires completely new analysis, which is the theoretical contribution of this paper.

\section{Model}
\label{sec:model}
%We first highlight common attributes of adversaries studied in prior work;
%we then present an analytical model of the Bitcoin P2P network---closely mirroring the models in \cite{dand}---which we use for our theoretical analysis. 

\subsection{Adversary} 
The adversaries studied in prior work exhibit two basic capabilities: creating nodes and creating outbound connections to other nodes.
At one extreme, a single supernode can establish outbound connections to every node in the network; this resembles recent attacks on the Bitcoin P2P network \cite{biryukov,koshy2014analysis} and related measurement tools \cite{coinseer,coinscope}.
At the other extreme is a botnet with many honest-but-curious nodes, each of which creates few outbound edges according to protocol. 
This captures the adversarial model in \cite{dand} and  botnets observed in Bitcoin's P2P network \cite{bitcoin_botnet}.
In this paper, we combine both models: a botnet  adversary that can corrupt some fraction of Bitcoin nodes and establish arbitrarily many outbound connections.
%We consider a botnet adversarial model that is a common threat to the Bitcoin network, and also consider metrics to quantify the average degree of anonymity experienced by users.    

%Our models and metrics closely follow~\cite{dand}. 

We model the botnet adversary as a set of colluding hosts spread over the network. %, colluding towards a common objective of deanonymizing users. 
Out of  $n$ total peers in the network, we assume a fraction $p$ (i.e., $np$ peers) are malicious. 
The botnet seeks to link transactions and their associated public keys with the IP addresses of the hosts generating those transactions.  
%Malicious hosts need not follow specified protocols, and they use all observed metadata to deanonymize honest users. % (i.e. `honest-but-curious' or `Byzantine' respectively). 
%The hosts are often malware-infected, and can be controlled remotely without the host-owners' knowledge. 
%They are also cheap, easy to access and powerful due to their distributed structure, 
%As such, Botnets are an important source of threat for the Bitcoin network, and are commonly studied as adversarial models for various Bitcoin attacks. 
%Though botnets can be used to perform a wide range of attacks, in this paper we specifically focus on the deanonymization problem.  
%In this setting, t
%Such an attack can expose the entire transaction history of deanonymized users due to the public nature of the Blockchain. 
%While IP addresses in itself does not completely reveal user identities, they can be easily used in conjuction with other side-information (geo-location etc.) for a more-thorough identity deduction.
%Hence it becomes imperative to design protocols that hide the source IP of transaction messages. 
The adversarial hosts (or \emph{spies}) need not follow protocol. 
Spies can generate as many outbound edges as they want, to whichever nodes they choose;
however, they cannot force honest nodes to create outbound edges to spies. 
The spies perform IP address deanonymization by observing the transaction propagation patterns in the network. 
Adversaries log transaction information, including timestamps, sending hosts, and control packets.
%Each time an adversary receives a transaction, metadata such as the timestamp, sending host, and any other control data are logged. 
This information, along with global knowledge (e.g., network structure) is used to deanonymize honest users.\footnote{Honest users are Bitcoin hosts that are not part of the adversarial botnet. We assume honest users follow the specified protocols.} 

We assume adversaries are interested in \emph{mass deanonymization}, and
our anonymity metrics (Section \ref{sec: anon met PR}) capture the adversary's success at both the individual level and the population level.
This differs from a setting where the adversary seeks to deanonymize a targeted user. 
While the latter is a more well-studied problem \cite{tracking,biryukov2015bitcoin,androulaki2013evaluating}, 
our adversarial model is motivated in part by the growing market for wide-area cryptocurrency analytics  (e.g., Chainalysis \cite{chainalysis}).
%Moreover, typical solutions tend to require hosts to change their behavior, e.g. by adopting a new cryptocurrency.
Our goal is to provide network-wide anonymity that does not require users to change their behavior. 
While this approach will not stop targeted attacks, it does provide a first line of defense against broad deanonymization attacks that are currently feasible.

Note that recent work on ISP- or AS-level adversaries \cite{apostolaki2016hijacking} can be modeled as a special case of this botnet adversary, except \emph{edges} rather than nodes are corrupted. 
This adversary is outside the scope of this paper, but the topic is of great interest. 
%, and we expect that the intuitions developed in this paper may be useful for defending against such an adversary.
A principal challenge with ISP-level adversaries is that they can eclipse nodes; under such conditions, routing-based defenses (like \Algopp) cannot provide any guarantees for a targeted node. 
%Such a powerful adversary may be able to achieve high values of $p$, exceeding 80 or 90 percent. 
%We maintain that the statistical approaches studied in this work are not appropriate for such an adversary;
%several anonymity guarantees still hold, but they are too weak to be meaningful. 
Nonetheless, in \S\ref{sec:conclusion}, we discuss the compatibility of our proposed methods with the countermeasures proposed in \cite{apostolaki2016hijacking} for large-scale adversaries.

\subsection{Anonymity Metrics: Precision and Recall} 
\label{sec: anon met PR}
The literature on anonymous communication has proposed several metrics for studying anonymity. 
The most common of these captures an adversary's ability to link a single transaction to a single user's IP address; this \emph{probability of detection}, and variants thereof,
have been the basis of much anonymity analysis \cite{chaum88,reiter1998crowds,danezis2009wisdom,KFSV14,vuvuzela1}.
However, this class of metrics does not account for the fact that adversaries can achieve better deanonymization by observing other users' transactions. 
We consider such a form of \emph{joint decoding}.

%\red{
%More recently, a general anonymity framework called AnoA was proposed, which defines anonymity in terms of transaction-source indistinguishability \cite{backes2013anoa}. 
%Under this framework, the likelihood of some observed transaction metadata conditioned on any two (adversarially-chosen) candidate sources should be within some constant factor of one another; this definition closely resembles differential privacy.
%However, it is easy to find counterexamples where \emph{no scheme} can provide AnoA indistinguishability. 
%For example, consider a graph topology in which malicious nodes surround a small set of honest nodes, $S$; this event can happen randomly, or due to malicious peers hijacking the network \cite{biryukov2015bitcoin}. 
%The malicious nodes can perfectly distinguish whether a transaction originated from $S$ or not, for \emph{any} key-less routing policy.
%Hence, AnoA is not a meaningful metric for our problem. 
%}

The adversary's goal is to associate transactions with users' IP addresses through some association map. 
This association map can be interpreted as a classifier that classifies each transaction (and its corresponding metadata) to an IP address.
Hence an adversary's deanonymization capabilities can be measured by evaluating the adversary's associated classifier. 
%The notions of \emph{precision} and \emph{recall} are common metrics for evaluating classifiers, and a natural choice of metric for this problem.
%As in \cite{dand}, 
We adopt a common metric for classifiers: {\em precision} and {\em recall}. %,  common metrics for evaluating classifiers through complementary notions of exactness and completeness, respectively. 
%This metric is also used in \cite{dand}.
As discussed in \cite{dand}, precision and recall are a superset of the metrics typically studied in this space; 
in particular, recall is equivalent (in expectation) to probability of detection. %, which is studied widely in prior work \cite{reiter1998crowds}.
On the other hand, \emph{precision}  can be interpreted as a measure of a node's plausible deniability; the more transactions get mapped to a single node, the lower the adversary's precision.

Let $V_H$ denote the set of all IP addresses of honest peers in the network, and let $\tilde{n} = |V_H|$ denote the number of honest peers.
In this work, a \emph{transaction} is abstracted as a tuple containing the sender's address, the recipient's address, and a payload. 
To begin, we will assume that each peer  $v\in V_H$ generates exactly one transaction $X_v$.
We relax this assumption in \S\ref{sec:intersection}. %, which it broadcasts over a fixed time-horizon.
%This assumption (also made in \cite{dand}) is not necessarily representative of real traffic. 
%However, modeling real traffic patterns is a research challenge on its own, 
%and the anonymity problem remains challenging even under the one-transaction-per-node assumption.
%We discuss the implications of relaxing this  assumption in \S\ref{sec:conclusion}, and leave its resolution to future work.
Let $\mathcal{X} = \cup_{v\in V_H} X_v$ denote the set of all  transactions.  
We assume the sets $V_H$ and $\mathcal{X}$ are known to the adversaries.
Let $\mathtt{M}:\mathcal{X}\rightarrow V_H$ denote the adversary's  map from  transaction $x\in\mathcal{X}$  to IP address $\mathtt{M}(x)\in V_H$. 
The precision and recall of  $\mathtt{M}$ at any honest peer $v\in V_H$ are given, respectively, by
\begin{align}
D(v) &= \frac{\mathds{1}(\mathtt{M}(X_v) = v)}{\sum_{u\in V_H} \mathds{1}(\mathtt{M}(X_u) = v)} \\
R(v) &= \mathds{1}(\mathtt{M}(X_v) = v), 
\end{align}
where $\mathds{1}(\cdot)$ denotes the indicator function.
Precision (denoted $D(v)$) measures accuracy by normalizing against the number of transactions associated with $v$.
A large number of transactions mapped to $v$ implies a greater plausible deniability for $v$. %, and hence a lower accuracy.
Recall (denoted  $R(v)$) measures the accuracy or completeness of the mapping.
We define the average precision and recall for the network as 
%\begin{align}
$D = \frac{\sum_{v\in V_H} D(v)}{|V_H|}$ and $R =  \frac{\sum_{v\in V_H} R(v)}{|V_H|}$. \label{eq: avg prec and rec}
%\end{align}
In our theoretical analyses under a probabilistic model (see \S\ref{sec: network model}), we will be interested in the expected values of these quantities, denoted by $\mathbf D$ and $\mathbf R$, respectively. % in Equation~\eqref{eq: avg prec and rec}.   

\vspace{0.05in}
\noindent \textbf{Fundamental bounds.} \cite{dand} shows fundamental lower bounds on the expected precision and recall of any spreading mechanism. 
In particular, no spreading algorithm can achieve lower expected recall than $p$, where $p$ is the fraction of spy nodes, or expected precision lower than $p^2$. 
These lower bounds assume an honest-but-curious adversary, in which case the recall lower bound is tight, and the precision bound is within a logarithmic factor of tight. 
However, the solution in \cite{dand} does not consider Byzantine adversaries. 
Hence in this work, we will aim to match these fundamental lower bounds for an honest-but-curious adversary, while also providing robustness against a \emph{Byzantine} adversary. 
Obtaining tight lower bounds for general Byzantine adversaries remains an open question, though the lower bounds from \cite{dand} naturally still hold.
We will show in Section \ref{sec:intersection} that for certain classes of Byzantine adversaries, we can achieve \cite{dand}'s lower bound on expected recall and within a logarithmic factor of its lower bound on precision.

\subsection{Transaction and Network Model} \label{sec: network model}
We follow the probabilistic network model of \cite[\S 2]{dand}.
%This allows us to evaluate candidate policies analytically or via simulation, and obtain key insights into their performance behavior.
%\vspace{0.in}
%\noindent {\bf Transaction model.} 
%Following the same notation as in \S\ref{sec: anon met PR}, 
%Let $V_H$ denote the set of IP addresses of current honest users in the network, $X_v$ the message from user $v\in V_H$ and $\mathcal{X}$ the set of all transaction messages. 
We assume a uniform prior on $X_v$ over the set $\mathcal{X}$, i.e., the ordered tuple $(X_{v_1},X_{v_2},\ldots,X_{v_{\tilde{n}}})$ is a uniform random permutation of messages in $\mathcal{X}$ where $V_H = \{v_1,v_2,\ldots,v_{\tilde{n}}\}$. 
We also assume transaction times are unknown to the adversary. 

\begin{algorithm}[t]
\DontPrintSemicolon
\KwIn{Set $V=\{v_1, v_2, \ldots, v_n\}$ of nodes; }
\KwOut{A connected, directed graph $G(V,E)$ with average degree $2\eta$}
\For{$v \gets V$} {
%  \For{$c \gets 1$ \textbf{to} $d$} {
    \tcc{pick $\eta$ random targets} 
    $N \gets \emptyset$ \;
    \For{$i \gets \{1,\ldots, \eta\}$}{
	    $e  \sim$ Unif$(V \setminus \{v\} \setminus N )$ \;
	    $N \gets N \cup \{e\}$
    }
    \tcc{make connections}
    $E = E \cup \{(v \rightarrow u), ~u \in N \}$ 
%  } 
}
\Return{$G(V,E)$}\;
\caption{{\sc Approximate $2\eta$-Regular Graph} }
\label{algo:dreg_approx}
\end{algorithm}

%\begin{figure}[!htb]
%\centering
%  \includegraphics[width=1.6in]{figures/dandelion}
%  \caption{Dandelion spreading forwards messages in a line over the graph, then broadcasts using diffusion. Figure reproduced with permission from \cite{dand}.}\label{fig:dandelion}
%  \end{figure}

%\begin{figure}
%\centering
%  \includegraphics[width=2in]{figures/dandelion_policy}
%  \caption{\Algo~networking policy: (1) \algo~spreading, (2) line topology, (3) frequent reshuffling of the anonymity graph $H$.
%  Figure reproduced with permission from \cite{dand}.}\label{fig:dandelion_policy}
%\end{figure}

%\begin{figure*}[!htb]
%\minipage{0.27\textwidth}
%  \includegraphics[width=\linewidth]{figures/dandelion}
%  \caption{Dandelion spreading forwards messages in a line over the graph, then broadcasts using diffusion. Figure reproduced with permission from \cite{dand}.}\label{fig:dandelion}
%\endminipage\hfill
%%
%%
%\minipage{0.32\textwidth}
%  \includegraphics[width=\linewidth]{figures/dandelion_policy}
%  \caption{The \Algo~networking policy: (1) \algo~spreading, (2) a line topology, (3) frequent reshuffling of the anonymity graph $H$.
%  Figure reproduced with permission from \cite{dand}.}\label{fig:dandelion_policy}
%\endminipage\hfill
%%
%%
%\minipage{0.37\textwidth}%
%%  \includegraphics[width=\linewidth]{figures/breakdown}
%%  \caption{\Algo~attributes discussed in \S\ref{sec: Dandelionpp}.}\label{fig:breakdown}
%\endminipage
%\end{figure*}

%\vspace{0.1in}
%\noindent {\bf Network model.} 
We model the Bitcoin network as a directed graph $G(V,E)$ where the vertices $V=V_H \cup V_A$ comprise honest peers $V_H$ and adversarial peers $V_A$.
The edges $E$ correspond to TCP links between peers in the network.
Although these links are technically bidirectional, the Bitcoin network treats them as directed.
An \emph{outbound} link from Alice to Bob is one that Alice initiated (and vice versa for inbound links);
we also refer to the \emph{tail} node of an edge as the one that originated the connection.

To construct the Bitcoin network, each node establishes up to eight outbound connections, and maintains up to 125 total connections \cite{bitcoind}.
In practice, the eight outbound connections are chosen from each node's locally-maintained address book; we assume each node chooses outbound connections randomly from the set of all nodes (Algorithm \ref{algo:dreg_approx}). 
This graph construction model results in a random graph where each node has expected degree $2\eta=16$.
Although this approximates the behavior of many nodes in Bitcoin's P2P network, Byzantine nodes need not follow protocol.  
We will describe the behavior of Byzantine nodes as needed in the paper.
%Under any protocol that specifies a fixed, unchanging network topology, we assume the peers in $V$ are equally likely to be located as any node in the graph. 
%We will discard this assumption as soon as we address Byzantine nodes that disobey graph construction protocols (i.e., starting in \S\ref{sec: const graph}).
%For simplicity, we do not model peer churn.

%\vspace{0.1in}
%\noindent {\bf Observations.}  
Upon creating a transaction message, peers propagate it according to a pre-specified spreading policy. 
The propagation dynamics are observed by the spies, whose goal is to estimate the IP addresses of transaction sources. 
For each transaction $x\in\mathcal{X}$ received by adversarial node $a\in V_A$, the tuple $(x,v,t)$ is logged where $v$ is the peer that sent the message to $a$, and $t$ is the timestamp when the message was received. 
The botnet adversary may also know partial information about the network structure.
Clearly, peers neighboring adversarial nodes are known. 
However in some cases the adversary might also be able to learn the locations of honest peers not directly connected to botnet nodes, either through adversarial probing or side information. 
For simplicity, we use $\mathbf{O}$ to denote all observed information---message timestamps, knowledge of the graph, and any other control packets---known to the adversary. 
Given these observations, one common source estimator is the simple-yet-robust \emph{first-spy estimator}, used in \cite{biryukov,koshy2014analysis,danezis2009wisdom}.
Recall that the first-spy estimator outputs the first honest node to deliver a given transaction to the adversary as the source. % and plays an important role in this paper.

%\subsection{Objectives}
\begin{table*}[t]
{\tiny
\centering
\caption{Summary of changes proposed in \Algopp, with references to relevant evidence and/or analysis.}
\label{tab:summary}
\begin{tabular}{|c|c|c|c|}
\hline
 \textbf{Attack}                    &  \textbf{Effect on \Algo~\cite{dand}}             & \multicolumn{2}{c|}{\Algopp}                                    \\ %\hline
 %\cline{3-4}
  & & \textbf{Proposed solution} & \textbf{Effect}  \\ \hline
Graph-learning (\S\ref{sec: topology})    & Order-level precision increase \cite{dand} & 4-regular anonymity graph & Limits precision gain (Thm. \ref{thm: dregular result}, Fig. \ref{fig:4regular})  \\ \hline
Intersection (\S\ref{sec:partial})                  & Empirical precision increase (Fig. \ref{fig:int_attack})                                       & Pseudorandom forwarding & Improved robustness (Thm. \ref{thm:intersection}) \\ \hline
Graph-construction (\S\ref{sec: const graph}) & Empirical precision increase (Fig. \ref{fig:misbehave})      & Non-interactive construction & Reduces precision gain (Figs. \ref{fig:approx_vs_reg}, \ref{fig:behave})   \\ \hline
%Graph construction attack (add links) (\S\ref{sec: const graph})           & Limited precision increase (Fig. \ref{fig:behave})      & Use small $q$ parameter & Reduces precision gain (Fig. \ref{fig:behave})   \\ \hline
Black-hole  (\S\ref{sec: msg frwd})                  & Transactions do not propagate                                     & Random stem timers & Provides robustness  (Prop. \ref{prop: clock rate}) \\ \hline
Partial deployment (\S\ref{sec:partial})                  & Arbitrary recall increase (Fig. \ref{fig:recall_partial})                                       & Blind stem selection & Improves recall (Thm. \ref{thm:recall_partial}, Fig. \ref{fig:recall_partial}) \\ \hline
\end{tabular}
}
\end{table*}

\section{\Algopp} \label{sec: Dandelionpp}

\Algo's theoretical anonymity guarantees make three idealized assumptions:
%The goal of this section is twofold: first, we highlight the unrealistic assumptions that underlie \Algo~ and show how they can undermine its anonymity guarantees.
%Then, we propose a modified networking policy called \Algopp~that is robust to these weaknesses.
%We demonstrate robustness through a combination of theoretical results and simulations.
%\Algo's theoretical guarantees rely on three primary assumptions:
(1) all nodes obey the protocol, 
(2) each node generates exactly one transaction,
(3) all Bitcoin nodes run \Algo.
None of these assumptions necessarily holds in practice.  %and it is unclear how their absence affects \Algo's anonymity properties.
In this section, we show how \Algo's anonymity properties break when the assumptions are violated,
and propose a modified solution called \Algopp~ that addresses these concerns. 
\Algopp~passes transactions over intertwined paths, or `cables', before diffusing to the network (Fig. \ref{fig:twistedpair}).
In practice, these cables may be fragmented (i.e. all nodes are not connected in a single Hamiltonian cycle), 
but the cable intuition applies within each node's local neighborhood.

\begin{figure}[!htb]
\centering
  \includegraphics[width=2.6in]{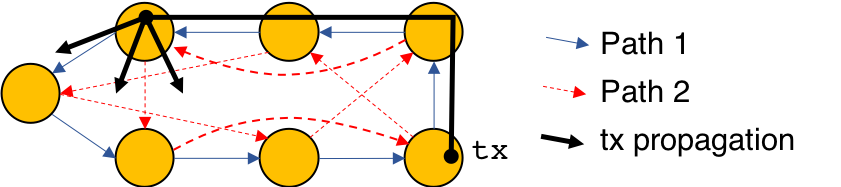}
  \caption{\Algopp~ forwards messages over one of two intertwined paths on a 4-regular graph, then broadcasts using diffusion. 
  Here, $\texttt{tx}$ propagates over the blue solid path. }\label{fig:twistedpair}
  \end{figure}

We begin with a brief description of \Algopp, and then give a more nuanced picture of how this design came about. 
We start with \Algo~as a baseline, adopting the same `stem phase' and `fluff phase' terminology, along with dandelion spreading (Algorithm \ref{algo:dandelion}).
Like \Algo, \Algopp~proceeds in asynchronous epochs; each node advances its epoch when its internal clock reaches some threshold (in practice, this will be on the order of 10 minutes). 
Within an epoch, the main algorithmic components of \Algopp~are: 
%\begin{enumerate}
\vspace{0.05in} \newline
\noindent \emph{(1) Anonymity Graph:} Use a random, approximately-4-regular graph instead of a line graph for the anonymity phase (\S\ref{sec: topology}). 
This quasi-4-regular graph is embedded in the underlying P2P graph by having each node choose (up to) two of its outbound edges, without replacement, uniformly at random as \Algopp~relays (\S\ref{sec: const graph}). 
The choice of \Algopp~relays should be independent of whether the outbound neighbors support \Algopp~or not (\S\ref{sec:partial}). 
Each time a node changes epoch, it selects fresh \Algopp~relays.
\vspace{0.05in} \newline
\noindent \emph{(2) Transaction Forwarding (own):} Every time a node generates a transaction of its own, it forwards the transaction, in stem phase, along the \emph{same} outbound edge in the anonymity graph. In \Algo, nodes are assumed to generate only one transaction, so this behavior is not considered in prior analysis. 
\vspace{0.05in} \newline
\noindent \emph{(3) Transaction Forwarding (relay):} 
Each time a node receives a stem-phase transaction from another node, it either relays the transaction or diffuses it. 
The choice to diffuse transactions is pseudorandom, and is computed from a hash of the node's own identity and epoch number. 
%I.e., each node $v$ computes $H_d(v || e_v(t))$, where $e_v(t)$ denotes the epoch of node $v$ at time $t$, $H_d(\cdot)$ denotes a hash function to $\mathbb Z_{1/q}$, and $(\cdot || \cdot)$ denotes concatenation.
%If $H_d(v || e_v(t))=0$, then $v$ begins to diffuse the transaction. 
Note that the decision to diffuse does not depend on the transaction itself---in each epoch, a node is either a diffuser or a relay node for \emph{all} relayed transactions.
If the node is not a diffuser in this epoch (i.e., it is a relayer), then it relays transactions pseudorandomly; each node maps each of its incoming edges in the anonymity graph to an outbound edge in the anonymity graph (with replacement).
This mapping is selected at the beginning of each epoch, and determines how transactions are relayed (\S\ref{sec:intersection}).
\vspace{0.05in} \newline
\noindent \emph{(4) Fail-Safe Mechanism:} Each node tracks, for each stem-phase transaction that was sent or relayed, whether the transaction is seen again as a fluff-phase transaction within some random amount of time. If not, the node starts to diffuse the transaction   (\S\ref{sec: msg frwd}).

%\end{enumerate}
%However, the primary contribution of \Algopp~is not its algorithms, which are simple extensions of known algorithms, but the associated analysis and evaluation.
These small algorithmic changes completely alter the  anonymity analysis by introducing an exponentially-growing state space. 
For example, moving from a line graph to a 4-regular graph (item (1)) invalidates the exact probability computation in \cite{dand}, and requires a more complex analysis to understand effects like intersection attacks. 
%The most challenging attacks to analyze were graph-learning attacks, intersection attacks, and the effects of deployment.
We also simulate the proposed mechanisms for all attacks and find improved anonymity compared to \Algo.\footnote{Code for reproducing simulation results can be found at \url{https://github.com/gfanti/dandelion-simulations}.}

The remainder of this section is structured according to Table~\ref{tab:summary}.
The weak adversarial model in~\cite{dand} enables five distinct attacks:
graph learning attacks, intersection attacks, graph construction attacks, black hole attacks, and deployment attacks.  
For each attack, we first demonstrate its impact on anonymity (and/or robustness); in many cases, these effects can lead to arbitrarily high deanonymization accuracies. 
Next, we propose lightweight implementation changes to mitigate this threat, and justify these choices with theoretical analysis and simulations.

%point (1) above helps with the graph learnability problem; we show that 4-regular graphs have an anonymity performance close to line graphs (i.e., near-optimal), but are harder to learn. 
%Moreover, even if the topology is learned, the resulting loss in anonymity is significantly smaller than with line graphs (\S\ref{sec: topology}). 
%Next, we identify the steps in \Algo~ that are susceptible to Byzantine attacks---namely, the graph-construction and message-forwarding protocols. 
%We discuss these attacks, and present a robust alternative algorithm for the two steps. 
%This is discussed in \S\ref{sec: const graph} and \S\ref{sec: msg frwd} respectively. 
%We then analyze the performance of \Algopp~ under partial adoption; we show that \Algopp~provides anonymity gains over Bitcoin's current approach even when only a fraction of nodes are using it (\S\ref{sec:partial}). 
%Finally, we summarize our findings and conclusions from this section in \S\ref{sec:summary}.

%\subsection{Summary}
%\label{sec:summary}
%This section has highlighted a number of practical weaknesses in \cite{dand}, while also evaluating and proposing solutions  to each of them. 
%Table \ref{tab:summary} summarizes the results from this section by explicitly listing the discussed weaknesses, their effects on anonymity, and references to the proposed solutions and associated theoretical or simulated evaluations. 

\subsection{Graph-Learning Attacks} \label{sec: topology}

%We first explain how line graphs can be learned from \Algo's spreading policy. 
%Consider any continuous segment of the line graph as shown in Figure~\ref{fig:spies} where $s$ and $t$ are adversarial nodes and the interior nodes $u_1,\ldots,u_k$ are all honest. 
%In the anonymity phase each node forwards transactions to its neighbor on the right, as they are received.
%Ideally the adversarial nodes $s,t$ should known only their immediate neighbors, $u_1$ and $u_k$ respectively, within the segment. 
%However supposing a transaction \texttt{tx} transitions from the anonymity to the spreading phase at node $u_i$.
%In this case the adversaries observe that (i) \texttt{tx} has been received and forwarded by $s$ but not by $t$, and (ii) a node $u_i$ broadcast \texttt{tx}. 
%This implies that node $u_i$ must lie in the segment between $s$ and $t$. 
%Thus any time a transaction is broadcast by an honest node, the segment to which the honest node belongs to can be deduced.  
%We note that the adversaries can also deliberately generate transactions in order to `probe' the line graph in the above fashion.
%\begin{figure}[h]
%    \centering
%  \includegraphics[width=.33\textwidth]{figures/spies}
%  \caption{The adversary can easily learn line graphs.}
%  \label{fig:spies}
%\end{figure}

Theoretical results in \cite{dand} assume that the anonymity graph is \emph{unknown} to the adversary; that is, the adversary knows the randomized graph construction protocol, but it does not know the realization of that protocol outside its own local neighborhood.
Under these conditions, \Algo, which uses a line topology, achieves a maximum expected precision of $O(p^2 \log(1/p))$, which is near-optimal (recall that $p$ denotes the fraction of malicious nodes in the network).
However, if an adversary somehow learns the anonymity graph, the maximum expected precision increases to $O(p)$~\cite[Proposition 4]{dand}---an order-level increase.
%Hence, precision guarantees are significantly degraded if the adversary somehow learns the graph. 
On the other hand, recall guarantees do not change if the adversary learns the graph. 
This is because under dandelion spreading (Algorithm \ref{algo:dandelion}), the first-spy estimator is recall-optimal \cite{dand}, meaning it maximizes the adversary's expected recall. 
Since the first-spy estimator is graph-independent, learning the graph does not improve the adversary's maximum expected recall.  
A natural question is whether one can avoid this jump in precision.
%For reference, Figure \ref{fig:theoretical} plots these theoretical bounds on expected precision as a function of $p$ (lower is better).
%The red curve (unknown graph) is an upper bound that is tight for small $p$. 
%Although the lines in Figure \ref{fig:theoretical} cross around $p=0.3$, the actual precision for a known graph always upper bounds the precision for an unknown graph, 
%since giving the adversary additional information cannot reduce the precision.

%\begin{figure}[t]
%    \centering
%  \includegraphics[width=.38\textwidth]{figures/theoretical_bounds}
%  \caption{Theoretical precision bounds from \cite{dand}.}
%  \label{fig:theoretical}
%\end{figure}

\Algo~proposes a heuristic solution in which the line graph is periodically reshuffled to a different (random) line graph \cite{dand}. 
The intuition is that changing the line graph frequently does not allow enough time for the adversaries to learn the graph. 
However, the efficacy of this heuristic is difficult to evaluate. %, particularly since Bitcoin users may perform bursts of transactions, and traffic patterns are complex to model \cite{ron2013quantitative}. 
%Each message's \Algo~stem can end at different, nearby nodes of the anonymity graph, enabling the adversary to learn portions of the anonymity graph topology faster than expected. 
%Moreover, frequent rearrangements of the graph may not provide enough transaction mixing for anonymity. % and the implications of this time-dependency are difficult to analyze.  
%These two effects---increased number of transactions per user and more frequent rearrangements---can compound to significantly degrade the algorithm's anonymity performance. 
Fundamentally, the  problem is that it is unclear how fast an adversary can learn a graph, which calls into question the resulting anonymity guarantees. 
%The time-dynamics of transactions complicate the matter even further, making an optimal rearrangement strategy %that precisely accounts for these effects 
%infeasible to design. 
%\red{would be nice to have a plot demonstrating this degradation...}

\subsubsection{Proposal: 4-Regular Graphs}
We explore an alternative solution that may protect against adversaries that are able to learn the anonymity graph---either due to the \Algopp~protocol itself or other implementation issues.
In particular, we suggest that \Algopp~ should use random, directed, 4-regular graphs instead of line graphs as the anonymity graph topology. 
4-regular graphs naturally extend line graphs, which are 2-regular. 
For now, we study exact 4-regular graphs, though the final proposal uses approximate 4-regular graphs.
%  an union of the edges of two random, 2-regular (line) graphs. 
Although the forwarding mechanism is revisited in \S\ref{sec:intersection}, for now, let us assume that users relay transactions randomly to one of their 2 outbound neighbors in the 4-regular graph until fluff phase. 
%As with the line graph, the anonymity provided by the 4-regular graph changes depending on the adversary's knowledge of the topology. 
%However, the difference in this case is smaller than for line graphs, making 4-regular graphs a more robust choice for spreading. 

Although theoretical analysis of expected precision on $d$-regular anonymity graphs is challenging for $d>2$, we instead simulate randomized spreading over different topologies, while measuring anonymity empirically using theoretically-optimal (or near-optimal) estimators.
Figure~\ref{fig:4regular} plots the average precision obtained on $d$-regular graphs, for $d=2$ (corresponding to a line graph), 4 and 6, as a function of $p$, the fraction of adversaries, for a network of 50 nodes. 
Since we are (for now) studying exact 4-regular graphs, the adversarial nodes also have degree 4; this assumption will be relaxed in later sections.
The blue solid line at the bottom corresponds to the line graph when the graph is unknown to the adversary; this matches the theoretical precision of $O(p^2\log (1/p))$ shown in \cite{dand}.
The solid lines in Figure \ref{fig:4regular} (i.e., unknown graph) were generated by running the first-spy estimator, which maps each transaction to the first honest node that forwarded the transaction to an adversarial node.  
When the graph is \emph{unknown}, we show in Theorem \ref{thm: dregular result} that the first-spy estimator is within a constant factor of optimal; 
hence, Figure \ref{fig:4regular} shows results from an approximately-precision-optimal estimator. 

\begin{figure}[t]
    \centering
  \includegraphics[width=.45\textwidth]{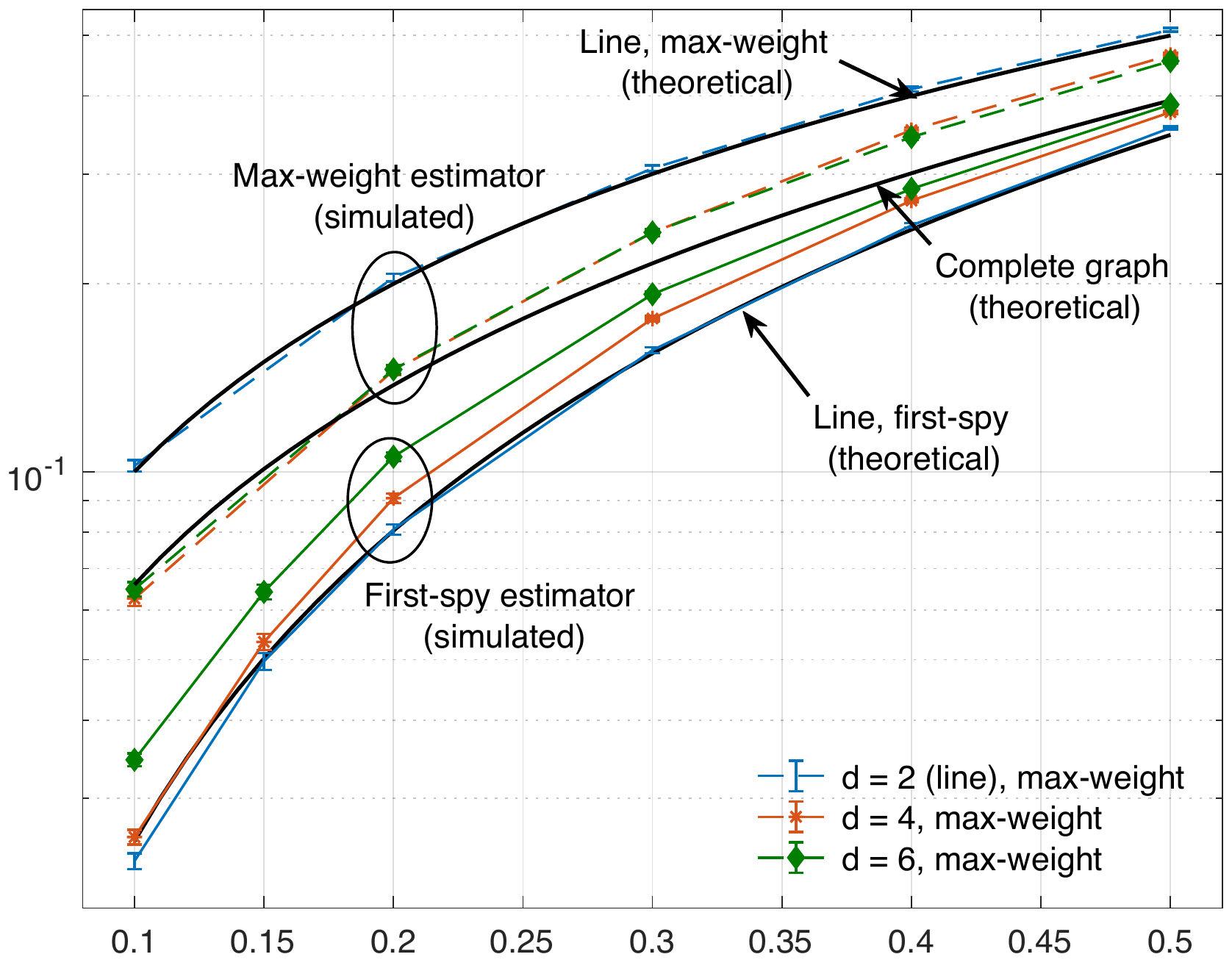}
  \put(-120,-8){Fraction of spies, $p$}
  \put(-190,40){\rotatebox{90}{Average precision}}
  \caption{Average precision as a function of $p$ for random, directed $d$-regular graphs.}
  \label{fig:4regular}
\end{figure}

When the graph is \emph{known}, we approximate the maximum-precision estimator differently.
\cite{dand} showed that the precision-optimal estimator is a maximum-weight matching from transactions to nodes, 
where the weight of the edge between node $v$ and transaction $x$ is the posterior probability $\prob(X_v=x | \mathbf O)$ (Theorem 3, \cite{dand}).
Since the anonymity graph contains cycles, this posterior  is difficult to compute exactly, 
because it requires (NP-hard) enumeration of every path between a candidate source and the first spy to observe a given transaction \cite{schrijver2002combinatorial}.
%This problem is known to be NP-hard .
We therefore approximate the posterior probabilities by assuming that each candidate source can only pass a given message along the \emph{shortest} path to the spy that first observes the message 
(note that the shortest path is also the most likely one).
This path likelihood can be computed exactly and used as a proxy for the desired posterior probability.
Given these approximate likelihoods, % for each candidate source and each transaction, 
we compute a maximum-weight matching, and calculate the precision of the resulting matching.
%To test the validity of this approach, we also computed the exact precision-optimal estimator on smaller graphs (30 nodes), but found that the difference in average precision was small. 

Assuming the adversary knows the graph and uses this quasi-precision-optimal estimator,  the precision on a line graph increases to the blue dotted line at the top of the plot. 
For example, at $p=0.15$, knowing the graph gives a precision boost of 0.12, or about $250$\% over not knowing the graph. 
On the other hand, if a 4-regular graph is unknown to the adversary, it has a precision very close to that of line graphs (orange solid line in Figure~\ref{fig:4regular}). 
But if the graph becomes known to the adversary (orange dotted line), the increase in precision is smaller. 
At $p=0.15$, the gain is $0.06$---half as large as the gain for line graphs.
This suggests that 4-regular graphs are more robust than lines to adversaries learning the graph,
while sacrificing minimal precision when the adversary does \emph{not} know the graph. % thus the precision gain is half as large compared to line graphs.

Figure~\ref{fig:4regular} also highlights a distinct trend in precision values, as the degree $d$ varies.  
If the adversary has not learned the topology, then line graphs have the lowest expected precision and hence offer the best anonymity. 
As the degree $d$ increases, the expected precision progressively worsens until $d=n$, where $n$ is the number of peers in the network. 
When $d=n$ (i.e., the graph is a complete graph), %can also be thought of as %a `diffusion-by-proxy' strategy~\cite{dand} 
%strategy wherein 
peers forward their transactions to a random peer in each hop of the dandelion stem. 
%This scenario was termed `diffusion-by-proxy' in \cite{dand}, and its expected precision was characterized theoretically, as shown by the middle black line in Figure \ref{fig:4regular}.
On the other hand, if the topologies are \emph{known} to the adversary, then the performance trend reverses. 
In this case, line graphs (for $d=2$) have the {\em worst} (highest) expected precision among random $d$-regular graphs; 
As the degree increases, precision decreases monotonically until $d=n$. 
%Note that at $d=n$ the performance of both the known and unknown topologies are identical. 

Finally, recall that our simulations used the `first-spy' estimator as the deanonymization policy for $d$-regular graphs when the graph is unknown to the adversary. 
\emph{A priori}, it is not clear whether this is an optimal estimator for $d$-regular graphs for $d\geq 4$ (\cite{dand} shows that the first-spy estimator is optimal for line graphs). 
As such the optimal precision curve could be much higher than the one plotted in Figure~\ref{fig:4regular}. 
However, the following theorem---our first main theoretical contribution of this paper---shows that this is not the case. 
\begin{restatable}{thm}{thmfirstspy} \label{thm: dregular result}
The maximum expected precision on a random 4-regular graph with graph unknown to the adversary is bounded by
%\begin{align}
$
\mathbf{D}_\mathtt{OPT} \leq 8 \mathbf{D}_\mathtt{FS} + 6p^2 + O(p^3),
$
%\end{align}
where $\mathbf{D}_\mathtt{OPT}$ and $\mathbf{D}_\mathtt{FS}$ denote the expected precision under the optimal and first-spy estimators respectively. 
\end{restatable} 
\emph{(Proof in Appendix~\ref{apx: dandelionpp})} The proof of this result relies on bounding the amount of information that a single node can pass to the adversary, and enumerating the local graph topologies that a single node can see. 
\begin{figure*}[t]
    \centering
    \begin{minipage}{.25\textwidth}
    \centering
  \includegraphics[width=\linewidth]{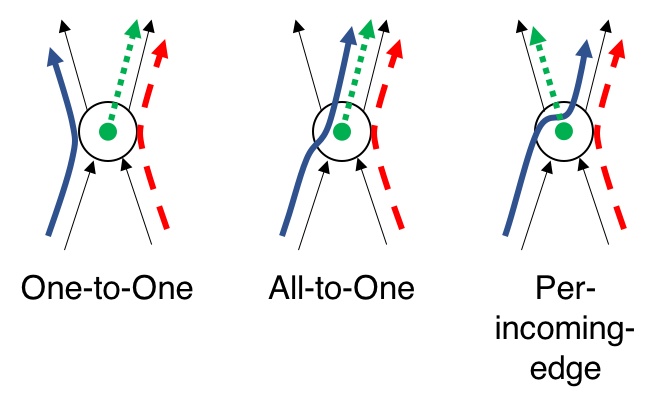}
  \caption{Pseudorandom forwarding. 
  Thick, colored lines denote forwarding rules for incoming edges and the node's own transactions (green dotted line).}
  \label{fig:forwarding}
    \end{minipage}
    \hspace{0.05in}
    \begin{minipage}{.33\textwidth}
    \centering
  \includegraphics[width=\linewidth]{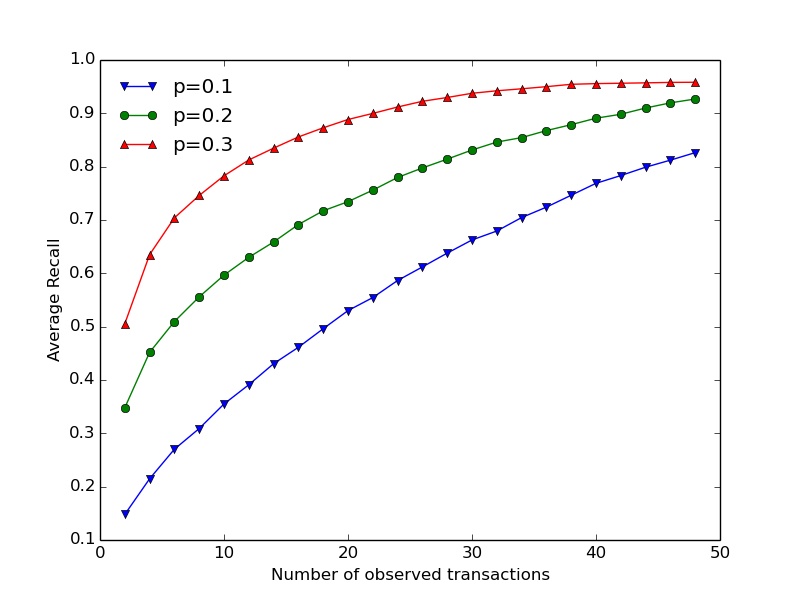}
  \caption{Recall vs. number of transactions per node in random
   $4$-regular graphs.} \label{fig:int_attack}
  \end{minipage} 
     \hspace{0.05in}
    \begin{minipage}{.33\textwidth}
%\end{figure}
%\begin{figure}[t]
    \centering
 \includegraphics[width=\linewidth]{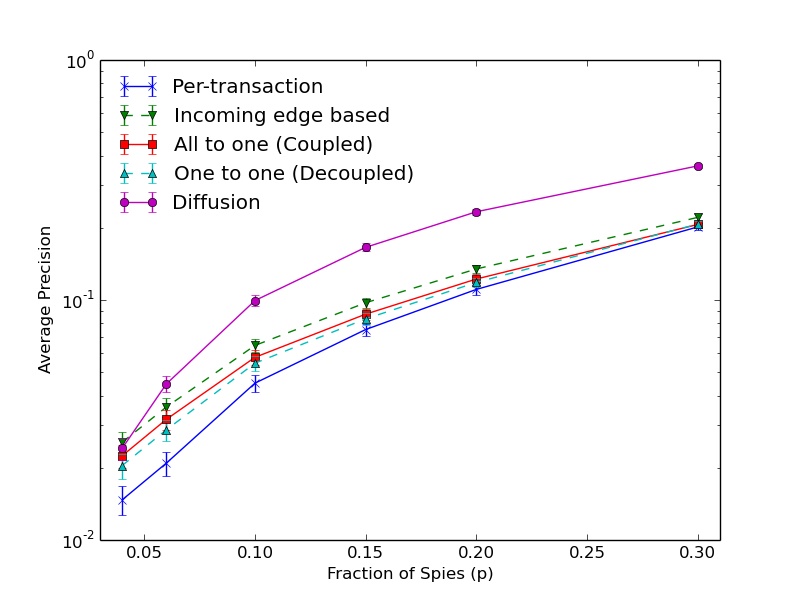}
  \caption{First-spy precision for 4-regular graphs under various forwarding schemes.}
  \label{fig:pseudorandom}
  \end{minipage}
\end{figure*}
This result says that precision under the first-spy estimator is within a constant factor of optimal ($p^2$ is a fundamental lower bound on the maximum expected precision of any scheme \cite{dand}). 
%\red{Since the first-spy estimator is close to the performance of line graphs, we conclude that the optimal estimator must also be close to line graphs. <- what does this mean?} 
Thus the precision gain from line graphs to 4-regular graphs is indeed small, as Figure \ref{fig:4regular} suggests. 
%The proof proceeds by bounding the expected precision for a given node $v$; to do this, we rely on the Markov property of \Algo~spreading.
%The intuition is that conditioned on: (1) the number of messages that reach $v$ prior to reaching a spy node, and (2) the set of corresponding sources, any mapping of transactions to sources would have been equally likely to generate the observed outputs.
%Hence, no estimator can achieve a precision significantly higher than one that randomly assigns sources to transactions among that set.  
These observations motivate the use of 4-regular graphs, specifically in lieu of higher-degree regular graphs. 
First, 4-regular graphs have similar precision to complete graphs when the graph is known (i.e., the red dotted line is close to the middle black solid line, which is a lower bound on precision for regular graphs), 
but they sacrifice minimal precision when the graph is unknown.
Hence, they provide robustness to graph-learning. %adversaries that are somehow able to learn the graph.
%This allows nodes to refresh their anonymity graph edges at a much slower rate. 
%Second, we are able to theoretically characterize expected precision on 4-regular graphs (Theorem \ref{thm: dregular result}).
%While a similar result may hold for higher-degree regular graphs, the techniques we used to prove Theorem \ref{thm: dregular result} would likely result in large constants, thus limiting the practical utility of such a result.
%Finally, 4-regular graphs are easier to maintain than higher-degree regular graphs, due to their edge sparsity.
%Thus we propose to use 4-regular graphs in \Algopp.

\noindent \textbf{Lesson:} \emph{4-regular anonymity graphs provide more robustness to graph-learning attacks than line graphs. }

Note that constructing an exact 4-regular graph with a fully-distributed protocol is difficult in practice; we will discuss how to construct an \emph{approximate} 4-regular graph in \S\ref{sec: const graph}, and how this approximation affects our anonymity guarantees.

\subsection{Intersection Attacks}
\label{sec:intersection}

The previous section showed that $4$-regular graphs are robust to
deanonymization attacks, even when the adversary knows the graph. 
However, those results assume that each user generates exactly \emph{one} transaction per epoch.
%real users generate multiple transactions.
In this section, we relax the one-transaction-per-node assumption, and allow nodes to generate an arbitrary number of transactions.
\Algo~  specifies that each transaction should take an \emph{independent} path over the anonymity graph. 
In our case, this implies that if a node generates multiple transactions, each one will traverse a random walk (of geometrically-distributed length) over a 4-regular digraph.  
Under such a model, adversaries can aggregate metadata from multiple, linked transactions to launch intersection attacks. 
We first demonstrate \Algo's vulnerability to intersection attacks,  and then
provide an alternative propagation technique.

\noindent \textbf{Attack.} Suppose the adversary knows the graph.
%The adversary can then simulate dandelion spreading originating at each honest node.  
For each honest source $v \in V_H$, each of its transactions will reach one of the $np$ spy nodes first. 
In particular, each spy node has some fixed probability of being the first spy for transactions originating at $v$, given a fixed graph topology.
We let $\Psi_v$ denote the pmf of the first spy for transactions starting at $v$; the support of this distribution is the set of all spies.
We hypothesize that in realistic graphs, for $v\neq w$, $\Psi_v\neq \Psi_w$.

This hypothesis suggests a natural attack, which consists of a training phase and a test phase. 
In the training phase, for each candidate source, the adversary simulates dandelion spreading $N$ times. 
The resulting empirical distribution of first-spies for a given source determines the adversary's estimate of $\Psi_v$.
%As $N \rightarrow \infty$, $$
The adversary computes such a signature for each candidate source.
%Because of the relative sparsity of the graph, each source node will 
%estimate the expected fraction of fresh transactions that will reach each spy
%node during the anonymity phase. This method results in a distribution over the
%spies for each honest node, defined as a user's signature.

At testing, the adversary gets to observe $m$ transactions from a given node, $m \ll N$. 
The adversary again computes the empirical distribution $\hat \Psi$ of first-spies from those $m$ observations. 
%Since $m$ is small, this estimate will be noisy. 
The adversary then classifies $\hat \Psi$ to one of the $|V_H|$ classes (i.e., source nodes) by matching $\hat \Psi$ to the closest $\Psi_i$ from training.
For each trial, $\hat \Psi$ and $\Psi_v$ are matched by maximizing the likelihood of signatures (i.e. by minimizing the KL divergence). 
%In our experiments, we use KL-divergence as a distance metric. 
%collects analogous signatures by aggregating over
%pseudonym the transactions observed at each spy node. An adversary can
%deanonymize users by matching the pre-estimated signatures with the observed
%pseudonym signatures.

Figure~\ref{fig:int_attack} shows the recall for such an attack on a 
$4$-regular graph of size $1000$ with various fractions of spies, as a function of the number of transactions observed per node. 
%The two trend
%lines correspond to signatures trained on $N=15000$ and $N=35000$ simulations.
By
observing  $10$ transactions per node, the recall exceeds $0.8$ when 30\% of nodes are adversarial.
This suggests that independent random forwarding leads to serious intersection attacks. 
Hence, a naive implementation of \Algo~critically damages anonymity properties.

%\begin{figure}[t]
%  \centering
%  \includegraphics[width=.45\textwidth]{figures/intersection_attack_regular_graph}
%  \caption{
%   Average recall as a function of number of transactions per node in random,
%   $4$-regular graphs.
%  }
%  \label{fig:int_attack}
%\end{figure}

\noindent \textbf{Solution.} To address these attacks, we consider  forwarding mechanisms with \emph{correlated} randomness.
The key insight is that messages from the same source should traverse the same path; this prevents adversaries from learning additional information from multiple transactions.
However, a naive implementation (e.g., adding a tag  that identifies transactions from the same source, and sending all such transactions over the same path) makes it trivial to infer that otherwise unlinkable transactions originate from a common source.
Hence, we consider three forwarding schemes that pseudorandomize the forwarding trajectory.
In ``one-to-one" forwarding, each node maps each of its inbound edges to a unique
outbound edge; messages in stem mode only get relayed according to this mapping (Fig. \ref{fig:forwarding}). 
This one-to-one forwarding captures the `cable' behavior described in Figure \ref{fig:twistedpair}.
Each node also chooses exactly one outbound edge for all of its own  transactions.
Here we randomize not by source, but by incoming edge (for relayed transactions).
Similarly,  ``all-to-one" forwarding maps all inbound edges to the same
outbound edge,  
and ``per-incoming-edge" forwarding maps each 
inbound edge to a uniform outbound edge (with replacement).

%\begin{figure}[t]
%    \centering
%  \includegraphics[width=2.3in]{figures/forwarding}
%  \caption{Pseudorandom forwarding mechanisms.}
%  \label{fig:forwarding}
%\end{figure}

Perhaps counterintuitively, these spreading mechanisms alter the anonymity guarantees even when the graph is unknown.
Our next result---the second main theoretical contribution of this paper---suggests that one-to-one forwarding has near-optimal precision when the adversary does not know the graph, 
\emph{even in the face of intersection attacks} (recall the lower bound of $p^2$, where $p$ is the fraction of spies \cite{dand}). 
The other two mechanisms do not.

\begin{restatable}{thm}{intersectionthm} \label{thm:intersection}
Suppose the graph is unknown to the adversary, each node generates an arbitrary number of transactions, and the adversary can link transactions from the same user.\footnote{If the user is using the same pseudonym for different transactions, this linkage is trivial. However, even when users use fresh keys for different transactions, practical attacks have been able to link the pseudonyms \cite{fistful}.} 
The expected precision of the precision-optimal estimator for the
one-to-one ($\mathbf{D}_{\mathtt{OPT-OtO}}$), all-to-one
($\mathbf{D}_{\mathtt{OPT-AtO}}$), and per-incoming-edge
($\mathbf{D}_{\mathtt{OPT-PIE}}$) message forwarding schemes are:
\begin{eqnarray}
\mathbf{D}_{\mathtt{OPT-OtO}} & = & \Theta\Big(p^2\log\Big(\frac{1}{p}\Big)\Big) \\
\mathbf{D}_{\mathtt{OPT-AtO}} & = & \Theta(p) \\
\mathbf{D}_{\mathtt{OPT-PIE}} & = & \Omega(p).
\end{eqnarray}
\end{restatable}
\noindent \emph{(Proof in Appendix \ref{sec:intersection proof})} 
This analysis exploits the tree-like neighborhoods of random 4-regular graphs. 
We first build a branching process that captures each routing mechanism, and then analyze the precision-optimal estimators accordingly.

Figure~\ref{fig:pseudorandom} illustrates simulated first spy precision values
for each of these techniques, as well as diffusion (the status quo). 
Simulations were again run on a 100-node graph with an exact 4-regular topology;  spies and sources were selected uniformly at random. 
`Per-transaction' forwarding denotes the baseline i.i.d. random forwarding. 
This figure is plotted for the special case of one transaction per node; if nodes were to generate arbitrarily many transactions, the pseudorandom lines would stay the same (all-to-one, one-to-one, and per-incoming-edge), 
whereas Figure \ref{fig:int_attack} suggests that the per-transaction curve could increase arbitrarily close to 1. %, depending on traffic patterns. 
The same is true for diffusion, as illustrated in prior work \cite{wang2014rumor}.
In addition, when the graph is unknown and there is only one transaction per node, Figure \ref{fig:pseudorandom} suggests that one-to-one forwarding achieves precision values that are close to the lower bound of per-transaction forwarding.
%While these schemes may result in lower message
%mixing (and thus a higher precision for an adversary), they can make the network
%more robust to intersection attacks.
%An important question is what happens when the graph is known.
The following corollary bounds the jump in precision when the graph is known to the adversary. % under such a scenario.
\begin{prop}
If the adversary knows the graph \emph{and} internal routing decisions, then the precision-optimal estimator for one-to-one forwarding has an expected precision of $O(p)$.
\label{cor:intersection}
\end{prop}
\emph{(Proof in Appendix \ref{app:intersection known})}

An adversary that does not know internal forwarding decisions for all nodes has lower
precision, because it must disambiguate between exponentially many paths for each transaction. 
Despite requiring completely new analysis, the 4-regular graph results for
Theorem \ref{thm:intersection} and Proposition \ref{cor:intersection} are
order-equivalent to the line graph results in \cite{dand}. 
This raises an important question: are 4-regular graphs really better than line graphs?
In an asymptotic sense, no. 
However, in practice, the story is more nuanced, and depends on the adversarial model. 
For adversaries that lack complete knowledge of the graph and each node's routing decisions, we observe constant-order benefits, which become more pronounced once we take into account the realities of approximating 4- and 2- regular graphs. 
This is explored in \S\ref{sec: const graph}.
However, when the adversary has full knowledge of the graph \emph{and} each node's internal, random routing decisions, the combination of 4-regular graphs and one-to-one routing has slightly higher (worse) precision than a line graph. 
We detail this tradeoff in Appendix \ref{app:tradeoffs}.
Another issue to consider is that line graphs do not help the adversary link transactions from the same source; meanwhile, on a 4-regular graph with one-to-one routing, two transactions from the same source in the same epoch will traverse the same path---a fact that may help the adversary link transactions.

Hence we recommend making the design decision between 4-regular graphs and line graphs based on the priorities of the system builders. 
If linkability of transactions is a first-order concern, then line graphs may be a better choice. 
Otherwise, we find that 4-regular graphs can give constant-order privacy benefits against adversaries with knowledge of the graph.
Overall, both choices provide significantly better privacy guarantees than the current diffusion mechanism. 
We also want to highlight that the remaining sections of this paper are agnostic to the choice of graph topology; the lessons apply equally whether one uses a 4-regular graph or a line graph.
%We will see in  that these benefits grow even larger once 
%\red{
%Nonetheless, we maintain that 4-regular graphs are preferable to line graphs because (1) the
%precision bound constants are smaller (as seen in \S\ref{sec: topology}), and
%(2) 4-regular graphs are more difficult to learn, since more information is needed to
%specify a 4-regular graph. Quantifying this fact is an interesting question for
%future work. 
%}

%\begin{figure}[t]
%  \centering
%  \includegraphics[width=.45\textwidth]{figures/approx_4_regular}
%  \caption{
%   Average first spy precision as a function of $p$ for random, approximate
%   $d$-regular graphs.
%  }
%  \label{fig:pseudorandom}
%\end{figure}

\vspace{0.05in}
\noindent \textbf{Lessons.} \emph{
%The difference between 4-regular and 2-regular anonymity graphs is smaller than it seemed in \S\ref{sec: topology}; 
If running Dandelion spreading over a 4-regular graph, use pseudorandom, one-to-one forwarding. If linkability of transactions is a first-order concern, then line graphs may be a more suitable choice.}

\begin{figure*}[t]
%    \centering
%    \begin{minipage}{.32\textwidth}
%    \centering
%    \includegraphics[width=\linewidth]{figures/degree_checking}
%    \caption{Average \Algo~precision when Byzantine nodes lie about their degrees in the \textsf{ApxLine($\ell$)} graph-construction algorithm.}
%    \label{fig:k_choices}
%    \end{minipage}
%    ~
    \begin{minipage}{.32\textwidth}
%\end{figure}
%\begin{figure}[t]
    \centering
  \includegraphics[width=\linewidth]{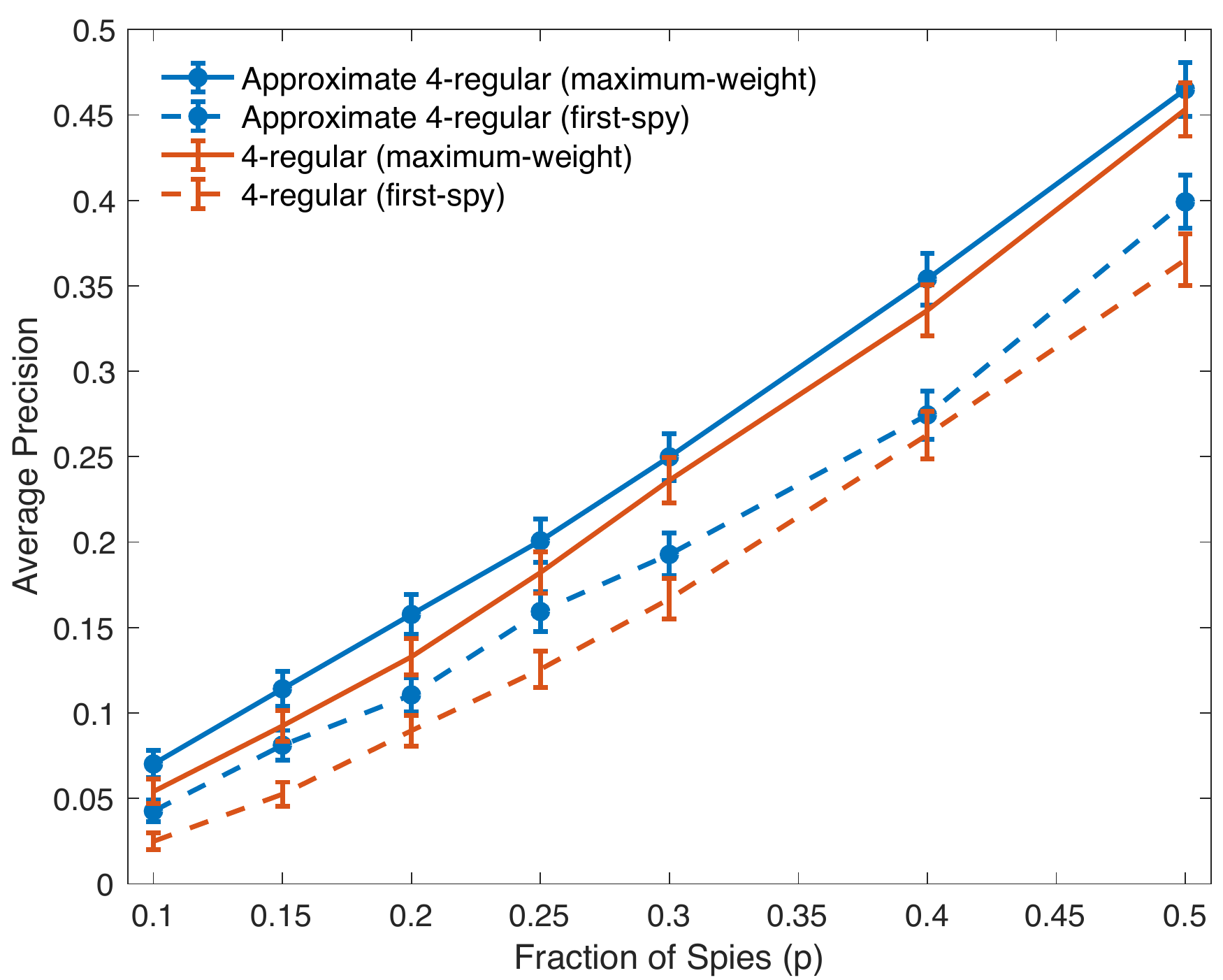}
  \caption{Average precision for approximate 4-regular graphs compared to exact 4-regular graphs.}
  \label{fig:approx_vs_reg}
  \end{minipage}
      \hspace{0.02in}
      \begin{minipage}{.32\textwidth}
    \centering
  \includegraphics[width=\linewidth]{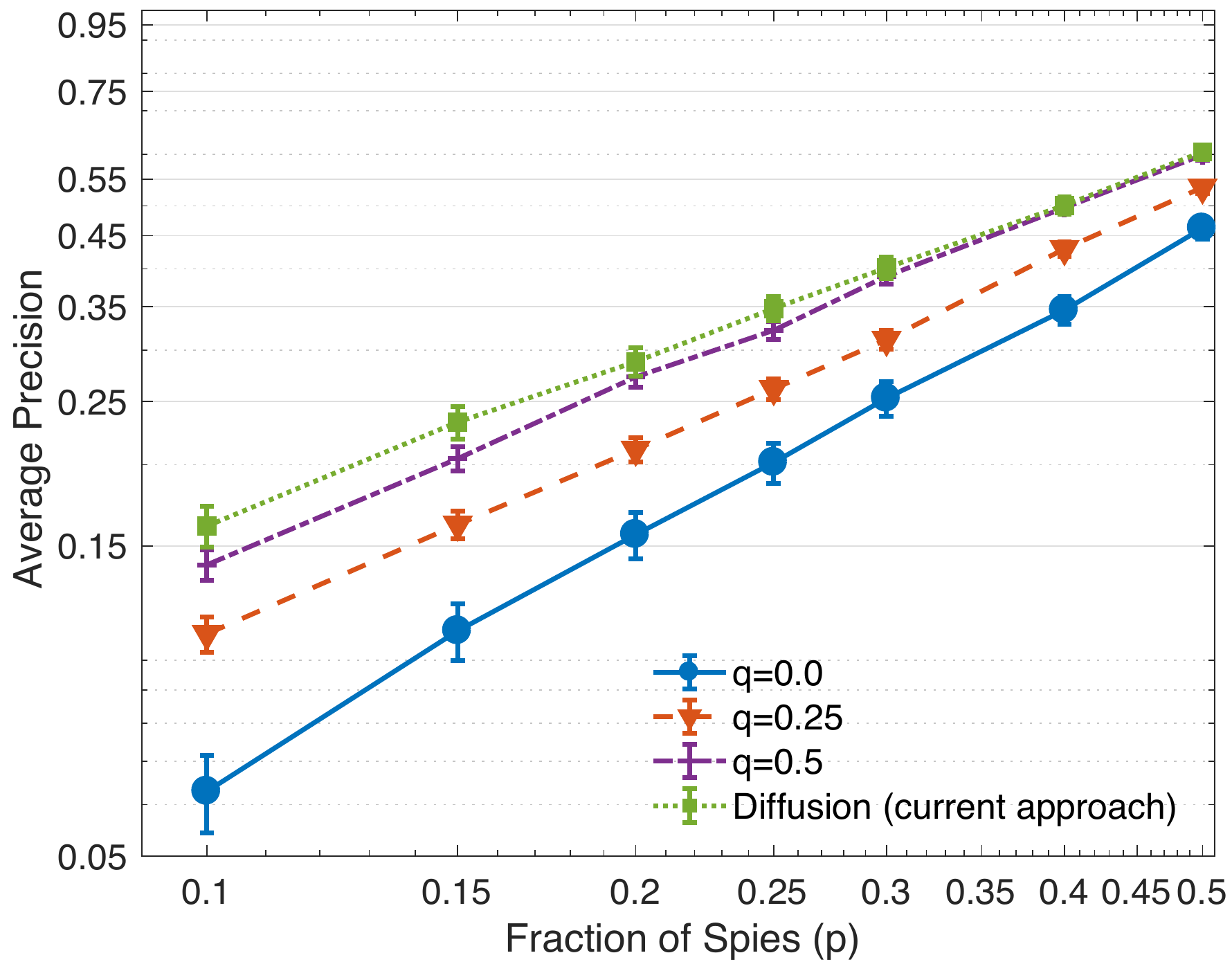}
  \caption{Honest-but-curious spies obey graph construction protocol.}
  \label{fig:behave}
    \end{minipage}
    \hspace{0.02in}
    \begin{minipage}{.32\textwidth}
%\end{figure}
%\begin{figure}[t]
    \centering
  \includegraphics[width=\linewidth]{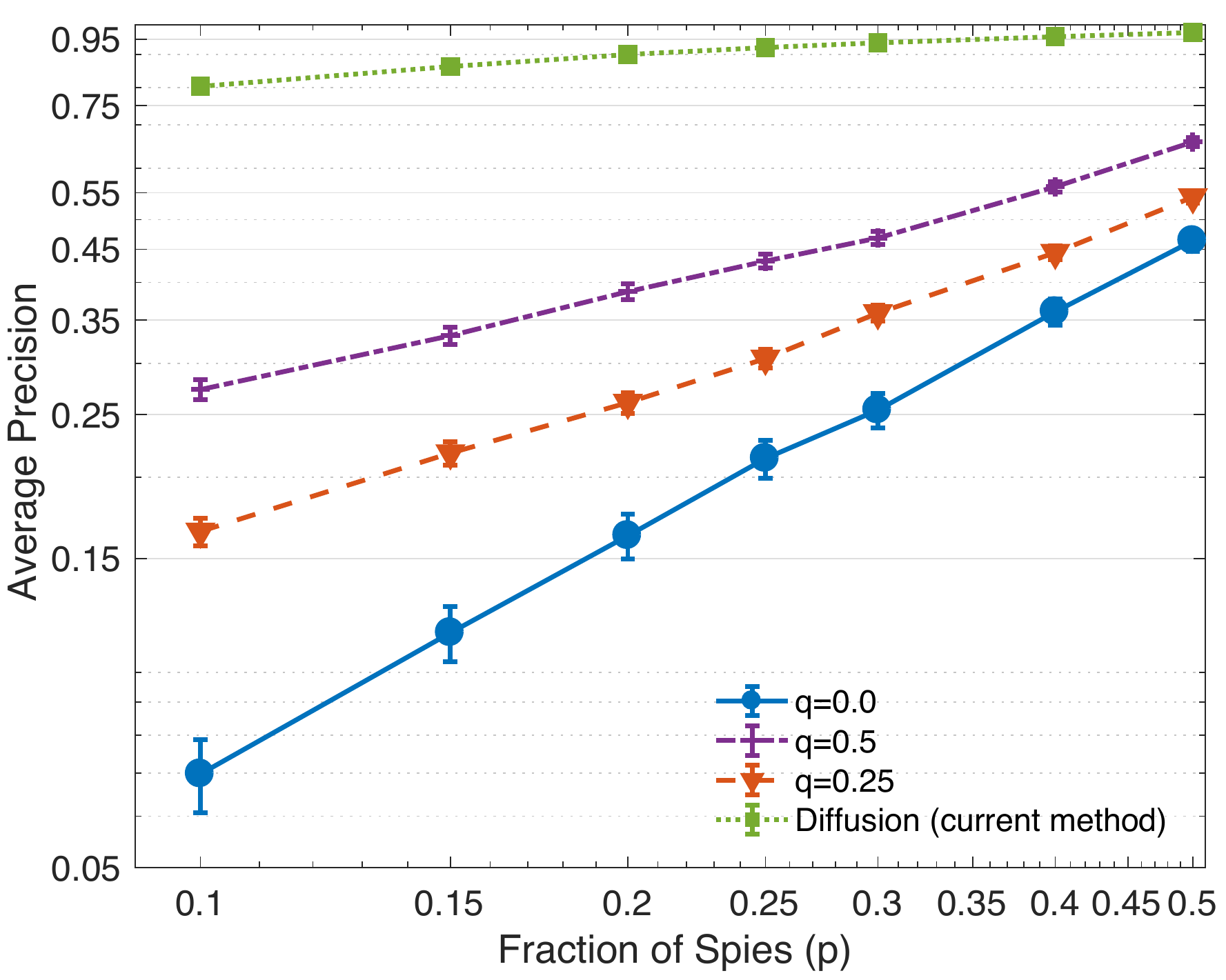}
  \caption{Malicious spies make outbound edges to every honest node.}
  \label{fig:misbehave}
  \end{minipage}
\end{figure*}

\subsection{Graph-Construction Attacks} \label{sec: const graph}
We have so far considered graph-learning attacks and intersection attacks. 
Another important aspect of \Algopp~is the anonymity graph construction, which should be fully distributed. 
\Algo~proposed an interactive, distributed algorithm (explained below) that constructs a randomized \emph{approximate} line graph. % where each node has an expected degree of two. 
%Approximate line graphs suffice for anonymity (compared to exact line graphs), provided the graph does not contain too many leaves. 
%To control the number of leaf nodes,  \Algo~requires each node to contact a random set of peers and retrieve local topology information from them;
%the node uses this topology information to decide which node it should connect to.
In this section, we study how Byzantine nodes can change the graph to boost their accuracy.
First, we show how to generate 4-regular graphs in the presence of Byzantine nodes. %, since \Algopp~does not use lines.
Next, we show how to choose $q$ (path length parameter) for robustness against Byzantine nodes. % manipulating the graph. 
%We describe a simple extension to \Algo's approach that can be used to generate approximately-$4$-regular graphs. 
%The second question, which represents the main contribution of this section, is how to deal with Byzantine nodes. % that disobey protocol to alter the underlying graph structure.
%%performs the requesting operation multiple times and then chooses the best peer to connect to.  
%%While such a policy can reduce the number of leaves in the constructed graph, Byzantine nodes can lie about their local topologies, 
%%thereby exerting undue influence on the graph topology. 
%We first show that Byzantine nodes can significantly degrade the anonymity properties of \Algo~ by lying during the interactive graph-construction protocol.
%Then, we propose a simplified graph-construction policy in \Algopp~ that favors robustness to Byzantine nodes over exactness of topology. 
%The anonymity penalty of this approximate construction is minimal when spy nodes are honest-but-curious, and the anonymity gains are significant when spy nodes are malicious. 
%We start with the graph-construction protocol from \Algo.

%\\  \\
%\noindent {\bf 
\subsubsection{Graph construction in \Algo~\cite{dand}} 
For any integral parameter $\ell>0$, \Algo~uses the protocol \textsf{ApxLine($\ell$)} to build a $\ell$-approximate line graph as follows:
\emph{(1) Each node contacts $\ell$ random candidate neighbors, and asks for their current in-degrees.
(2) The node makes an outbound connection to the candidate with the smallest  in-degree. Ties are broken at random.}
%Each user executes the above two steps.  
This protocol  is simple, distributed, and allows the graph to be periodically rebuilt.
Though the resulting graph need not be an exact line, %(e.g. it could be disconnected), 
nodes have an expected degree of two, and experiments show low precision for adversaries. 
Increasing  $\ell$ also enhances the likeness of the graph to a line and reduces the expected precision in simulation. 
%\\  \\
%\noindent {\bf 
\subsubsection{Construction of $4$-regular graphs.} 
A natural extension of \Algo's graph-construction protocol to 4-regular graphs involves repeating  \textsf{ApxLine($\ell$)}  twice for parameter $\ell>0$.
That is, each peer makes two outgoing edges, where the target of each edge is chosen according to \textsf{ApxLine($\ell$)}.
As in the approximate line algorithm, the resulting graph is not exactly regular. 
However, the expected node degree is 4, and because each node generates two outgoing edges, the resulting graph has no leaves. 
This improves anonymity because leaf nodes are known to degrade average precision \cite{dand}.
%We shall see later that this improves anonymity guarantees compared to approximate lines.
%Next, we explore the anonymity implications of this protocol in the presence of Byzantine nodes.

%\vspace{0.05in}
%\noindent {\bf 
\subsubsection{The impact of Byzantine nodes} 
%The simplicity of the graph-construction protocol makes it straightforward to understand the effects of malicious parties.
Byzantine nodes can misbehave as recipients and/or creators of edges. 
As recipients, nodes can lie about their in-degrees during the degree-checking phase.
As creators of edges, misbehaving nodes can generate many edges, even connecting to each honest node.
%In the extreme, each spy could connect to each honest node.
%We discuss the impact of both issues here.

\vspace{0.05in}
\noindent {\bf Lying about in-degrees.}
%The \textsf{ApxLine(k)} algorithm is vulnerable to adversaries manipulating the protocol. 
In step (1) of the graph-construction protocol, when a user queries an adversarial neighbor for its current in-degree, the adversary might deliberately report a lower degree than its actual degree. 
This can cause the querying user to falsely underestimate the neighbor's in-degree and make a connection. 
In the extreme case, the adversarial node can consistently report an in-degree of zero, thus attracting many incoming edges from honest nodes. 
This degrades anonymity by increasing the likelihood of honest nodes passing their transactions directly to the adversary in the first hop. % allowing adversaries to learn which node created which . 
We find experimentally that such attacks significantly increase precision as $\ell$ grows; 
plots are omitted due to space constraints. %included in Appendix \ref{app:graph_construction}.

To avoid nodes lying about their in-degrees, we abandon the interactive aspect of \Algo~graph construction. % and use $k=1$ throughout.
%However, as discussed previously, adversaries can abuse the protocol to gain more incoming edges.
%To avoid this problem, we first note that such an abuse of \textsf{ApxLine(k)} occurs only at high values of the parameter $k$, wherein the lying adversarial nodes have a greater likelihood of being selected by honest nodes.
In \textsf{ApxLine(1)}, users select a random peer and make an edge regardless of the recipient's in-degree.
We therefore  run \textsf{ApxLine(1)} twice, as shown in Algorithm~\ref{algo:4reg_approx}.
Note that Algorithm \ref{algo:4reg_approx}  closely mirrors the graph construction protocol used in Bitcoin's P2P network today, so the protocol itself is not novel.
In the line graphs of \Algo, a higher $\ell$ value was needed since  \textsf{ApxLine($\ell$)} for $\ell = 1$ was shown to have significantly worse anonymity performance than $\ell\geq 2$.
However such a loss is avoided in \Algopp~since the application of \textsf{ApxLine(1)} twice eliminates leaves. %, thereby improving anonymity.

%Eliminating the degree-checking step reduces the complexity of graph construction and removes the opportunity for Byzantine nodes to impact the graph structure by lying. 
%However, it is unclear how much anonymity we lose by using a graph that is not exactly $d$-regular. 
%This question was explored for line graphs in \cite{dand}, but we are advocating the use of 4-regular graphs.
What is \emph{not} previously known is how the approximate-regular construction in Algorithm \ref{algo:4reg_approx} affects anonymity compared to an exact 4-regular topology.
%We wish to understand how Algorithm \ref{algo:4reg_approx} impacts the anonymity properties of \Algopp.
%In particular, we want to understand the effect of using an approximately-4-regular topology compared to an exact 4-regular topology.
First, note that the  expected recall does not change
because dandelion spreading (Algorithm \ref{algo:dandelion}) has an optimally-low maximum recall of $p+O(\frac{1}{n})$, regardless of the underlying graph (Thm. 4, \cite{dand}).
Hence, we wish to understand the effect of approximate regularity on maximum expected precision.
We simulated dandelion spreading on approximate 4-regular graphs and exact 4-regular graphs,
using the same approximate precision-optimal estimator from Figure \ref{fig:4regular}.
For comparison, we have also included the first-spy estimator. %, which is consistently lower than the maximum-weight matching estimator.
Figure \ref{fig:approx_vs_reg} shows that the % resulting average precision curves as a function of the fraction of spies.
 difference in precision between 4-regular graphs and approximate 4-regular graphs (computed with Algorithm \ref{algo:4reg_approx}) is less than 0.02 across a wide range of spy fractions $p$.
Compared to line graphs \cite[Figure 8]{dand}, 4-regular graphs appear significantly more robust to irregularities in the graph construction.
%Hence we lose minimal anonymity by using Algorithm \ref{algo:4reg_approx} instead of constructing an exact 4-regular graph.
%In return, we gain a fully-distributed graph-maintenance algorithm that is robust to nodes lying during the graph construction protocol.

%\begin{figure*}[t]
%    \centering
%    \begin{minipage}{.32\textwidth}
%    \centering
%  \includegraphics[width=\linewidth]{figures/spies_outbound_edges_behaving_new}
%  \caption{Honest-but-curious spies obey graph construction protocol.}
%  \label{fig:sub1}
%    \end{minipage}
%    ~
%    \begin{minipage}{.32\textwidth}
%%\end{figure}
%%\begin{figure}[t]
%    \centering
%  \includegraphics[width=\linewidth]{figures/spies_outbound_edges_misbehaving_new}
%  \caption{Malicious spies make outbound edges to every honest node.}
%  \label{fig:behave}
%  \end{minipage}
%      ~
%    \begin{minipage}{.32\textwidth}
%%\end{figure}
%%\begin{figure}[t]
%    \centering
% \includegraphics[width=\linewidth]{figures/approx_4_regular}
%  \caption{  Average first spy precision as a function of $p$ for random, approximate
%   $d$-regular graphs.}
%  \label{fig:pseudorandom}
%  \end{minipage}
%\end{figure*}

\begin{algorithm}[t]
\DontPrintSemicolon
\KwIn{Set $V=\{v_1, v_2, \ldots, v_n\}$ of nodes; }
\KwOut{A connected, directed anonymity graph $H(V,E)$ with average degree $4$}
\For{$v \gets V$} {
%  \For{$c \gets 1$ \textbf{to} $d$} {
    \tcc{pick two random targets} 
    $u_1 \sim$ Unif$(V \setminus \{v\})$ \;
    $u_2 \sim$ Unif$(V \setminus \{v,u_1\})$ \;
    \tcc{make connections}
    $E = E \cup (v \rightarrow u_1) \cup (v \rightarrow u_2)$ 
%  }
  
}
\Return{$H(V,E)$}\;
\caption{{\sc Approximate 4-Regular Graph} Approximates a directed 4-regular graph in a fully-distributed fashion.}
\label{algo:4reg_approx}
\end{algorithm}

%\vspace{0.05in}
\noindent {\bf Creating many edges.}
\Algo~is naturally robust to nodes that create a disproportionate number of edges, because
%The reason is that 
spies can only create \emph{outbound} edges to honest nodes. %; they cannot force honest nodes to establish outbound edges to spies.
This matters because in the stem phase,  honest nodes only forward messages on outbound edges. 
%The following proposition makes this intuition precise.
%Moreover, once the message reaches a spy in the stem, subsequent observations are conditionally independent given the first spy's observation, 
%so spy nodes gain nothing by altering the propagation pattern once a transaction has reached them.
%This implies that the likelihood of a given candidate stem trajectory from a candidate source to a spy node is independent of the in-degrees of the nodes in the trajectory.
%So although we cannot prevent spy nodes from generating arbitrarily many outbound edges, there is no deanonymization incentive to do so.

%\red{should we include a small proposition here to this effect?}
\begin{prop} \label{thm:outbound}
Consider dandelion spreading (Algorithm \ref{algo:dandelion}) with $q=0$ over a connected anonymity graph $H$ constructed according to graph-construction policy $\mathcal P$ (Algorithm \ref{algo:dreg_approx}).\footnote{$q$ denotes the probability of transitioning to fluff phase in each hop.}
%Let the probability of transitioning from stem to fluff phase $q=0$.
Let $\mathbf{D}_\mathtt{OPT}(\mathcal P)$ and $\mathbf{R}_\mathtt{OPT}(\mathcal P)$ denote the maximum expected precision and recall over graphs constructed according to $\mathcal P$. 
Now consider an alternative policy $\mathcal Q$ that is identical to $\mathcal P$ except adversarial nodes are allowed to choose their outbound edges arbitrarily. 
Let $\mathbf{D}_\mathtt{OPT}(\mathcal Q)$ and $\mathbf{R}_\mathtt{OPT}(\mathcal Q)$ denote the maximum expected precision and recall over all graphs constructed according to $\mathcal Q$.
Then 
\begin{eqnarray}
\mathbf{R}_\mathtt{OPT}(\mathcal Q) &=& \mathbf{R}_\mathtt{OPT}(\mathcal P) \notag \\
\mathbf{D}_\mathtt{OPT}(\mathcal Q) &=& \mathbf{D}_\mathtt{OPT}(\mathcal P).
\end{eqnarray}
\end{prop}
%(Proof in Appendix \ref{app:outbound})
\emph{(Proof in Appendix \ref{app:outbound}).}~~
%These bounds are loose because we assume that 
%Intuitively, the result holds because  \Algopp~spreading is independent of the inbound connections to each stem node. %, the expected precision and recall guarantees are unchanged 
%We formalize this observation in Proposition \ref{}, later in this section.
This result bounds the deanonymization abilities of Byzantine nodes in general and supernodes in particular \cite{biryukov,koshy2014analysis}, neither of which was covered by the analysis of \cite{dand}.
%with respect to an adversary that creates the protocol-specified number of outbound edges. 
It shows that for the special case where the transition probability from stem to fluff phase $q=0$ (i.e., infinite stem phase), supernodes gain no deanonymization power by connecting to most or all of the honest nodes. 
In practice, we need $q>0$ to reduce broadcast latency, but
analyzing this requires an upper bound on the probability of detecting the source of a diffusion process under sampled timestamp observations---a known open problem  \cite{PTV12,ZY13}.
%Since theoretical analysis appears challenging in this case, 
We therefore simulate precision for nonzero $q$ as a function of spy fraction $p$, when spies obey protocol (Figure \ref{fig:behave}) and when they form outbound edges to all honest nodes (Figure \ref{fig:misbehave}).
We generate a P2P graph via Algorithm \ref{algo:dreg_approx} with out-degree $\eta=8$, and an anonymity graph $H$ via Algorithm \ref{algo:4reg_approx},  except spies form outbound edges to \emph{all} honest nodes. 
%This worst-case scenario bounds the adversary's capabilities.

%Each honest node creates a transaction, which propagates according to \Algopp. 
%For each transaction, we first check if it reached a spy during the stem phase. 
%If so, the preceding honest node in the stem is treated as the stem's terminus. 
%If not, the adversarial spies must infer the stem's terminus using a diffusion source estimator;
%in this simulation, we use the first-spy estimator.
%Now, given the estimated stem terminus for each transaction, we compute the approximate likelihoods of each node being the source (using full knowledge of the anonymity graph), and compute a max-weight matching, just as in prior simulations.
%Note that this approach does not strictly upper bound the maximum expected precision because we use (a) a suboptimal diffusion source estimator to identify the stem's terminus, and (b) approximate likelihoods in our max-weight matching.  
%Nonetheless, we expect the first-spy estimator to be near-optimal in this experiment since spy nodes connect to every honest node. 

Figures \ref{fig:behave} and \ref{fig:misbehave} highlights two points: 
First, even when spies follow  protocol, increasing $q$ increases precision. 
$q$ has the largest effect when $p$ is small; when $p=0.1$, using $q=0.5$ increases the expected precision by about 0.1  compared to  $q=0.0$.
For $q\leq 0.2$, we expect increases in precision on the order of 0.05.  
Second, spies can increase their precision by adding outbound edges; when $p=0.1$ and $q=0.5$, we observe a precision increase of about 0.2. 
%This data can inform the choice of parameter $q$. 
Thus by choosing parameter $q \leq 0.2$, we can limit the increase in average precision to 0.1, even when spies connect to every honest node. 
%We revisit this in \S\ref{sec:evaluation}, where we make recommendations for $q$ based on latency measurements.

\noindent \textbf{Lesson.} \emph{To defend against graph-manipulation attacks,  use noninteractive protocols and small $q$.}

\subsection{Black-Hole Attacks} \label{sec: msg frwd}

%Adversaries can also misbehave during message forwarding in the stem phase.
Since \algo~spreading forwards messages to exactly one neighbor in each hop, propagation can terminate entirely if an adversarial relay in the stem chooses not to forward a message; we refer to this as a \emph{black-hole attack}. 
To prevent black-hole attacks, \Algopp~ sets a random expiration timer at each stem relay immediately  upon receiving transaction messages. 
If the relay does not receive an \texttt{INV} (i.e. an advertisement) for the transaction before his timer expires, then the relay diffuses the transaction. 
This policy provides a two-fold advantage over \Algo: (1) messages are guaranteed to eventually propagate through the network, and (2) the random timers can help anonymize peers initiating the spreading phase in the event of a black-hole attack.

%\begin{figure}[t]
%    \centering
%  \includegraphics[width=.2\textwidth]{figures/dandelion_stall}
%  \caption{Adversaries can prevent messages from propagating in \Algo.}
%  \label{fig:dandelion_stall}
%\end{figure}

%\vspace{0.1in}
%\noindent {\bf \Algopp~ propagation protocol.}
%These timers can be set without leaking any information about the source as follows: 
To implement this, \Algopp~nodes are initialized with a timeout parameter $T_\mathrm{base}$. 
%For each new transaction, peers first compute a timeout parameter $T_\mathrm{base}$ that is appended to the transaction before forwarding. 
In the stem phase, when a relay $v$ receives a transaction, it sets an expiration time $T_\mathrm{out}(v)$: 
\begin{align}
T_\mathrm{out}(v) \sim \text{\texttt{current\_time}} + \exp(1/T_\mathrm{base}),
\end{align}
i.i.d. across relays. 
If the transaction is not received again by relay $v$ before $T_\mathrm{out}(v)$, $v$ broadcasts the message using diffusion.
Pseudocode is shown in Algorithm~\ref{algo:dandelionpp} (Appendix \ref{app:algo}).
%\begin{algorithm}[t]
%\DontPrintSemicolon
%\KwIn{Message and timeout parameter $(X, T_\mathrm{base})$ received by $v$ in the anonymity phase, out-neighbors $\mathcal N_{out}(G,v)$ on anonymity graph $G$, spreading graph $H$, parameter $q\in (0,1)$}
%%\KwOut{A connected, directed graph $G(V,E)$ with average degree $d$}
%$T_\mathrm{out}(v) \sim \exp(1/T_\mathrm{base})$ \tcp*{set timer} 
%forward $(X,T_\mathrm{base})$ according to dandelion\; 
%\tcc{wait until message re-received}
%\While{current\_time $\leq T_\mathrm{out}$} {
%\If{$X$ received} {
%timer $\leftarrow$ inactive\;
%break\;
%}
%continue\;
%}
%\tcc{start diffusion}
%\If {timer is active} 
%{ 
%{\sc Diffusion}$(X, v, H)$ \;
%}
%\caption{{\sc Dandelion++ Spreading} at node $v$. The protocol guarantees eventual network-wide propagation of transactions.}
%\label{algo:dandelionpp}
%\end{algorithm}

Algorithm~\ref{algo:dandelionpp} solves the problem of message stalling, as relays independently broadcast if they have not received the message within a certain time.
However, the protocol also ensures that the first relay node to broadcast is approximately uniformly selected among all relays that have received the message. 
%For example, in Figure~\ref{fig:dandelion_stall} if node $v_4$ stalls a message that is propagating along the path $v_1,v_2,v_3,v_4$, then each of the peers that have previously forwarded the message---$v_1,v_2,v_3$---are approximately equally likely to expire their timer and start broadcasting. 
This is due to the memorylessness of the exponential clocks: conditioned on a given node blocking the message, each of the remaining clocks %at $v_1,v_2$ and $v_3$ 
can be reset assuming propagation latency in the stem is negligible.
%However, care has to be taken to ensure that the timers do not prematurely trigger diffusion.
Ideally, the exponential clocks should be slow enough that they only trigger (with high probability) during a black-hole attack. %to allow the message to propagate a sufficient number of hops before diffusion starts; this is necessary for anonymity. 
On the other hand, they must also be fast enough to keep propagation latency low. 
This trade-off is analyzed in the following proposition. 
\begin{prop} \label{prop: clock rate}
For a timeout parameter $T_\mathrm{base} \geq \frac{-k(k-1)\delta_\mathrm{hop}}{2\log(1-\epsilon)}$, where $k, \epsilon$ are parameters and $\delta_\mathrm{hop}$ is the time between each hop (e.g., network and/or internal node latency), transactions travel for $k$ hops without any peer initiating diffusion with a probability of at least $1-\epsilon$. 
\end{prop}
%(Proof in Appendix \ref{app: clock rate})
\emph{(Proof in Appendix \ref{app: clock rate})}~~
Let $\Delta_1 \triangleq k\delta_\mathrm{hop}$ be the time taken for the message to traverse $k$ hops. 
Conditioned on reaching $k$ hops and stalling, let $\Delta_2$ denote the additional time taken for the message to start diffusion. 	
Then, $\Delta_2$ is the minimum of exponential random variables each of rate $1/T_\mathrm{base}$.  
As such $\Delta_2$ is itself exponentially distributed with rate $k/T_\mathrm{base}$. 
Choosing $T_\mathrm{base}$ as in Proposition~\ref{prop: clock rate}, the mean additional time taken to diffuse the message is
\begin{align}
\E[\Delta_2] = \frac{T_\mathrm{base}}{k} = \frac{-(k-1)\delta_\mathrm{hop}}{2\log(1-\epsilon)} \leq  \frac{-\Delta_1}{2\log(1-\epsilon)} \approx \frac{\Delta_1}{2\epsilon}. 
\end{align}
The standard deviation of $\Delta_2$ is identical to the mean. 
Thus by choosing $T_\mathrm{base}$ as in Proposition~\ref{prop: clock rate} in Algorithm~\ref{algo:dandelionpp} we incur an additional delay at most a constant factor $1/(2\epsilon)$ from our delay $\Delta_1$ otherwise.  

\noindent \textbf{Lesson.} \emph{Use random timers selected according to Prop. \ref{prop: clock rate}.}

%\vspace{0.1in}
%\noindent {\em Note:} 
%While we have chosen a timer-based protocol in \Algopp~for practical reasons, it is not clear if this is an optimal solution from an anonymity perspective. 
%For example, an alternative protocol in which source peers simply restart dandelion spreading if they have not heard back their message reveals less information to the adversary and hence provides better anonymity. 
%However this protocol can significantly increase latency;
%for small $q$ (which provides better anonymity in the honest-but-curious case), the probability of encountering a spy in the stem is high.
%If spies are not forwarding \Algopp~transactions, the source may need to wait a prohibitively long time before one of its stems successfully reaches the spreading phase.
%%A systematic study of this problem to understand all the latency-anonymity tradeoffs is an important direction for future work. 

\subsection{Partial-Deployment Attacks}
\label{sec:partial}

%The analysis of \Algo~ in \cite{dand} does not consider two practical aspects of integration and deployment:
%(1) maintaining a separate anonymity graph from the main P2P network may be undesirable; 
The original \Algo~analysis does not consider the fact that instantaneous, full-network deployment of \Algopp~is practically infeasible. 
In this section, we show that if implemented naively, Byzantine nodes can exploit partial \Algo~deployment to launch serious anonymity attacks. 
We also demonstrate a (counterintuitive) implementation mechanism that neutralizes this threat.
%The first issue arises because added complexity can be difficult to maintain; 
%poor implementation choices may introduce new attack vectors that are difficult to foresee and analyze.
%Therefore, simplicity is a desirable property. 
%For this reason, we suggest that the anonymity graph should be a subgraph of the main P2P graph, rather than a distinct graph. 
%%That is, the anonymity graph should be chosen as a subgraph of the main P2P graph.
%While this may increase the adversary's ability to learn the anonymity graph, we make three observations:
%(1) the implementation complexity of \Algopp~declines significantly if the anonymity graph is chosen as a subset of the main graph;
%(2) by using a $4$-regular topology for the anonymity graph, we limit the adversary's gain in precision that comes from knowing the graph (\S \ref{sec: topology}); 
%(3) it may still be difficult for the adversary to learn the anonymity graph, since each node would only use a subset of its edges for \Algopp.
%Moreover, nodes can still cycle their connections periodically to further complicate anonymity graph estimation.
%Our discussions with Bitcoin Core developers suggest that they are more likely to accept solutions in which the anonymity graph is an overlay of the existing P2P network.
%Assuming the anonymity graph is an overlay of the main P2P network, the key question is how to choose it.
We consider two natural approaches for constructing the anonymity graph. 
%:\emph{version-checking} and \emph{no-version-checking}.

\begin{algorithm}[t]
\DontPrintSemicolon
\KwIn{Main (directed) P2P graph $G(V,E)$, desired degree $d$ of output anonymity graph}
\KwOut{Directed anonymity graph $H(V,\tilde E)$}
\For{$v \in V$} {
    \tcc{Find TwistedPair neighbors} 
    $\mathcal D_v \gets \{w\in \mathcal N_{out}(G,v) ~|~ w\text{ supports \Algopp}\}$\;
    \If{$0 \leq |\mathcal D_v| < \frac{d}{2}$} {
      $\tilde E \gets \tilde E \cup \mathcal D_v$
    }
    \tcc{Non-TwistedPair neighbors}  
    \If{$|\mathcal D_v| == 0$}{
    	$\mathcal R \gets \frac{d}{2}$ nodes drawn uniformly from $\mathcal N_{out}(G,v)$, without replacement \;
	$\tilde E \gets \tilde E \cup \mathcal R$
    }
}
\caption{{\sc Version-checking}. $\mathcal N_{out}(G,v)$ denotes the out-neighbors of node $v$ on digraph $G$.}
\label{algo:version-checking}
\end{algorithm}

\noindent \textbf{Version-checking} (Algorithm \ref{algo:version-checking}) generates an anonymity graph from edges between \Algopp-compatible nodes.
Each node $v$ first identifies its outbound peers on the main P2P network that support \Algopp; we call this set $\mathcal D_v$.
$\mathcal D_v$ can be learned from existing signaling in Bitcoin's version handshake. 
Next, %to approximate a $4$-regular anonymity graph, 
$v$ runs Algorithm \ref{algo:4reg_approx}, drawing candidate neighbors only from $\mathcal D_v$.
If $|\mathcal D_v| =1$, $v$ uses the single node in $\mathcal D_v$ as its outbound anonymity graph edge.
If $|\mathcal D_v|=0$, $v$ picks $2$ outbound neighbors uniformly at random, and uses them as anonymity graph edges. 
% the remaining outbound edges uniformly from the set of nodes that do not support \Algopp.
Upon forwarding a transaction to a node that does not support \Algopp, the receiving node will, by default, relay the message using diffusion, thereby ending the stem phase.
While version checking is a natural strategy, adversarial nodes can lie about their version number and/or run nodes that support \Algopp. 
%This may change the anonymity guarantees offered by \Algopp.
%Below we discuss the corresponding anonymity ramifications on \Algopp.

Under the second approach, \textbf{no-version-checking}, each node $v$ instead selects $2$ outgoing edges uniformly from the set of \emph{all} outgoing edges, without considering \Algopp-compatibility.
This noninteractive protocol shortens the expected length of the stem, thereby potentially weakening anonymity guarantees.

If all nodes supported \Algopp, these two approaches would be identical. 
%differences arise only because of gradual rollout.
%---i.e., the observation that we cannot feasibly force all Bitcoin nodes to adopt \Algopp~ at the same time. 
%Understanding the anonymity properties of these two approaches is closely related to the second deployment issue: dealing with gradual adoption.
%Hence, any realistic implementation must be interoperable with nodes that do not support it.
%We wish to understand how such a gradual rollout would affect the anonymity guarantees of \Algopp.
To model gradual deployment, we assume that all spy nodes run \Algopp, 
and a fraction $\beta $ of the remaining honest nodes are using \Algopp.
Honest users are distributed uniformly over the network.
Let $V_D$ denote the nodes that support \Algopp: $|V_D| = pn + (1-p)\beta n$.
We wish to characterize the maximum expected recall of \Algopp~ as a function of $p$ and $\beta $,
 for version-checking and no-version-checking.
The following theorem bounds this quantity under a recall-optimal estimator. 
%The implication is that when \Algopp~adoption is low, version-checking can lead to an expected recall that is significantly higher than diffusion itself.

\begin{thm}\label{thm:recall_partial}
Consider $n$ nodes in an approximately-$2\eta$-regular graph $G$ generated according to Algorithm \ref{algo:dreg_approx}. 
A fraction $p$ of nodes are spies running \Algopp. 
Among the remaining honest nodes, a fraction $\beta $ support \Algopp. 
Under \textbf{version-checking}, the expected recall across honest, \Algopp-compatible nodes, under a recall-optimal mapping strategy satisfies 
\begin{align}
\frac{p}{f} \left (1- (1-f)^\eta \right ) \leq \mathbf{R}_\mathtt{OPT}  \lesssim & \frac{p}{f} \left (1- (1-f)^\eta \right ) +  (1-f)^\eta   + C (1-\beta)) 
\end{align} 
where $f$ is the fraction of \emph{all} nodes that support \Algopp~
 and $C = \left( 1 - \frac{p}{f}\right)  \left (1 - \left(1-f\right )^\eta \right )\frac{1 - (1-\phi)^{\tilde n}}{\tilde n \phi}$, where $\phi=1-(1-\frac{1}{n-\eta})^\eta$, and $\tilde n=(1-p)n$ is the number of honest nodes in the system.
Under \textbf{no-version-checking}, % we have 
\begin{align}
%p + \frac{(1-p)(1-\beta )}{d} \leq \mathbf{R}_\mathtt{OPT}  &\leq p + \frac{1-p}{d} .
p \leq \mathbf{R}_\mathtt{OPT}  &\lesssim p +  (1-\beta(1-q))(1-p)\frac{1 - (1-\phi)^{\tilde n}}{\tilde n \phi}.
\end{align} 
%where $\zeta=\frac{1 - (1-\phi)^{\tilde n}}{\tilde n \phi} \triangleq \zeta$.
%where $C' = (1-p)\left (\frac{1-(1-\phi)^{\tilde n}}{\tilde n \phi} \right )$.
\end{thm}
%\noindent(Proof in Appendix \ref{app:recall_partial})
\emph{(Proof in Appendix \ref{app:recall_partial})}~~
Here $\lesssim$ denotes approximate inequality; i.e., $A(n) \lesssim B(n)$ implies that there exist constants $n_0>0$ and $C'>0$ such that for all $n>n_0$, $A(n) \leq C' B(n)$.
In this case, such a condition holds for any $C'>1$.
%The approximate inequality in Theorem \ref{thm:recall_partial} is used purely for readability. 
%We have exact bounding expressions that are included in the proof, and plotted in Figure \ref{fig:recall_partial}.

%\begin{figure*}[!htb]
%\minipage{0.32\textwidth}
%  \includegraphics[width=\linewidth]{figures/intersection_attack_regular_graph}
%  \caption{Average recall as a function of number of transactions per node in random,
%   $4$-regular graphs.} \label{fig:int_attack}
%\endminipage\hfill
%\minipage{0.32\textwidth}
%  \includegraphics[width=\linewidth]{figures/approx_4_regular}
%  \caption{  Average first spy precision as a function of $p$ for random, approximate
%   $d$-regular graphs.}
%  \label{fig:pseudorandom}
%\endminipage\hfill
%\minipage{0.32\textwidth}%
%  \includegraphics[width=\linewidth]{figures/recall_partial_deployment}
%\caption{Bounds on maximum expected recall for  $p=q=0.2$ and $n=1000$, on an approximately-16-regular graph.}
%  \label{fig:recall_partial}
%\endminipage
%\end{figure*}

%\begin{figure}[t]
%    \centering
%  \includegraphics[width=.45\textwidth]{figures/recall_partial_deployment}
%  \caption{Bounds on maximum expected recall for  $p=q=0.2$ and $n=1000$, on an approximately-16-regular graph.}
%  \label{fig:recall_partial}
%\end{figure}

\begin{figure*}[t]
    \begin{minipage}{.32\textwidth}
    \centering
    \includegraphics[width=\linewidth]{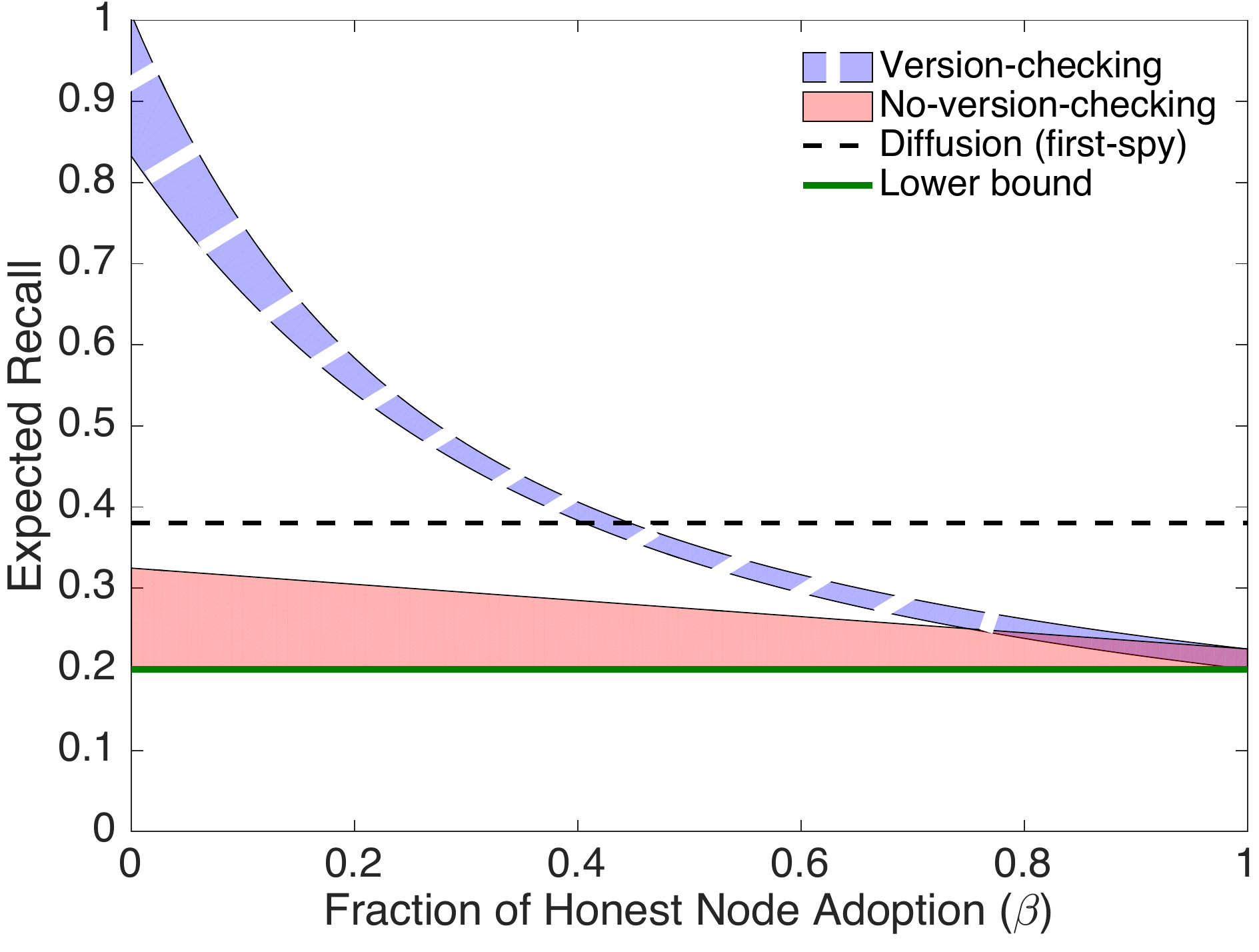}
  \caption{Bounds on maximum expected recall for  $p=q=0.2$ and $n=1000$, on an approximately-16-regular graph.}
  \label{fig:recall_partial}
  \end{minipage}
  \hfill
      \begin{minipage}{.32\textwidth}
    \centering
    \includegraphics[width=\linewidth]{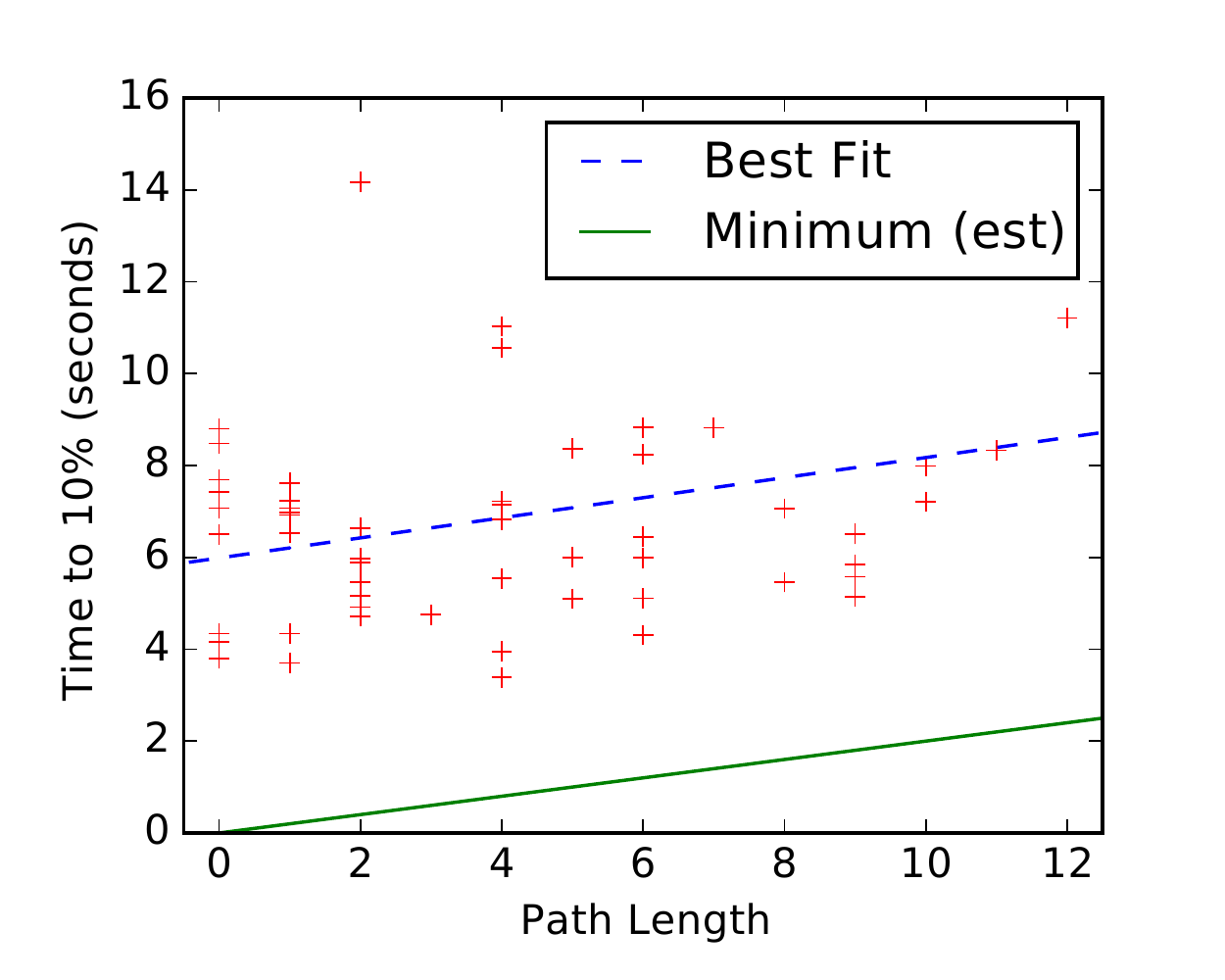}
    \caption{Time to propagate to 10\% of the network as a function of the path length. Blue is the line of best fit at 218ms per hop, and green is the minimum estimated delay at 300ms per hop.
%    \red{Why is the slope of the blue line less than 300 ms per hop?}
    }
    \label{fig:timetoten}
    \end{minipage}
    \hfill
    \begin{minipage}{.32\textwidth}
    \centering
    \includegraphics[width=\linewidth]{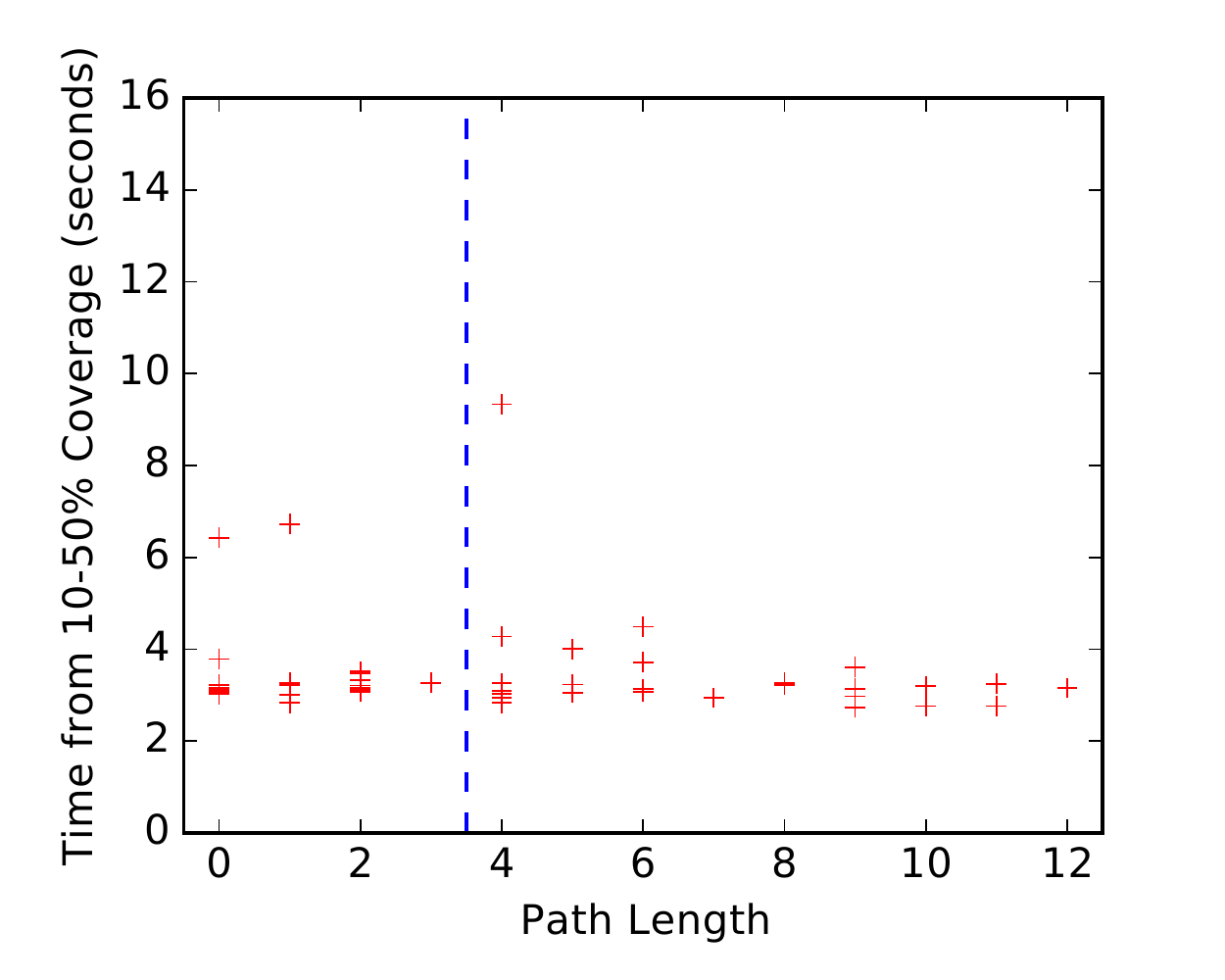}
    \caption{Time  to go from 10\% to 50\% of the network. The blue line indicated where data was split into two categories by hop length with almost equal numbers of samples (low: 26, high: 29).}
    \label{fig:tentofifty}
  \end{minipage}
\end{figure*}

Figure \ref{fig:recall_partial} plots these results for $p=0.2$ fraction of spies and $q=0.2$ probability of ending the \Algopp~stem, on an approximately-16-regular P2P graph. % (Bitcoin nodes establish up to 8 outbound connections).
Our theoretical bounds delimit the shaded regions; %, so the true probability of detection lies somewhere inside the shaded region.
for comparison, we include a lower bound on the recall of diffusion, computed by simulating diffusion with a first-spy estimator. 
The green solid line is a lower bound on the maximum expected recall of any spreading protocol (Thm. 2, \cite{dand}).
When $\beta $ is small, version-checking recall is close to 1; in the same regime, no-version-checking exhibits lower recall than both diffusion and version-checking.
%Hence, no-version-checking may be a better approach, despite creating shorter stems in the anonymity phase.
Intuitively, when adoption is low (low $\beta$), any \Algopp~peers are likely to be spies.
Therefore, version-checking actually \emph{increases} the likelihood of getting deanonymized. %, because it favors \Algopp~nodes.
In the same low-$\beta$ regime, no-version-checking is more likely to choose a \Algopp-incompatible node $w$ as the next stem node.
While this prematurely ends the stem, it still introduces more uncertainty than vanilla diffusion.

\noindent \textbf{Lesson.} \emph{Use no-version-checking to construct the graph. }
%Similar intuition holds for expected precision. 

%Finally, \Algopp~does not increase (and may actually decrease) the expected precision and recall for \emph{non}-\Algopp~nodes.
%%The recall-optimal estimator simply outputs the MAP source estimate for each transaction independently, so recall for non-Dandelion nodes is unaffected by Dandelion nodes.
%%Moreover, non-Dandelion nodes may actually experience \emph{reduced} expected precision due to their Dandelion-enabled peers.
%\cite{dand} shows that the recall-optimal estimator simply outputs the MAP source estimate for each transaction independently; 
%similarly, the precision-optimizing estimator is a maximum-weight matching between transactions and nodes, where the weights are the \emph{a posteriori} probabilities of each node being the source of a particular transaction. %, given the adversary's observations.
%When some nodes are running \Algopp, these posteriors get flattened, in that non-source-nodes have a higher probability of being the source of a given transaction. 
%%The observed metadata for transactions from non-Dandelion nodes is independent of \Algopp~nodes, so the weights associated with such transactions also remain unchanged. 
%%However, transactions generated by Dandelion nodes will achieve some degree of mixing, which causes the weights associated with other (false) source nodes to increase.
%This can cause the resulting matchings to misclassify a Dandelion transaction to a non-Dandelion node (and vice-versa), which can only reduce expected precision and recall. % for non-Dandelion nodes.

\section{Evaluation}
\label{sec:evaluation}

\subsection{Implementation in Bitcoin}
We have developed a prototype implementation of \Algopp.
Our prototype is a modification of Bitcoin Core (referred to as Core from now on), the most commonly-used client implementation in the Bitcoin network. In total, our implementation required a patch modifying approximately 500 lines of code.
A vital part of our implementation is allowing nodes to recognize other \Algopp~ nodes in a straightforward way. In Core, supported features are signaled by modifying the \texttt{nServices} field
in the handshake. For example, segwit support is signalled by setting bit $4$ in \texttt{nServices}. In our case, \Algopp~ support is signaled by setting the $25$th bit.

To minimize our
footprint on the Core codebase, we insert \Algopp~ functionality into the preexisting main threads/signals
that handle the processing and transmission of messages. 
However, creating the $4$-regular anonymity graph and processing \Algopp~ transactions requires careful consideration
of the many concurrency and DoS-protection mechanisms already at play in Core. 
For instance, Core's data structures for transactions and inventory messages are designed to facilitate 
responses to \texttt{GetData} requests, while broadcasting transactions to all nodes with exponential delays. 
%Poisson process. 
%The data structures in Core, on the other hand, are not sufficient for \Algopp~ because they are meant
These data structures are insufficient for \Algopp~ because they facilitate broadcasting knowledge of transaction and block hashes. 
%In particular, our implementation needs to hide all knowledge of transactions until the stem phase begins, while still ensuring proper relaying and DoS protections for propagation. 
In particular, \Algopp~ nodes need to hide knowledge of transactions that are still in the stem phase, but at the same time ensure that they are relayed properly and not stalled by adversaries. 
Our approach is to store stem mode transactions in an additional data structure, the ``embargo map,'' such that embargoed transactions are omitted from \texttt{GetData} responses.
The embargo map serves two purposes: 1. it tracks transactions currently in the stem phase and 2. it ensures that malicious adversaries can not stop transaction propagation (\S \ref{sec: msg frwd}). 

\subsection{Experimental Setup}
We used our prototype implementation to conduct integration experiments, by launching our own nodes running the \Algopp~ software, and connecting to the actual Bitcoin network.
The goal of our experiments is to characterize how \Algopp~ affects transaction propagation latency. The experiments also  validate our implementation and its compatibility with the existing network.
%both typical cases and an attack scenario, for a wide range of parameters, and under varying levels of adoption.

For our experiments we launched  30 Dandelion instances of \texttt{m3.medium} Amazon EC2 nodes (\texttt{t2.medium} used in Seoul and Mumbai where \texttt{m3.medium} is not available). The nodes are spread geographically across 10 different AWS regions (California, Sydney, Tokyo, Frankfurt, etc.)---3 nodes per region~\cite{awsregions}. To control the topology between the Dandelion nodes, we use Core's \texttt{-connect} or \texttt{-addnode} command line flags.
Our measurements use the Coinscope~\cite{coinscope} tool to connect to each node in the Bitcoin network and record a timestamped log of transaction propagation messages.

%In order to evluate our implementation, we conducted a few tests of \Algopp~ in different scenarios.
%The first measurement is of the propagation delay as a function of the number of hops in the stem
%phase. This experiment measures the delay agnostic of the actual coin flip implementation and shows
%how propagation behaves over fixed stem lengths. The second experiment measures delay as a function
%of network adoption of \Algopp~. The third, and final, measures delay in the face of adversaries 
%that stall stem propagation (blackholes).

%% \begin{figure}
%%     \includegraphics[width=3.5in]{figures/hops-cdf-percent-all.pdf}
%%     \caption{Mean time to reach network for varying stem lengths. For every extra hop added the latency should increase on the scale of a second from the peer-to-peer latency, and this is what we see from this figure.}
%%     \label{fig:hops}
%% \end{figure}

\subsection{Evaluation Results}

\subsubsection{Propagation latency at fixed stem lengths} \label{subsec:hopdelay}
We conducted a preliminary experiment to inform our choice of the coin flip (stem length) parameter, $q$. In this experiment, we arranged our \Algopp~ nodes in a chain topology (each with one outgoing connection to the next, with the last node connected to the Bitcoin network) so that we could deterministically control the stem length. Based on 20 trials for stems up to length 12, we estimated that each additional hop adds an expected 300 miliseconds to the propagation latency. There is also a constant 2.5 second delay added to each transactions due to the exponential process of propagation when it first enters the fluff phase.
%This latency stems mainly from the exponential delays in Core's transaction propagation mechanism.
Taking a propagation delay of 4.5 seconds as our goal, we chose $q=0.1$ or $q=0.2$ as our parameter (recall the expected stem length is $1/q$).

The observed delays originate from two main sources.
%Each additional node in the stem introduces two forms of delay. 
First, each hop incurs network latency due to transit time between nodes. A recent measurement study of the Bitcoin network estimated a median latency of 110ms between nodes~\cite{stateofbitcoin}, 
and Bitcoin transaction propagation requires three messages (\texttt{INV}, \texttt{GETDATA}, \texttt{TX}). The latency between our EC2 nodes is faster, with a median of only 86ms across all pairs. Although our EC2 nodes are geographically distributed, they are closely connected to internet backbone endpoints. Second, Bitcoin Core buffers each \texttt{Inv} message for an average of 2.5 seconds; however, our implementation relays \Algopp~ transactions immediately, so internal node delays should be negligible.  We therefore estimate 300 miliseconds of delay per stem node, which is consistent with our preliminary experiments.

%In the first experiment, our nodes are connected in a fixed, geographically distributed chain.
%The nodes in the stem only form one outgoing connection to the next node in the chain, but
%continue to accept any incoming connections (so that embargo expiration has many candidates for
%relaying transactions). For every stem length, 20 transactions were sent out to the first
%node in the stem, through Bitcoin's rpc interface, at 10 second intervals. The transactions
%were indistinguishable from any other transaction and spent a fee that passes most minimum relay
%fees. The results of each stem length are show in Figure \ref{fig:hops}. The figure shows that there
%is very minimal latency added by the our implementation of stem propagation aside from the additional
%latency of making another hop.

\subsubsection{Topology} \label{subsec:topology}
In order to create a connected graph of \Algopp~ nodes, we use all of the eight outgoing connections from each of our nodes to connect to other \Algopp~ nodes. This ensures that we can measure many different stem lengths and how they affect the propagation to the rest of the network.
%As mentioned in Section~\ref{sec:partial}, we envision a gradual deployment of \Algopp~ nodes in the network. If $\beta$ fraction of the nodes in the network support \Algopp, then in expectation we would expect each node to make $8\beta$ outgoing connections to other \Algopp~ nodes.
%Since we do not wish to connect too many of our own nodes to the Bitcoin network (as this amounts to a Sybil attack and could adversely affect its topology), we instead approximate varying adoption levels by controlling the number of interconnections among our 30 \Algopp~ nodes.
%For example, since our \Algopp~ nodes are based on Bitcoin Core and therefore support a maximum number of 8 outgoing connections, to approximate 50\% adoption we would set 4 of these outgoing connections to randomly chosen \Algopp~ peers.
We must also account for an artifact of our experiment setup, namely that short cycles in the stems are more likely among our 30 well connected \Algopp~ nodes \footnote{Our implementation enters fluff mode if there is a loop in the stem.}. We therefore parse debug logs from our nodes for each trial in order to determine the effective stem length, which we then use as the basis of our evaluation.

In our experiment, we re-randomize the topology to avoid biasing our results. For each randomization of the connections, we generate 5 transactions, 10 seconds apart. We repeat this ``burst'' 4 times a day over three separate days.  The transactions are injected into the network via a randomly-chosen node. We show the results of the experiment in Figures~\ref{fig:timetoten} and \ref{fig:tentofifty}.

Figure~\ref{fig:timetoten} plots the time it takes \Algopp~ transactions to reach 10\% of the network. As mentioned in Section~\ref{subsec:hopdelay}, for every additional hop in the stem phase there is a minimum delay added by three messages and a expected delay of 2.5 seconds. The solid green line represents the minimum expected delay and the dotted blue line representst the best linear fit over the transaction data.
The propagation delay to reach 10\% coverage increases with the path length due to the minimum added delay. We also computed the two-sided Pearson Correlation Coefficient over the two variables: path length and time to reach 10\%. The coefficient $r=0.292$ implies a small positive correlation between the two variables.% and at worse no significant relationship between the two.

Figure~\ref{fig:tentofifty} plots the the time it takes a transaction to go from 10\% to 50\% of the network with respect to the path length. Unlike the first scatter plot, visual inspection doesn't reveal any relationship between the two variables in this case. This is also what we expect because \Algopp~ should not have any impact on transaction propagation after it has entered fluff phase. We also perform a Mann-Whitney U Test to test the null hypothesis: time from 10-50\% coverage does not depend on path length. Using the path length as the independent variable, we split it into two categories: high and low. The blue dotted line in Figure~\ref{fig:tentofifty} is the boundary of the two categories where there are 26 samples in ``low'' and 29 in ``high''.  The Mann Whitney U test gives a U statistic of 443 and a p-value of 0.269. This implies that there is weak evidence against the null hypothesis, therefore we fail to reject it.

%In all of the transaction data, there is high variance in the time it takes to reach 10\% of the network. To confirm that this variance is not introduced by \Algopp~, we also measure normal transaction propagation on the network during the time of the experiments. We find that there is generally very high variance in how fast transactions propagate through the network in any give period of time. Therefore, we can conclude that the high variance is not a consequence of the \Algopp~ but just normal behavior of the Bitcoin network.

As expected, the minimum delay brought on by \Algopp~ has a positive correlation with the time it takes to reach 10\% of the Bitcoin network. Similarly, once a transaction has left the stem phase, it begins normal propagation through the network and is therefore no longer be affected by \Algopp. This prediction is  confirmed by our test of independence of hop length (high v. low) and time from 10-50\% coverage.

\section{Conclusion}
\label{sec:conclusion}

A gap exists between the theory and practice of protecting user anonymity in cryptocurrency P2P networks.
In particular, there are no safeguards against population-level deanonymization, which is the focus of this paper.
We aim to narrow that gap by identifying strong or unrealistic assumptions in a state-of-the-art proposal \cite{dand}, 
demonstrating the anonymity effects of violating those assumptions, and proposing lightweight, theoretically-justified fixes in the form of \Algopp.
This methodology complements the usual development  pattern in cryptocurrencies,  which has mainly  evolved by applying ad hoc patches against specific attacks. 
We instead take a first-principles approach to design. 

Note that \Algopp~does not explicitly protect against ISP- or AS-level adversaries, which can deanonymize users through routing attacks \cite{apostolaki2016hijacking}. 
Understanding how to analyze and protect against such attacks is of fundamental interest. 
However, \Algopp~is already compatible with a number of the countermeasures proposed in \cite{apostolaki2016hijacking}. For instance, \cite{apostolaki2016hijacking} proposes to enhance network diversity through multi-homing of nodes and routing-aware network connectivity. Such countermeasures directly support \Algopp~by ensuring that nodes are less likely to establish outbound anonymity edges exclusively to spies. 
%We expect that understanding how to enforce such network diversity in the presence of Byzantine nodes will be an important question moving forward, both for anonymity and general network robustness. 

%\bibliographystyle{IEEEtranS}
\bibliographystyle{ACM-Reference-Format}
\bibliography{privacy} 

%\newpage
\appendix

\section{Algorithms}
\label{app:algo}
\Algo~pseudocode is presented in Algorithm \ref{algo:dandelion}. 
Pseudocode for handling black-hole attacks is included in Algorithm \ref{algo:dandelionpp}.

\begin{algorithm}[t]
\DontPrintSemicolon
\KwIn{Message $X_v$, source $v$, anonymity graph $H$, spreading graph $G$, parameter $q\in (0,1)$}
%\KwOut{A connected, directed graph $G(V,E)$ with average degree $d$}
anonPhase $\gets$ True \;
head $\gets v$ \;
recipients $\gets  \{v \}$ \;
\While{anonPhase} {
    \tcc{relay message to random node} 
    target $\sim$ Unif$(\mathcal N_{out}(H, \text{head}))$\;
    recipients $\gets$  recipients $\cup \{X_v\}$ from head to target \;
    head $\gets$ target \;
    $u \sim$ Unif$([0,1])$ \;
    \If{$u \leq q$} {
      anonPhase $\gets$ False \;
    }
}
\tcc{Run diffusion on $G$ from `head'} 
{\sc Diffusion}$(X_v, \text{head}, G)$
\caption{{\sc Dandelion Spreading}~\cite{dand}. $\mathcal N_{out}(G,v)$ denotes the out-neighbors of node $v$ on directed graph $G$.}
\label{algo:dandelion}
\end{algorithm}

\begin{algorithm}[t]
\DontPrintSemicolon
\KwIn{Message and timeout parameter $(X, T_\mathrm{base})$ received by $v$ in the anonymity phase, out-neighbors $\mathcal N_{out}(G,v)$ on anonymity graph $G$, spreading graph $H$, parameter $q\in (0,1)$}
%\KwOut{A connected, directed graph $G(V,E)$ with average degree $d$}
$T_\mathrm{out}(v) \sim \exp(1/T_\mathrm{base})$ \tcp*{set timer} 
forward $(X,T_\mathrm{base})$ according to dandelion\; 
\tcc{wait until message re-received}
\While{current\_time $\leq T_\mathrm{out}$} {
\If{$X$ received} {
timer $\leftarrow$ inactive\;
break\;
}
continue\;
}
\tcc{start diffusion}
\If {timer is active} 
{ 
{\sc Diffusion}$(X, v, H)$ \;
}
\caption{{\sc \Algopp~ Spreading} at node $v$. The protocol guarantees eventual network-wide propagation of transactions.}
\label{algo:dandelionpp}
\end{algorithm}

\section{Proofs}

\subsection{Proof  Theorem~\ref{thm: dregular result}} \label{apx: dandelionpp}

\thmfirstspy*
\begin{proof}
Each node has two predecessors and two successors; let us arbitrarily label these as the left predecessor (successor) and the right predecessor (successor). 
For $i,j,k,l \in \{a,h\}$, let $\mathcal{E}_v(^{i,j}_{k,l})$ denote the event that $v$'s left successor, right successor, left predecessor and right predecessor are of types $i,j,k$ and $l$ respectively, where type of $a$ denotes an adversarial node and a type $h$ denotes an honest node.
Also, assume that for any honest node $v\in V_H$ the number of messages forwarded by $v$ is statistically independent of the local neighborhood upstream of $v$. 
This assumption follows from the locally-tree-like nature of sparse random graphs \cite{mezard2009information}.
We use the following two lemmas, whose proofs are included  below. %in Appendix \ref{apx: dandelionpp}:

\begin{restatable}{lemma}{dregbound} \label{lemma: d regular fwd bound}
Let $J_v$ denote the number of transactions (from honest servers) that reach $v$ before reaching an adversary. Then, for each event $\mathcal{E}\in \{\mathcal{E}_v(^{a,h}_{h,h}), \mathcal{E}_v(^{h,a}_{h,h}), \mathcal{E}_v(^{h,h}_{a,h}), \mathcal{E}_v(^{h,h}_{h,a})\}$ 
%\begin{align*}
$
%\mathbb{E}[\max_{x\in\mathcal{X}}\mathbb{P}(X_v=x|\mathbf{S},\Gamma(V_A),\mathcal{E}, J_v)|\mathcal{E},J_v] \leq \frac{1}{J_v + 1}.
\mathbb{E}[\max_{x\in\mathcal{X}}\mathbb{P}(X_v=x|\mathbf{O},\mathcal{E}, J_v)|\mathcal{E},J_v] \leq \frac{1}{J_v + 1}.
$
%\end{align*}
\end{restatable}

\begin{restatable}{lemma}{msgsbound} \label{lemma: fwd msgs bound}
For any server $v\in V_H$, let $F_v$ denote the number of transactions that (i) reach $v$ before reaching any adversary and (ii) are forwarded by $v$ along its left outgoing edge. Then
%\begin{align}
$
\mathbb{E}\left[ \frac{1}{F_v + 1} \right] \leq \frac{2\mathbf{D}_\mathtt{FS}(v)}{p}.
$
%\end{align} 
\end{restatable}

Now, recall the events $\mathcal{E}_v(^{i,j}_{k,l})$, where $i,j,k,l\in\{a,h\}$. 
There are a total of $2^4=16$ such events that are possible for the neighborhood around server $v$.
Out of these events, $\binom{4}{2} = 6$ of them occur with a probability of $p^2(1-p)^2$ (such as $\mathcal{E}_v(^{a,a}_{h,h})$ for e.g.). Similarly $\binom{4}{3} = 4$ events occur with a probability of $p^3(1-p)$ and one event occurs with a probability of $p^4$. Since the per-node precision can be at most 1, the above events contribute to a cumulative precision gain of at most $6p^2 + 4p^3 + p^4$.  

The remaining cases are events where only one neighbor is adversarial---$\mathcal{E}_v(^{a,h}_{h,h}), \mathcal{E}_v(^{h,a}_{h,h}), \mathcal{E}_v(^{h,h}_{a,h})$ and $\mathcal{E}_v(^{h,h}_{h,a})$---or when all of the neighbors are honest $\mathcal{E}_v(^{h,h}_{h,h})$. 
Note that each of these events occur with a probability of at least $p(1-p)^3$ and hence the trivial bound used above cannot be used here. Let us first consider the event $\mathcal{E}_v(^{h,h}_{a,h})$ where only the left predecessor node of $v$ is an adversary. Let $U\in V_H$ denote the right predecessor of $v$ and $F_U$ the number of fresh transactions that are forwarded by $U$ to $v$. Then from Lemma~\ref{lemma: d regular fwd bound} we have 
\begin{eqnarray*}
\scalebox{0.9}{$\mathbb{E}[\max_{x\in\mathcal{X}}\mathbb{P}(X_v=x|\mathbf{O}), F_U, \mathcal{E}_v(^{h,h}_{a,h}))|F_U, \mathcal{E}_v(^{h,h}_{a,h})] \leq  \frac{1}{F_U+1} $} \notag \\
\Rightarrow \sum_{f\geq 0} \mathbb{P}(F_U = f, \mathcal{E}_v(^{h,h}_{a,h}))  \cdot \notag \\ 
\mathbb{E}[\max_{x\in\mathcal{X}}\mathbb{P}(X_v=x|\mathbf{O}, F_U, \mathcal{E}_v(^{h,h}_{a,h}))|F_U, \mathcal{E}_v(^{h,h}_{a,h})] \notag \\
\leq \mathbb{P}(\mathcal{E}_v(^{h,h}_{a,h})) \mathbb{E}\left[ \frac{1}{F_U +1} \right]  \leq p \frac{2\mathbf{D}_\mathtt{FS}(v)}{p} = 2\mathbf{D}_\mathtt{FS}(v).
\end{eqnarray*} 
where we use (a) the independence of $F_U$ and the local neighborhood upstream of $U$, and (b) Lemma~\ref{lemma: fwd msgs bound}. By analogous arguments we can similarly bound the expected precision under events $\mathcal{E}_v(^{h,h}_{h,a}), \mathcal{E}_v(^{a,h}_{h,h})$ and $\mathcal{E}_v(^{h,a}_{h,h})$. 

Finally consider $\mathcal{E}_v(^{h,h}_{h,h})$, in which case $v$'s location is completely hidden from the adversaries. Let $I$ be the set of such nodes. Since each adversary is a neighbor to at most 4 honest nodes, there are at least $\tilde{n}-4np = (1-5p)n$ nodes in $I$. So  $\forall x\in\mathcal{X}$, we have $\mathbb{P}(X_v = x|\mathbf{O},G,I,\mathcal{E}_v(^{h,h}_{h,h}))$
\begin{align}
 = \mathbb{P}(X_{v'} = x|\mathbf{O},G,I,\mathcal{E}_v(^{h,h}_{h,h}))&~\forall v'\in I \\
 \Rightarrow \mathbb{P}(X_{v} = x|\mathbf{O},G,I,\mathcal{E}_v(^{h,h}_{h,h})) &\leq\frac{1}{|I|} \leq \frac{1}{(1-5p)n} \notag \\
\Rightarrow \max_{x\in\mathcal{X}} \mathbb{P}(X_v = x|\mathbf{O},\mathcal{E}_v(^{h,h}_{h,h})) &\leq \frac{1}{(1-5p)n},
\end{align}
and hence $\mathbb{E}[\max_{x\in\mathcal{X}} \mathbb{P}(X_v = x|\mathbf{O},\mathcal{E}_v(^{h,h}_{h,h}))|\mathcal{E}_v(^{h,h}_{h,h})] \leq \frac{1}{(1-5p)n}$. Summing over all the cases considered gives the result. 
\end{proof}

\subsubsection{Proof of Lemma \ref{lemma: d regular fwd bound}}
\dregbound*
\begin{proof}
W.l.o.g., let $U_1,U_2,\ldots,U_{J_v}$ be the servers whose transactions are received by $v$ and let $W_v = \{v, U_1,U_2,\ldots,U_{J_v}\}$. Consider any matching $\mathbf{x}$ where $x_u$ is the message assigned to server $u$. Then 
\begin{align}
\mathbb{P}(\mathbf{S}|G, W_v, \mathbf{X}=\mathbf{x}) = \mathbb{P}(\mathbf{S}|G, W_v, \mathbf{X}=\mathbf{x'})
\end{align}
where $\mathbf{x'}$ is a new assignment of messages such that $x'_u = x_u$ for all $u\in V_H, u\notin W_v$ and $x'_{W_v} = x_{W_v}$. This implies, for any fixed $x\in\mathcal{X}$,  $\mathbb{P}(X_v = x | \mathbf{O},\mathcal{E}, J_v, G, W_v, X_{V_H\backslash W_v})=$ 
\begin{eqnarray}
&=& \frac{\mathbb{P}(X_v = x, \mathbf{O}, X_{V_H\backslash W_v} | \mathcal{E}, J_v, G, W_v )}{\mathbb{P}(\mathbf{O},  X_{V_H\backslash W_v}|\mathcal{E}, J_v,G, W_v)} \notag \\
&=& \frac{\mathbb{P}(\mathbf{O} | \mathcal{E}, J_v,G, W_v, X_v = x, X_{V_H\backslash W_v} )}{\sum_{x'}\mathbb{P}(\mathbf{O} | \mathcal{E}, J_v,G, W_v, X_v = x', X_{V_H\backslash W_v} )} \notag \\
& \leq& \frac{1}{J_v + 1} 
\end{eqnarray}
\begin{eqnarray}
&\Rightarrow & \mathbb{P}(X_v=x|\mathbf{O},\mathcal{E}, J_v)  \leq \frac{1}{J_v + 1} \notag \\
&\Rightarrow & \max_{x\in\mathcal{X}}\mathbb{P}(X_v=x|\mathbf{O},\mathcal{E}, J_v)  \leq \frac{1}{J_v + 1} \notag \\
&\Rightarrow & \mathbb{E}[\max_{x\in\mathcal{X}}\mathbb{P}(X_v=x|\mathbf{O},\mathcal{E}, J_v)|\mathcal{E},J_v] \leq \frac{1}{J_v + 1} .
\end{eqnarray}
The claim follows.
\end{proof}

\subsubsection{Proof of Lemma \ref{lemma: fwd msgs bound}}
\msgsbound*
\begin{proof}
For server $v$, consider event $\mathcal{E}_v$ in which the node incident on $v$'s left outgoing edge is an adversary. Also, let $\mathcal{L}_v$ denote the event that $X_v$ is forwarded along $v$'s left outgoing edge. Then clearly, 
\begin{eqnarray}
\mathbf{D}_\mathtt{FS}(v) \geq \mathbb{P}(\mathcal{E}_v, \mathcal{L}_v)\mathbb{E}[D_\mathtt{FS}(v)|\mathcal{E}_v, \mathcal{L}_v] = \frac{p}{2}\mathbb{E}[D_\mathtt{FS}(v)|\mathcal{E}_v, \mathcal{L}_v]. 
\label{eq: d regular lemma}
\end{eqnarray}
Now, from our assumption $F_v$ is independent of the events $\mathcal{E}_v$ and $\mathcal{L}_v$. In this case, the expected precision becomes
%\begin{align}
$
\mathbb{E}[D_\mathtt{FS}(v)|\mathcal{E}_v, \mathcal{L}_v] = \mathbb{E}\left[\frac{1}{F_v + 1}\bigg|\mathcal{E}_v, \mathcal{L}_v\right] = \mathbb{E}\left[\frac{1}{F_v + 1}\right]$, 
%\end{align}
which combined with Equation~\eqref{eq: d regular lemma} gives the lemma. 
\end{proof}

\subsection{Proof of Theorem \ref{thm:intersection}} \label{sec:intersection proof}
\intersectionthm*
%\subsubsection{Proof of Theorem \ref{thm:intersection}} 
\begin{proof}
We start with two lemmas that reduce the problem to a first-spy precision calculation. 
\begin{lemma}[Intersection]
Under each of the spreading mechanisms (all-to-one, one-to-one, and per-incoming-edge), the adversary's maximum expected precision $\mathbf{D}_{\texttt{OPT}}$ is not a function of the number of transactions per node. 
\label{lem:inter}
\end{lemma}
This proof follows directly from the pseudorandomness of the forwarding mechanisms, and is omitted for brevity. 

\begin{lemma}[First-Spy Optimality]
For all spreading mechanisms (all-to-one, one-to-one, and per-incoming-edge), there exists a constant $C$ such that $\mathbf{D}_{\texttt{OPT}} \leq C \cdot \mathbf{D}_{\texttt{FS}} + O(p^2)$, where $\mathbf{D}_{\texttt{FS}}$ denotes the expected precision of the first-spy estimator. \label{lem:fs_pseudorandom}
\end{lemma}
Lemma \ref{lem:fs_pseudorandom} is proved analogously to Theorem \ref{thm: dregular result}.
%, except instead of conditioning on the number of transactions that reach a given node $v$, 
%we condition on the number of transactions that traverse the same \emph{edge} as $v$'s own transactions. 
The full proof is omitted for brevity.
%
%Leveraging independence between
%the event that a downstream node is a spy and the expected precision, the
%optimal and first spy estimators become related by a constant factor. 
%
Together, lemmas \ref{lem:inter} and \ref{lem:fs_pseudorandom} imply that to characterize the precision-optimal estimator, we can focus on the first-spy estimator. 
For brevity, we prove only all-to-one and one-to-one results in the following lemmas.
%Characterizing this precision is itself challenging.
%\noindent \textbf{One-to-one forwarding.} 
\begin{lemma}[One-to-One First-Spy]
The expected precision of the first-spy estimator for one-to-one forwarding 
satisfies $\mathbf{D}_{\texttt{FS-OtO}} = O\Big(p^2\log\Big(\frac{1}{p}\Big)\Big)$.
\label{lem:oto}
\end{lemma}
\begin{proof}
Let $v \in V_H$, and denote the vertex to which $v$ forwards its message $X_v$
as $s \in V$. In the case that $s \in V_H$ the precision of the first spy
estimator for $v$ is $0$, 
%i.e. $D_{\texttt{FS}}(v) = 0$ when $s \not\in V_A$,
because $X_v$ is never matched to $v$. % (assuming no cycles in the graph).
When $s \in V_A$, the precision of the first spy estimator for $v$ is
$\frac{1}{|W_v|}$, where $W_v$ denotes the set of nodes from which all fresh messages
that $v$ transmits to $s$ originate. Note that $v \in W_v$. 
%So $D_{\texttt{FS}}(v) = \frac{1}{|W_v|}$ when $s \in V_A$. 
Now
$\mathbb{E}[D_{\texttt{FS}}(v)] =
 \mathbb{P}(s \in V_A)\mathbb{E}[\frac{1}{|W_v|}|s \in V_A]$,
where $\mathbb{P}(s \in V_A) = p$ and
$
\mathbb{E}\Big[\frac{1}{|W_v|} \Big| s \in V_A\Big] =
\sum_{w=1}^{\infty}\mathbb{P}\big(|W_v|=w\big| s \in V_A\big)\frac{1}{w}
$.
%To characterize this expectation, we use the following lemma:
\begin{restatable}{lemma}{primelemma}
%\begin{lemma}
%Let $v \in V_H$, and denote the vertex to which $v$ forwards its message $X_v$
%as $s \in V$. All messages received by $v$ from some edge are transmitted to $s$
%under the one-to-one forwarding scheme. Let $W_v$ be the set of nodes from which
%all fresh messages transmitted by $v$ to $s$ originate. 
Under one-to-one forwarding,
$
\mathbb{P}\big(|W_v|=w\big| s \in V_A\big) =
\frac{2p}{1-p}  \Big(\frac{1-p}{1+p}\Big)^{w}.
$ 
\label{lem:wardsize}
\end{restatable}
\emph{(Proof in Appendix \ref{app:intersection_lem}).}
%Lemma 1 states:
%\begin{gather*}
%\mathbb{P}\big(|W_v|=w\big| s \in V_A\big) =
%\frac{2p}{1-p} (-1)^w \Big(\frac{-1+p}{1+p}\Big)^{w}
%\end{gather*}

Using Lemma \ref{lem:wardsize} to expand the summation gives $\mathbb{E}  \Big[\frac{1}{|W_v|} \Big| s \in V_A\Big]=\sum_{w=1}^{\infty}\frac{1}{w}
             \frac{2p}{1-p}  \Big(\frac{1-p}{1+p}\Big)^{w} =\frac{-2p}{1-p}\log\Big(\frac{2p}{1+p}\Big)$.
%\begin{align*} 
%             \\
%            \frac{2p}{1-p} \sum_{w=1}^{\infty}
%             \frac{ (\frac{1-p}{1+p})^{w}}{w} =
%             
%\end{align*}
Thus,
$
\mathbb{E}\Big[\frac{1}{|W_v|} \Big| s \in V_A\Big] =
\frac{2p}{1-p}\log\Big(\frac{1+p}{2p}\Big)
$
and
$
\mathbb{E}[D_{\texttt{FS-OtO}}(v)]=\frac{2p^2}{1-p}\log\Big(\frac{1+p}{2p}\Big).
$
%The claim follows.
%Recall that the overall expected first spy precision
%$\mathbf{D}_{\texttt{FS}} =
%\frac{1}{(1-p)n} \sum_{i=1}^{(1-p)n} \mathbb{E}[D_{\texttt{FS}}(v_i)]$.
%Therefore,
%\begin{gather*}
%\mathbf{D}_{\texttt{FS-OtO}} = O\Big(p^2\log\Big(\frac{1}{p}\Big)\Big)
%\end{gather*}
%$\square$
\end{proof}

%\noindent \textbf{All-to-one forwarding.} 
\begin{lemma}[All-to-One First-Spy]
The expected precision of the first-spy estimator for all-to-one forwarding 
satisfies $\mathbf{D}_{\texttt{FS-AtO}} = \Theta (p)$.
\end{lemma}
\begin{proof}
%\noindent\textbf{Lemma 1}
We demonstrate upper and a lower bounds on $\mathbf{D}_{\texttt{FS-AtO}}$.
As before, these bounds are obtained by computing the expected precision of a given node $v \in V_H$, and we denote the node to which $v$ forwards its message $X_v$
as $s \in V$.
%In this case, all messages received by $v$ are transmitted to $s$.
Let $W_v$ be the set of nodes from which all fresh messages transmitted by $v$
to $s$ originate; our goal is to compute $\mathds E[\frac{1}{W_v} | s\in V_A]$.

\noindent \textbf{Lower bound.} 
We use the following lemma:
\begin{restatable}{lemma}{conditional}
Under all-to-one forwarding,
$
\mathbb{P}  \big(|W_v|=w\big| s \in V_A\big) = 
            \frac{(2w)!}{(w+1)!w!} \Big(\frac{1-p}{2}\Big)^{w-1}
             \Big(\frac{1+p}{2}\Big)^{w+1}.
$
\label{lem:conditional}
\end{restatable}
\noindent \emph{(Proof in Appendix \ref{sec:proof_lem8}).}

Note that
$\mathbb{E}\Big[\frac{1}{|W_v|} \Big| s \in V_A\Big] =
\sum_{w=1}^{\infty}\mathbb{P}\big(|W_v|=w\big| s \in V_A\big)\frac{1}{w}$.
We can lower bound this summation by computing only the first term using Lemma \ref{lem:conditional}.
This gives
$\mathbb{P}(|W_v|=1| s\in V_A)=\frac{1}{4}p^2+\frac{1}{2}p+\frac{1}{4}$.
Hence $\mathbb{E}[D_{\texttt{FS}}(v)] =
 \mathbb{P}(s \in V_A)\mathbb{E}[\frac{1}{|W_v|}|s \in V_A] \geq \frac{p}{4} + O(p^2)$,
or $\mathbb{E}[D_{\texttt{FS}}(v)] = \Omega(p)$.

\noindent \textbf{Upper bound.}
Due to the branching properties of all-to-one forwarding, we now model $v$'s upstream neighborhood (i.e. nodes that can send transactions to $v$) as a Galton-Watson (GW) branching process \cite{athreya2004branching}. 
This modeling assumption is appropriate as $n\rightarrow \infty$ due to the locally-tree-like properties of sparse random graphs \cite{mezard2009information}.

Each node in a realization of the GW process is a viable candidate source whose own transactions could reach $v$.
The branching distribution of the process is therefore determined by the fraction of spies $p$, and the forwarding mechanism (in this case, all-to-one). 
$W_v$ is the number of nodes in this tree.
\begin{figure}[t]
    \centering
  \includegraphics[width=.35\textwidth]{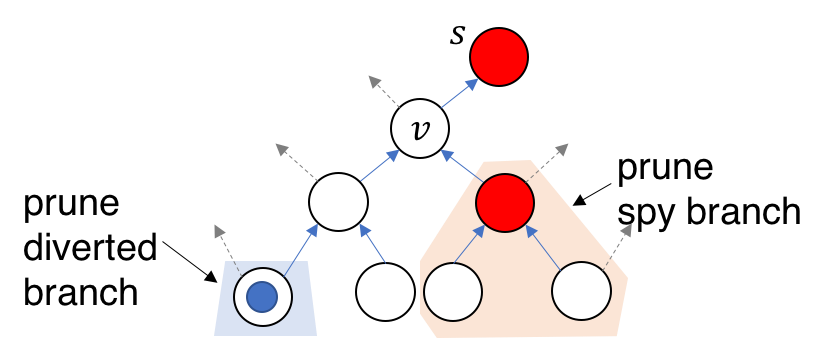}
  \caption{Branching process model of upstream neighborhood. Red nodes are spies; blue dots denote nodes that divert their transactions out of the tree.
  Both events result in pruning.}
  \label{fig:branching}
\end{figure}
Figure \ref{fig:branching} illustrates this model. 
Suppose we start with a binary tree  rooted at $v$---this includes the entire upstream anonymity graph with respect to $v$.
Each node has a second outbound edge (dashed grey line); due to our locally-tree-like assumption, those edges cannot be used to reach $v$, and are not part of the branching process. 
If a node is a spy, the spy and its upstream ancestors get pruned from the tree. 
%This happens with probability $p$, i.i.d. for each node.
Similarly, if a node chooses to forward its incoming messages along the grey edge, then \emph{none} of its upstream children's transactions will reach $v$; hence we prune the node and its ancestors.
%This happens with probability $1/2$.
We let $\nu$ denote the probability that a node is a spy or forwards its transactions on its grey edge, $\nu= \frac{1}{2}(1+p)$.
%Hence the number of children per node $U$ has the following distribution:
%\[ U=\begin{cases} 
%      0 & w.p. ~~\nu^2 \\
%      1 & w.p. ~~2 \nu (1-\nu) \\
%      2 & w.p. ~~(1-\nu)^2.
%   \end{cases}
%\]
The mean of this offspring distribution $\mu=2(1-\nu)=1-p<1$ for  $p\neq 0$, so this GW process is subcritical and goes extinct with probability 1 \cite{athreya2004branching}. 
Hence we can focus our analysis on trees of finite size. 
%As such, we can iterate over 
%Hence by standard results from GW branching processes, the expected number of nodes in this process $E[W_v] = 1/p$. 
%Using Jensen's inequality gives $\mathbb{E}[\frac{1}{|W_v|}|s \in V_A] \geq 1/ \mathbb{E}[|W_v| |s \in V_A]=p$.
%  rely on a characterization of the upstream neighborhood of 
We use Lemma \ref{lem:conditional} to expand $\mathbb{E}\Big[\frac{1}{|W_v|} \Big| s \in V_A\Big]=$
\begin{eqnarray}
&=& \sum_{w=1}^{\infty}\frac{1}{w}
             \frac{(2w)!}{(w+1)!w!} (\frac{1-p}{2})^{w-1}
             (\frac{1+p}{2})^{w+1} \nonumber \\ 
%\end{eqnarray*}
%Using Stirling's inequalities,
%$
%\frac{(2w)!}{(w+1)!w!}  \le \frac{2e^2}{\pi}\frac{1}{w+1} 4^{w-1}.
%$
%
%As a result, $\mathbb{E}  \Big[\frac{1}{|W_v|} \Big| s \in V_A\Big] \leq $
%\begin{align*}
 &=&  \frac{1}{4}(1+p)^2\sum_{w=1}^{\infty}\frac{1}{w}
       \frac{2e^2}{\pi}\frac{1}{w+1} 4^{w-1} \big(\frac{1}{4}(1-p^2)\big)^{w-1}.  \label{eq:stirling}
% & \leq & \frac{e^2}{2\pi}(1+p)^2\sum_{w=1}^{\infty}\frac{(1-p^2)^{w-1}}{w(w+1)}. \nonumber
\end{eqnarray}
where \eqref{eq:stirling} uses  Stirling's inequalities.
For $p \in [0,1]$, $\sum_{w=1}^{\infty}\frac{(1-p^2)^{w-1}}{w(w+1)} \le
 \sum_{w=1}^{\infty}\frac{1}{w(w+1)} = 1$.
Thus,
$
\mathbb{E}\Big[\frac{1}{|W_v|} \Big| s \in V_A\Big] \le \frac{e^2}{2\pi}(1+p)^2
$
and
%$
%\mathbb{E}[D_{\texttt{FS-AtO}}(v)] \le
%\frac{e^2}{2\pi}p + \frac{e^2}{\pi}p^2 + \frac{e^2}{2\pi}p^3.
%$
%Therefore,
$
\mathbf{D}_{\texttt{FS-AtO}} \le \frac{e^2}{2\pi}p + O(p^2).
$
%The claim follows.
\end{proof}

\end{proof}

\subsubsection{Proof of Lemma \ref{lem:wardsize}}
\label{app:intersection_lem}
%\noindent\textbf{Proof of Lemma \ref{lem:wardsize}}
\primelemma*

\begin{proof}
Under  one-to-one forwarding, all messages other than $X_v$
transmitted by $v$ to $s$ are received by $v$ from the same predecessor $v'$.
Likewise, all messages transmitted by $v'$ to $v$ are received by $v'$ from the
same predecessor $v''$ (other than $X_{v'}$ in the case that $X_{v'}$ is
transmitted to $v$).
Continuing this reasoning results in a line graph of predecessor nodes. Note
that the first adversary predecessor node in this line graph prevents any
subsequent predecessors from contributing to $W_v$.
%Under this model, $\mathbb{P}\big(|W_v|=w\big| s \in V_A\big)$
%may be 
We condition on $|T|$,  the number of vertices in the line graph:
%Specifically,
$
\mathbb{P}  \big(|W_v|=w\big| s \in V_A\big) = 
  \sum_{t=w}^{\infty}\mathbb{P}\big(|T|=t\big| s \in V_A\big)
   \mathbb{P}\big(|W_v|=w\big| |T|=t, s \in V_A\big).
%\end{align*}
$
Here, $\mathbb{P}\big(|T|=t\big| s \in V_A\big) = p(1-p)^{t-1}$ (recall that $v$
is a member of $T$ with probability $1$). Additionally,
$\mathbb{P}\big(|W_v|=w\big| |T|=t, s \in V_A\big) =
 {t-1 \choose w-1} (\frac{1}{2})^{w-1} (\frac{1}{2})^{t-w}$.
Therefore, $\mathbb{P}  \big(|W_v|=w\big| s \in V_A\big) = $
%\begin{align*}
$
  \sum_{t=w}^{\infty} p(1-p)^{t-1}\frac{(t-1)!}{(w-1)!(t-w)!}
   \Big(\frac{1}{2}\Big)^{w-1}\Big(\frac{1}{2}\Big)^{t-w} = 
  \frac{p}{(w-1)!}\Big(\frac{1}{2}\Big)^{w-1} \sum_{t=w}^{\infty} (1-p)^{t-1}
   \frac{(t-1)!}{(t-w)!}\Big(\frac{1}{2}\Big)^{t-1}\Big(\frac{1}{2}\Big)^{1-w} = 
  \frac{p}{(w-1)!}\sum_{t=w}^{\infty}\frac{(t-1)!}{(t-w)!}\Big(\frac{1}{2}(1-p)\Big)^{t-1} =
   2p \frac{(1-p)^{w-1}}{(1+p)^w} = 
  \frac{2p}{1-p}\Big(\frac{1-p}{1+p}\Big)^w
%\end{align*}
$
\end{proof}

\subsubsection{Proof of Lemma \ref{lem:conditional}} \label{sec:proof_lem8}
\conditional*
\begin{proof}
In a 4-regular digraph, every $v\in V$ has two predecessors. 
Thus, the upstream graph of $v$
may be modeled as a binary tree $B_v$ rooted at $v$. For any member vertex $m$
of $B_v$, if $m \in V_A$ then neither $m$ nor any members of the sub-tree rooted
at $m$ transmit fresh messages to $s$ via $v$. 
%These cases occur with probability $p$.
Additionally, every vertex transmits all received and generated messages across
one outbound edge (either left or right with equal probability) for the entire
epoch. No messages are transmitted across the other outbound edge for the entire
epoch. For any member vertex $m$ of $B_v$, if $m$ transmits messages across the
outbound edge that is not part of $B_v$ then neither $m$ nor any members of the
sub-tree rooted at $m$ may contribute messages transmitted by $v$ to $s$. These
cases occur with probability $\frac{1}{2}$.

Accounting for these cases yields a new tree $T_v$ rooted at $v$. This tree
$T_v$ consists of the remaining members of $B_v$ after pruning sub-trees rooted
at adversary nodes and sub-trees rooted at nodes that transmit all messages
across the outbound edge that is not part of $B_v$. A node is the root of a
pruned sub-tree with probability $\frac{1}{2}(1+p)$.

All messages generated by members of $T_v$ are transmitted by $v$ to $s$. The
number of nodes in $T_v$ is equal to $|W_v|$. Since a node is the root of a
pruned sub-tree with probability $\frac{1}{2}(1+p)$, then a node is a member of
$T_v$ with probability $\frac{1}{2}(1-p)$. Note that $v$ is a member of $T_v$
with probability $1$. When $T_v$ consists of $|W_v| = w$ nodes, the leaves of
$T_v$ each have two pruned children ($w+1$ in total).
Therefore,
%\begin{align*}
$
\mathbb{P}  \big(|W_v|=w\big| s \in V_A\big) = \\
            \frac{(2w)!}{(w+1)!w!} \Big(\frac{1-p}{2}\Big)^{w-1}
             \Big(\frac{1+p}{2}\Big)^{w+1}.
%\end{align*}
$
%where $\frac{(2w)!}{(w+1)!w!}$ is the Catalan number.
\end{proof}

\subsection{Proof of Proposition \ref{cor:intersection}} \label{app:intersection known}
Given a graph $G$ and observations after one epoch $\mathbf{S}$, an adversary
can construct sets $S_v$ at every adversary node. Each set $S_v$ consists of the
fresh messages forwarded by $v$. 
As a worst-case assumption, suppose the adversary learns the one-to-one
forwarding mappings, but not the edges over which honest nodes send their own
messages.
As a result, an adversary must first match each honest node $u \in V_H$ with one
of two possible sets $S_v$ and $S_{v'}$. Then, the adversary must match these
honest nodes with messages in these sets. Let $\mathcal{A}_S$ denote the event
in which $u$ is matched to the correct set $S_v$, and let $\mathcal{A}_S^C$
denote the event in which $u$ is matched to the incorrect set $S_{v'}$.

The expected precision for  $u \in V_H$ is
$
\mathbb{E}  [D_{\texttt{MAT}}(u)|G,\mathbf{S}] = 
            \mathbb{P}(\mathcal{A}_S)
             \mathbb{E}[D_{\texttt{MAT}}(u)|G,\mathbf{S},\mathcal{A}_S] +
             \mathbb{P}(\mathcal{A}_S^C)
             \mathbb{E}[D_{\texttt{MAT}}(u)|G,\mathbf{S},\mathcal{A}_S^C] = 
            \mathbb{P}(\mathcal{A}_S)
             \mathbb{E}[D_{\texttt{MAT}}(u)|G,\mathbf{S},\mathcal{A}_S] = 
            \mathbb{P}(\mathcal{A}_S)
             \mathbb{E}\Big[\frac{1}{|S_v|}\Big|G,\mathbf{S},\mathcal{A}_S\Big]=
             \mathbb{P}(\mathcal{A}_S) \frac{1}{|S_v|}.
$
The overall expected precision may be written as
$\frac{1}{(1-p)n}\sum_{u\in V_H} \mathbb{P}(\mathcal{A}_S) \frac{1}{|S_v|}$.
Each term $\frac{1}{|S_v|}$ occurs in the summation $|S_v|$ times, which means
that the expected precision over all graphs $G$ may be written as
$\frac{2pn}{(1-p)n} \mathbb{P}(\mathcal{A}_S) \mathbb{P}(|S_v|>0)$.

Note that $\mathbb{P}(|S_v|=0)$ is given by
$
\mathbb{P}  (|S_v|=0) =
             \sum_{t=0}^{\infty}\mathbb{P}(|T|=t)P\big(|S_v|=0\big||T|=t\big)=
            \sum_{t=0}^{\infty} (1-p)^t p^2 \Big(\frac{1}{2}\Big)^t =
             p^2 \frac{2}{1+p}.
$
Thus,
$
\mathbb{P}(|S_v|>0) = 1 - \frac{2p^2}{1+p} = \frac{1+p-2p^2}{1+p}.
$
Since $\frac{1}{2} \le \mathbb{P}(\mathcal{A}_S) \le 1$, then 
\begin{align*}
\frac{2pn}{(1-p)n} \frac{1}{2} \frac{1+p-2p^2}{1+p} \le
\mathbf{D}_{\texttt{MAT}} \le
\frac{2pn}{(1-p)n} 1 \frac{1+p-2p^2}{1+p}
\end{align*}
Simplifying gives
$
\frac{p+p^2-2p^3}{(1-p^2)} \le \mathbf{D}_{\texttt{MAT}} \le
\frac{2p+2p^2-4p^3}{(1-p^2)}.
$
%Therefore $\mathbf{D}_{\texttt{OPT-OtO}} = O(p)$.

\subsection{Proof of Proposition \ref{thm:outbound}}
\label{app:outbound}

The proof follows by identifying that the first spy node to receive a message along \algo's stem is a sufficient statistic for detection. 
Since the stem-phase occurs via peers' outbound edges, $\mathcal{P}$ and $\mathcal{Q}$ have similar stem-phase propagation and differ only in the diffusion phase. 
As such precision and recall are not affected by how the message spreads in the diffusion phase. 

For simplicity, let us assume a small $q$, i.e., messages are always received by a spy node in the stem-phase before diffusion begins. 
For any message $x\in\mathcal{X}$ let $\mathbf{O}^1_x$ be the random variable comprising a three-tuple $(U_h,U_a,T)$ where, in the stem-phase propagation of $x$ (i) $U_a\in V_A$ is the first spy to receive $x$, (ii) $U_h\in V_H$ is the honest peer that forwarded $x$ to $U_a$ and (iii) $T$ is the time when $x$ was received by $U_a$.  
Next, let $\mathbf{O}^2_x$ denote the random variable that comprises of observations made by the adversary after it has been forwarded by $U_a$ in the stem-phase. 
This includes all tuples $(u,v,t)$ such that honest peer $u\in V_H$ forwarded $x$ to $v\in V_A$ at time $t > T$. 
Lastly let $Y_x$ denote the honest peer $v\in V_H$ that is the source of transaction $x$. 
To get a worst-case guarantee we also assume that the adversary has complete knowledge of the topology $H$ of the anonymity graph. 

Since the stem-phase propagation is over a line, we observe that once a message $x$ reaches a spy node $U_a$ for the first-time, the subsequent spreading dynamics depends entirely on the action taken by $U_a$ (who it forwards $x$ to, when it forwards etc.) and is conditionally independent of the past. 
This is true for every message $x\in \mathcal{X}$.  
As such we have, 
\begin{align}
\prob(\mathbf{O}_{x_1}^2,\ldots,\mathbf{O}_{x_n}^2|\mathbf{O}_{x_1}^1,\ldots,\mathbf{O}_{x_n}^1,Y_{x_1},\ldots,Y_{x_n},H) \notag \\
= \prob(\mathbf{O}_{x_1}^2,\ldots,\mathbf{O}_{x_n}^2|\mathbf{O}_{x_1}^1,\ldots,\mathbf{O}_{x_n}^1,H) \\
\Rightarrow \prob(Y_{x_1},\ldots,Y_{x_n}|\mathbf{O}_{x_1}^1,\ldots,\mathbf{O}_{x_n}^1, \mathbf{O}_{x_1}^2,\ldots,\mathbf{O}_{x_n}^2,H) \notag \\
= \prob(Y_{x_1},\ldots,Y_{x_n}|\mathbf{O}_{x_1}^1,\ldots,\mathbf{O}_{x_n}^1, H) \\
\Rightarrow \prob(Y_{x_i}|\mathbf{O}_{x_1}^1,\ldots,\mathbf{O}_{x_n}^1, \mathbf{O}_{x_1}^2,\ldots,\mathbf{O}_{x_n}^2,H) \notag \\
= \prob(Y_{x_i}|\mathbf{O}_{x_1}^1,\ldots,\mathbf{O}_{x_n}^1,H) \quad \forall i\in[n]. \label{eq: posterior}
\end{align}
Thus the posterior is conditionally independent of later observations, given stem-phase observations  $\mathbf{O}_{x_1}^1,\ldots,\mathbf{O}_{x_n}^1$.
Now, consider a network $H'$ that is derived from $H$ by removing all outgoing edges from adversarial peers.
Since the observations $\mathbf{O}^1_x$ log transactions that have been received for the first time by a spy, 
it implies the routes taken by the transactions do not include any spy node. Hence the statistics of the stem-phase spreading are identical in $H'$ and $H$. 
Mathematically this implies,
\begin{align}
\prob(Y_{x_i}|\mathbf{O}_{x_1}^1,\ldots,\mathbf{O}_{x_n}^1,H) &= \frac{\prob(Y_{x_i},\mathbf{O}_{x_1}^1,\ldots,\mathbf{O}_{x_n}^1|H)}{\prob(\mathbf{O}_{x_1}^1,\ldots,\mathbf{O}_{x_n}^1|H)} \notag \\
=  \frac{\prob(Y_{x_i},\mathbf{O}_{x_1}^1,\ldots,\mathbf{O}_{x_n}^1|H')}{\prob(\mathbf{O}_{x_1}^1,\ldots,\mathbf{O}_{x_n}^1|H')} &= \prob(Y_{x_i}|\mathbf{O}_{x_1}^1,\ldots,\mathbf{O}_{x_n}^1,H'). \label{eq: posterior 2}
\end{align}
From Theorems 3 and 4 in~\cite{dand} we know that the optimal value of the expected precision and recall is a function of the posterior probabilities $\prob(Y_{x_i}|\mathbf{O}_{x_1}^1,\ldots,\mathbf{O}_{x_n}^1, \mathbf{O}_{x_1}^2,\ldots,\mathbf{O}_{x_n}^2,H)$, which by combining Equations~\eqref{eq: posterior} and~\eqref{eq: posterior 2} in turn equals $\prob(Y_{x_i}|\mathbf{O}_{x_1}^1,\ldots,\mathbf{O}_{x_n}^1, H')$. 

We finish the proof by applying the above results on networks $\mathcal{P}$ and $\mathcal{Q}$. 
Let $H_\mathcal{P}$ and $H_\mathcal{Q}$ denote the topologies of $\mathcal{P}$ and $\mathcal{Q}$ respectively;
let $H'_\mathcal{P}$ and $H'_\mathcal{Q}$ denote the networks obtained by removing outgoing spy edges from $\mathcal{P}$ and $\mathcal{Q}$ respectively. 
By construction we have $H'_\mathcal{P} = H'_\mathcal{Q}$. 
As such the two probability spaces, each comprising of the random variables $Y_{x_1},\ldots,Y_{x_n},\mathbf{O}^1_{x_1},\ldots,\mathbf{O}^1_{x_n}, H'$, pertaining to the networks $\mathcal{P}$ and $\mathcal{Q}$ are identical. 
Hence we conclude the optimal values of precision and recall in the two networks must also be the same.

\subsection{Proof of Proposition~\ref{prop: clock rate}} \label{app: clock rate}

\begin{proof}
%The proof is straightforward and follows by evaluating the probability of the desired event. 
Let $v_1$ be the source of a message that propagates along a path $v_1,v_2,\ldots,v_k$ of length $k$. 
Let $\delta_\mathrm{hop}$ be the delay incurred between each hop, and let $T_\mathrm{out}(v_i)$ be the random timeout at node $v_i$ for $i=1,\ldots,k$. 
Note that the message takes $\delta_\mathrm{hop} k$ time to traverse $k$ hops and reach $v_k$. 
We desire that none of the $v_i$, $i=1,\ldots,k$, initiate diffusion during this time with high probability. 
Since the random variables $T_\mathrm{out}(v_i)$ are exponential, this probability can be bounded as
\begin{align*}
\left( e^{-(k-1)\delta_\mathrm{hop}/T_\mathrm{base}} \right)  \left(e^{-(k-2) \delta_\mathrm{hop}/T_\mathrm{base}} \right)  \ldots \left( e^{-(k-k)\delta_\mathrm{hop}/T_\mathrm{base}} \right) \notag \\
= e^{-k(k-1)\delta_\mathrm{hop}/(2T_\mathrm{base})} \geq 1 - \epsilon 
\Rightarrow T_\mathrm{base} \geq \frac{-k(k-1)\delta_\mathrm{hop}}{2\log(1-\epsilon)}, 
\end{align*}
%and the proposition follows. 	
\end{proof}

\subsection{Proof of Theorem \ref{thm:recall_partial}}
\label{app:recall_partial}
Under \algo~spreading, for any message $x$, the recall-optimal estimator chooses the node $v$ that maximizes $\prob(X_v=x | \mathbf O)$ (Theorem 4 \cite{dand}).
Since we assume a uniform prior on sources, this is equivalent to a maximum-likelihood estimator that returns $\hat v = $ argmax$_v ~\prob(\mathbf O | X_v=x )$.
%The first-spy estimator satisfies this criterion if there is a spy node in the stem (Theorem 4 in \cite{dand}). 
From \cite{dand}, we know that for a given adversarial mapping strategy, the expected recall is equivalent to the probability of detection, $\prob(\detec)$.
%Therefore, if the message reaches a spy node before entering the spreading phase, we only need to analyze the first-spy estimator.

Consider a \Algopp~ source node $v$ %(drawn uniformly from $V_D$), 
that transmits a  transaction $x$. 
In order to characterize the expected recall over all honest nodes in $V_D$, by symmetry, it is sufficient to compute the probability of detecting $v$ as the source of $x$, 
where the probability is taken over the spreading realization, randomness in the graph, and any randomness in the adversary's estimator.
This proof bounds the probability of detection for $v$ under version-checking and no-version-checking.
Let $D$ denote the event where $v$ has at least one outbound neighbor that supports \Algopp~(i.e., $|\mathcal D_v| > 0$), and let $\overline D$ denote the complement of that event.

\vspace{0.05in}
\noindent \textbf{Version-checking:} 
For the lower bound, we have
%\begin{eqnarray*}
$\prob(\detec) = \prob(\detec |D) \prob(D) + \prob(\detec | \overline D) \prob(\overline D)
		\geq  \prob(\detec |D) \prob(D)$.
%\end{eqnarray*}
We can separately bound each of these terms:
\begin{eqnarray*}
\prob(D) = 1 - \frac{n-1-|V_D|}{n-1} \frac{n-2-|V_D|}{n-2} \ldots \frac{n-\eta-|V_D|}{n-\eta} \\
	\geq 1 - \left( \frac{n-1-|V_D|}{n-1}\right )^\eta
	\geq 1 - \left( 1 - \frac{|V_D|}{n}\right )^\eta =1 - \left( 1 - f\right )^\eta,
\end{eqnarray*}
where $f = p + (1-p)\beta $ is the total fraction of nodes  running \Algopp, 
and
$
\prob(\detec | D) \geq \frac{p}{f},
$
since the first-spy estimator detects the true source if the first node in the stem is a spy node.
Thus $\prob(\detec)\geq \frac{p}{f}(1-(1-f)^\eta)$.

\vspace{0.1in}
\noindent To compute the upper bound, we have
\begin{eqnarray*}
\prob(D) &\leq& 1 - \left (\frac{n-\eta-|V_D|}{n-\eta} \right)^\eta = 1 - \left (1 - \frac{fn}{n-\eta} \right)^\eta \\
\prob(\overline D) &\leq& \left( \frac{n-1-|V_D|}{n-1}\right )^\eta = (1-f)^\eta.
\end{eqnarray*}
Trivially, $\prob(\detec|\overline D) \leq 1$.To bound $\prob(\detec|D)$, we condition on the event $S$, where $v$'s first node in its \Algopp~stem (call it $w$) is a spy node. 
The first-spy estimator is recall-optimal if there is a spy node in the stem (Theorem 4 in \cite{dand}). 
We have $\prob(\detec|D) = \prob(\detec|D,S)\prob(S|D) + \prob(\detec|D,\overline S)\prob(\overline S|D)$, and
\begin{eqnarray*}
%\end{eqnarray*}
%We can bound these terms by assuming that the adversary knows the anonymity graph.
%Moreover, an oracle gives the adversary the following information: 
%for each hop traversed by the transaction message between the source and the last stem node, the adversary learns which edges were used to forward the message. 
%Along with these edges, the adversary is also given a set of ``fake" edges traversed prior to the source. 
%That is, starting with the true source, the oracle picks an edge uniformly at random from the inbound edges to the source.
%The chosen node becomes the head of the ``extended stem." 
%This process continues from the head node, until all the honest nodes in the network have been included in the extended stem.
%Thus, the adversary learns a sequence of nodes; it does not know where the fake edges stop and the true edges start.
%However, it does know where the stem ended, and it knows that by construction, this node cannot be the true source.
%In this scenario, the first-spy estimator is recall-optimal (where we treat the last node in the stem as a spy). 
%Given that, the adversary's recall-maximizing strategy is to estimate the source as the second-to-last node in the stem. 
%This strategy gives the following:
%\begin{eqnarray*}
\prob(\detec|D,S) \leq 1, ~ \prob(S|D) = \frac{pn}{fn-1}, ~ \prob(\overline S|D) \leq 1-\frac{p}{f}.
%\implies \prob(\detec|D) &\leq& \frac{pn}{fn-1}.
\end{eqnarray*}
To bound $\prob(\detec|D,\overline S)$, we condition on $F$, the event where $w$ chooses to extend the stem.
%\begin{eqnarray}
$
\prob(\detec|D,\overline S) = \prob(\detec|F, D,\overline S)\prob(F|D,\overline S) 
 + \prob(\detec|\overline F, D,\overline S)\prob(\overline F|D,\overline S) $.
%\end{eqnarray}
For this upper bound, we also assume that an oracle gives the adversary the source of the diffusion process in the spreading phase.
That is, let $\ell_1(x), \ldots, \ell_{M_x}(x)$ denote the stem nodes associated with transaction $x$; 
in this case, $\ell_1(x)=v$, and $M_x$ denotes the length of the stem.
We assume an oracle gives the adversary $\ell_{M_x}(x)$.
If $w$ chooses not to extend the stem (event $\overline F$), then the $\ell_{M_x}(x) = w$. 
We also assume that the adversary knows $V_D$, the set of nodes running \Algopp.
Hence, we have 
%the first-spy estimator maximizes the adversary's recall, where we treat the final node in the stem as a spy.
%	\prob(\detec|F, D,\overline S) &=& 0 \label{eq:dand_forward} \\
	$\prob(F|D,\overline S) = 1-q $ and
	$\prob(\overline F|D,\overline S) = q$.
%\end{eqnarray} 
%where \eqref{eq:dand_forward} holds because the most likely source is always the last node in the stem.
%Since each \Algopp~stem node extends the stem with probability $\beta$, under events $D$ and $\overline S$, $v$ gets 
We now wish to bound $\prob(\detec|F, D,\overline S)$ and $\prob(\detec|\overline F, D,\overline S)$---the probabilities of detection given that $v$ passed the message to honest \Algopp~neighbor $w$, conditioned on $w$'s decision to either extend or terminate the stem, respectively. 
%First, note that  $\prob(\detec|\overline F, D,\overline S)=\prob(\detec| \overline D,\overline S)$; we will use this fact later.
%Although we assume the adversary knows $\ell_{M_x}(x)$, 
%the adversary does \emph{not} know the previous node in the stem (since $w$ is honest, it does not report this information to the adversary).
%Moreover, 
%it does not know if $\ell_{M_x}(x)$ was the source or a relay.

%It is straightforward to show that when there are no spies in the stem (i.e., $V_A \cap \{\ell_1(x),\ldots,\ell_{M_x}(x)\} = \emptyset$), if $\ell_{M_x}(x) \notin V_D$, the adversary's recall-maximizing strategy is to guess $\ell_{M_x}(x)$. 
%So in this case, $\ell_1(x)=v$, $M_x=2$, and $\ell_{2}(x)=w$.
Recall that if there is a spy in the stem (e.g., $\ell_i(x) \in V_A$ for some $i \in [M_x]$), then the first-spy estimator is recall-optimal.
If there are \emph{no} spies in the stem (i.e., when $V_A \cap \{\ell_1(x),\ldots,\ell_{M_x}(x)\} = \emptyset$),
 $\ell_1(x) \rightarrow \ell_{M_x}(x) \rightarrow \mathbf O$ form a Markov chain.
Since none of the stem nodes are are spies, the spy observations are conditionally independent of the source node given $\ell_{M_x}(x)$.
Since the adversary learns $\ell_{M_x}(x)=\ell$ exactly from the oracle, the recall-optimizing strategy becomes 
$
\hat v = \argmax_{u} \prob(\ell_{M_x}(x)=\ell |X_u=x).
$
%If $\ell \in V_D$, then $\prob(\ell_{M_x}(x)=\ell |X_\ell=x)=0$, and the most likely source becomes the second-to-last node in the stem.
%Otherwise, $\prob(\ell_{M_x}(x)=\ell |X_\ell=x)=1$, so the most likely source is $\ell$ itself.
%Now we consider the two desired probabilities:

%\vspace{0.1in}
\noindent \textit{Part 1:} 
To bound $\prob(\detec|\overline F, D,\overline S)$, note that $\ell_{M_x}(x)=w$. 
Suppose that in addition to revealing $w$, the oracle also tells the adversary that $w$ is not the true source. 
%Since $w \in V_D$, the adversary knows that if $w$ is the true source, it must have passed the transaction in a loop over $G$ (because \Algopp~nodes always forward their own messages at least one hop). 
Consider the set $R=\{u \in V_D ~|~ w \in \mathcal D_u\}$, which contains all nodes that could have feasibly relayed a \Algopp~transaction to $w$.
%For each $u\in R$, we have
%$$
%\prob(\ell_{M_x}(x)=w |X_u=x) = \frac{1}{|\mathcal D_u|}(q + (1-q)\delta),
%$$
%where 
%$$
%\delta=\sum_{i=3}^\infty \prob(\ell_{i}(x)=w, M_x=i |\ell_2(x)=w)
%$$ 
%is the probability that the stem, having reached $w$ without terminating, loops back to $w$ and terminates at some later hop.
%Conditioned on the stem passing through $w$, the rest of the stem is independent of the source, so $\delta$ is source-independent. 
%Similarly, we have $\prob(\ell_{M_x}(x)=w |X_w=x) = \delta$.
%Since bidirectional edges are not allowed in $G$, it must hold that a stem starting from $w$ must traverse at least two hops before returning to $w$; 
%hence $\delta<(1-q)^2$.
%Let $\hat v = \argmin_u |\mathcal D_u|$.
%Thus when 
%\begin{eqnarray*}
%\delta &<& \frac{1}{|\mathcal D_{\hat v}|}(q + (1-q)\delta), \\
%\implies \delta \left ( 1 - \frac{1-q}{|\mathcal D_{\hat v}|}\right) &<& \frac{q}{|\mathcal D_{\hat v}|} \\
%\implies \delta &<& \frac{q}{|\mathcal D_{\hat v}| - 1 + q}
%\end{eqnarray*}
For nodes $u\in R$, the likelihood of each node being the source is
%\begin{eqnarray*}
$
  \prob(\ell_{M_x}(x)=w |X_u=x) = \frac{1}{|\mathcal D_u|}(q + (1-q)\delta) 
$
%	&=& \argmin_{u \in R} |\mathcal D_u| ,
%\end{eqnarray*}
where $\delta$ is the probability that the stem, having reached $w$ without terminating, loops back to $w$ and terminates at some later hop.
Conditioned on the stem passing through $w$, the rest of the stem is independent of the source, so $\delta$ does not depend on $u$.
Let $\hat v$ denote the most likely source among the nodes in $R$.
$
\hat v = \argmax_{u \in R} \prob(\ell_{M_x}(x)=w |X_u=x) 
      = \argmin_{u \in R} |\mathcal D_u|.
$
It is straightforward to show that for any alternative source $z\neq \hat v$, $z\in V_H$ has a lower likelihood.
This follows trivially if $z\in R$. 
If $z \notin R$, the stem from $z$ to $w$ would require at least two hops, which reduces the likelihood of candidate source $v^*$ by a factor of at least $(1-q)$.
Hence $\hat v$ is the ML source estimate, and we want to know $\prob(\hat v=v)$.
Since each node $u$ in $R$ is equally likely to have the smallest $\mathcal D_u$ set,
we have 
%$\prob(\hat v=v) = 1/|R|$, where the probability is taken over the random construction of $H$ and $G$.
$\prob(\detec | \overline F, D,\overline S) \leq \E[\frac{1}{|R|}]$.
We know that $|R| \geq 1$, because $v$ is connected to $w$ by construction.
%Hence the probability of detection is $\E[1/|\mathcal Z|]$, where the expectation is taken over the graph construction algorithm.
Since each node chooses its $\eta$ connections independently, this can be computed as $\E[1/(Z+1)]$ where $Z\sim \text{Binom}(\tilde n-1, \phi)$, and $\phi$ is the probability of any given node choosing $w$ as one of its outbound edges. We can compute $\phi$ as 
$
\phi = 1 - \frac{n-2}{n-1} \frac{n-3}{n-2} \ldots \frac{n-\eta-1}{n-\eta} \\
	\geq  1 - \left (1 - \frac{1}{n-\eta} \right )^\eta.
$
Henceforth, we will abuse notation and take $\phi = 1 - \left (1 - \frac{1}{n-\eta} \right )^\eta$.
By using this lower bound on $\phi$, we are reducing the probability of any given node choosing $w$ as an outbound edge, and thereby increasing the overall probability of detection. 
Given that, it is straightforward to show that 
\begin{equation}
\prob(\detec|\overline F, D,\overline S) \leq \E\left [\frac{1}{Z+1} \right ] = \frac{1 - (1-\phi)^{\tilde n}}{\tilde n \phi} \triangleq \zeta.
\label{eq:zeta0}
\end{equation}

\vspace{0.1in}
\noindent \textit{Part 2:} We want to show that $\lim_{n\rightarrow \infty}\prob(\detec|F, D,\overline S)=0$; 
if this is the case, then the asymptotic inequality in Theorem \ref{thm:recall_partial} holds for any $C'>1$.
There are three ways the adversary can identify the correct source $v$ under conditions $F$, $D$, and $\overline S$: 
1) the stem eventually loops back to $v$, and then transmits to a spy node (we call this event $A$), 
2) the stem loops back to $v$, which terminates the stem (event $B$), or
3) the stem terminates before reaching a spy node, and the set of \Algopp~nodes with outbound edges to the stem's terminus $\ell$ includes $v$ (event $C$).
Note that we are still assuming the adversary learns the last stem node $\ell$.
Also, $B$ and $C$ are not sufficient conditions for detection, but they do guarantee a nonzero probability of detection under a recall-optimal estimator.
Therefore, since these events are disjoint, $\prob(\detec|F, D,\overline S) \leq \prob(A) + \prob(B) + \prob(C)$;  we can individually bound each of these events.

To bound $\prob(A)$, we first compute the probability of the stem looping back to $v$ without reaching any spy nodes or terminating (we call this event $A'$); 
since this event is necessary but not sufficient for detection under event $A$, we have $\prob(A) \leq \prob(A')$.
$\prob(A')$ can be upper bounded by relaxing the restriction of not hitting any spy nodes before reaching $v$. 
So we just want the probability of the stem looping around and reaching $v$ before terminating. 
We let  $\tilde M_x$ denote the stem length (ignoring the first two hops of $v$ and $w$);  it is a geometric random variable with parameter $q$.
We let $T_v$ denote the random number of hops before hitting node $v$, given starting point $w$ (the number of hops excludes $w$).
Thus  $\prob(A') \leq \prob(T_v \leq \tilde M_x)$.
As an upper bound, we assume all nodes run \Algopp. 
Hence, each node has an equal probability $\frac{1}{n-1}$ of forwarding the message to $v$, so $T_v \sim \text{Geom}(\frac{1}{n-1})$.
%\begin{eqnarray*}
Then $
\prob(T_v \leq \tilde M_x) = \sum_{i=1}^\infty \prob(\tilde M_x = i) \prob(T_v \leq i ) 
= \sum_{i=1}^\infty q(1-q)^{i-1} (1-(1-\frac{1}{n})^i) 
%&= & 1 - \frac{q}{1-q}\sum_{i=1}^\infty (1-q)^i (\frac{n-1}{n})^i \\
%&= & 1 - \frac{q}{1-q} (1-q) (\frac{n-1}{n}) \frac{1}{1 - (1-q)(1-\frac{1}{n})} \\
=  1 - q (1 - \frac{1}{n}) \frac{1}{1 - (1-q)(1-\frac{1}{n})} .
%\end{eqnarray*}
$
Taking the limit gives
%\begin{eqnarray*}
$\lim_{n\rightarrow \infty} 1 - q (1 - \frac{1}{n}) \frac{1}{1 - (1-q)(1-\frac{1}{n})}  = 1-\frac{q}{q}=0$,
%\end{eqnarray*}
so $\lim_{n\rightarrow \infty} \prob(A) =0$.
The same argument applies for event $B$, so $\lim_{n\rightarrow \infty} \prob(B) =0$.

For event $C$ we have $\prob(C) = \sum_{u \in V_H \setminus \{v\}} \prob(\ell_{M_x(x) }= u) \prob(u \in \mathcal D_v)=$
\begin{eqnarray}
 &=& \prob(\ell_{M_x(x) }= w)  +  \sum_{u \in V_D \setminus \{v,w\}} \prob(\ell_{M_x(x) }= u)\prob(u \in \mathcal D_v) \label{eq:terminus}\\
 &\leq& \prob(\ell_{M_x(x) }= w)  + \gamma  \label{eq:terminus2}
\end{eqnarray}
where \eqref{eq:terminus} holds because we already know that  $w \in \mathcal D_v$ by definition, and since we are conditioning on event $D$, $v$ will never relay the message to a non-\Algopp~node.
\eqref{eq:terminus2} holds because each \Algopp~node other than $w$ is equally likely to be in $\mathcal D_v$, so for $u \in V_D \setminus \{v,w\}$, $\gamma \triangleq \prob(u \in \mathcal D_v) $ does not depend on $u$.
By the same logic as before, $\lim_{n\rightarrow \infty} \prob(\ell_{M_x(x) = w}) = 0$.
Hence we only need to show that $\lim_{n\rightarrow \infty} \gamma = 0$.
We can write out $\gamma = \frac{{n \choose \eta-2}}{{n \choose \eta-1}} = \frac{\eta-1}{n-\eta+2}$, so $\lim_{n\rightarrow \infty} \gamma = 0$.
Given this, we have that $\lim_{n\rightarrow \infty}\prob(\detec|F, D,\overline S)=0$.

%It will not transmit to a node that does not support \Algopp~because of the version-checking policy.
%Note that the stem length $M_x$ is geometrically-distributed with parameter $q$; this random variable's distribution does not scale with the network size. 
%Hence, as $n\rightarrow \infty$, the probability of the stem looping around back to $v$ within $M_x-1$ steps vanishes, 
%because the likelihood of any given node having an outbound connection to $v$ vanishes. 
%Hence, $\prob(\detec|F, D,\overline S) \rightarrow 0$.
%The claim follows.

%Suppose $w\in V_D$, $F$, and there is a spy in the stem. 
%Then the adversary runs the first-spy estimator.
%Since the first spy occurs after $w$ in the stem, first-spy estimator returns $v$ with a probability at most $(1-q)^2$,
%because the stem cannot backtrack in subsequent time steps.
%Hence the probability of detection is this case is at most $(1-q)^2$.
%APPROX 0

\noindent Combining the two parts gives
\begin{eqnarray}
\prob(\detec|D,\overline S) \lesssim  \zeta q .
\label{eq:zeta}
\end{eqnarray}
Overall, we have
%\begin{eqnarray*}
%\prob(\detec) &\leq& \left (1 - \left(1-\frac{fn}{n-d}\right )^d \right ) \\
%&& \left (\frac{pn}{fn-1}  + \left( 1-\frac{p}{f}\right)(1-\beta)\zeta \right ) + (1-f)^d.
%\end{eqnarray*}
%\begin{eqnarray*}
$
\prob(\detec) \lesssim \left (1 - \left(1-\frac{fn}{n-\eta}\right )^\eta \right )\left (\frac{pn}{fn-1} + (1-\frac{p}{f})q \zeta \right)  + (1-f)^\eta.
%&& \left (\frac{pn}{fn-1}  + \left( 1-\frac{p}{f}\right)(1-\beta)\zeta \right ) + (1-f)^\eta.
$
%\end{eqnarray*}
Taking the limit as $n\rightarrow \infty$ gives 
$
%\left (1 - \left(1-f\right )^d \right ) \left \frac{p}{f}  + \left( 1-\frac{p}{f}\right)(1-\beta)\zeta \right ) + (1-f)^d.
\left (1 - \left(1-f\right )^\eta \right ) \left (\frac{p}{f} + (1-\frac{p}{f})q \zeta \right ) + (1-f)^\eta.
$
The claim follows.

\vspace{0.05in}
\noindent \textbf{No-version-checking:} 
The lower bound comes directly from Theorem 2 in \cite{dand}. 

To prove the upper bound, we first condition on whether $v$'s selected stem relay $w$ is a spy node; as before, $S$ denotes this event.
In this section, we will repurpose our previous notation and use $D$ to denote the event where the \emph{selected} relay supports \Algopp.
Again, we will assume that the adversary knows the underlying graph $H$, $V_D$, and the final stem node $\ell_{M_x}(x)$. 
We have 
\begin{eqnarray}
\prob(\detec) = \prob(\detec|S)\prob(S) + \prob(\detec|\overline S)\prob(\overline S) \label{eq:det} \\
\prob(S) = \frac{pn}{n-1} ~~ \prob(\overline S) = \frac{(1-p)n}{n-1} ~~ \prob(\detec | S) = 1 \nonumber 
\end{eqnarray}
To bound $\prob(\detec | \overline S)$, we condition on $D$: $\prob(\detec|\overline S) = \prob(\detec|\overline S,D)\prob(D | \overline S) + \prob(\detec|\overline S,\overline D)\prob(\overline D | \overline S)$.
\begin{eqnarray}
\prob(D | \overline S) &=& \frac{\beta (1-p)n}{(1-p)n-1} = \frac{\beta \tilde n}{\tilde n-1} \nonumber \\
\prob(\overline D | \overline S) &=& \frac{(1-\beta )(1-p)n}{(1-p)n-1} = \frac{(1-\beta)\tilde n}{\tilde n -1} \nonumber \\
\prob(\detec|\overline S,D) &\lesssim & \zeta q, \label{eq:zeroprob} \\
\prob(\detec|\overline S,\overline D) &\leq & \zeta, \label{eq:zeroprob2} 
\end{eqnarray}
where \eqref{eq:zeroprob} follows from \eqref{eq:zeta}, and \eqref{eq:zeroprob2} follows from \eqref{eq:zeta0} by assuming the adversary is told that $w$ is not the true source.
Combining gives
$$
\prob(\detec|\overline S) \lesssim \zeta \left (\frac{q \beta \tilde n}{\tilde n-1} + \frac{ (1-\beta )\tilde n}{\tilde n-1} \right ) = \frac{\zeta (1-\beta(1- q) )\tilde n}{\tilde n-1}.
$$
Plugging into \eqref{eq:det} gives
$
\prob(\detec) \leq \frac{pn}{n-1}+ \frac{(1-p)n}{n-1}\frac{\zeta (1-\beta(1-q))\tilde n}{\tilde n-1}.
$
Taking the limit as $n\rightarrow \infty$ gives $p+\zeta (1-\beta(1-q))(1-p)$.

\begin{figure*}[t]
    \centering
    \begin{minipage}{.45\textwidth}
    \centering
  \includegraphics[width=\linewidth]{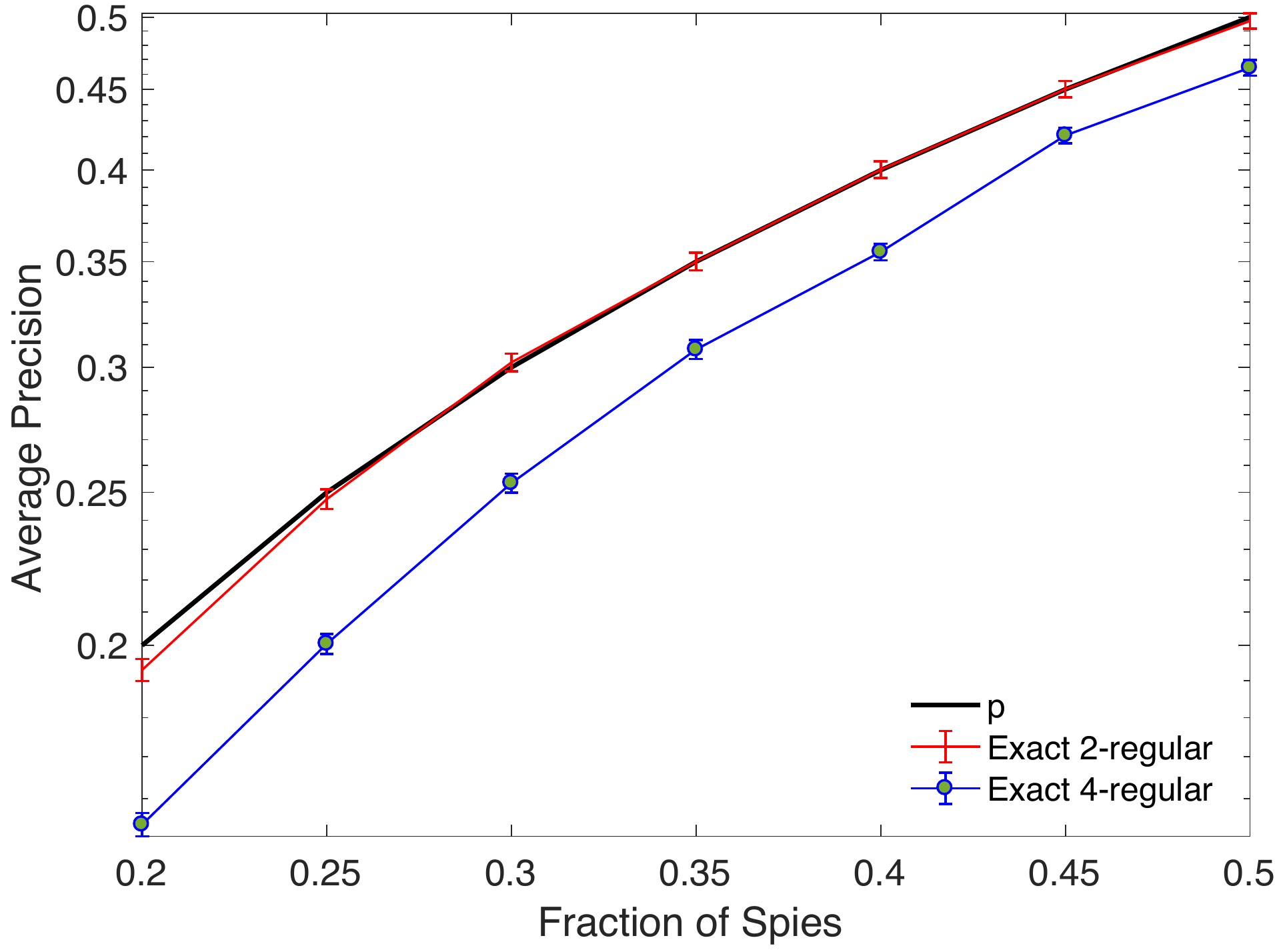}
  \caption{One-to-one transaction forwarding when the adversary knows the graph but does not know the nodes' internal routing decisions.}
  \label{fig:unknown}
    \end{minipage}
    \hspace{0.05in}
    \begin{minipage}{.45\textwidth}
    \centering
  \includegraphics[width=\linewidth]{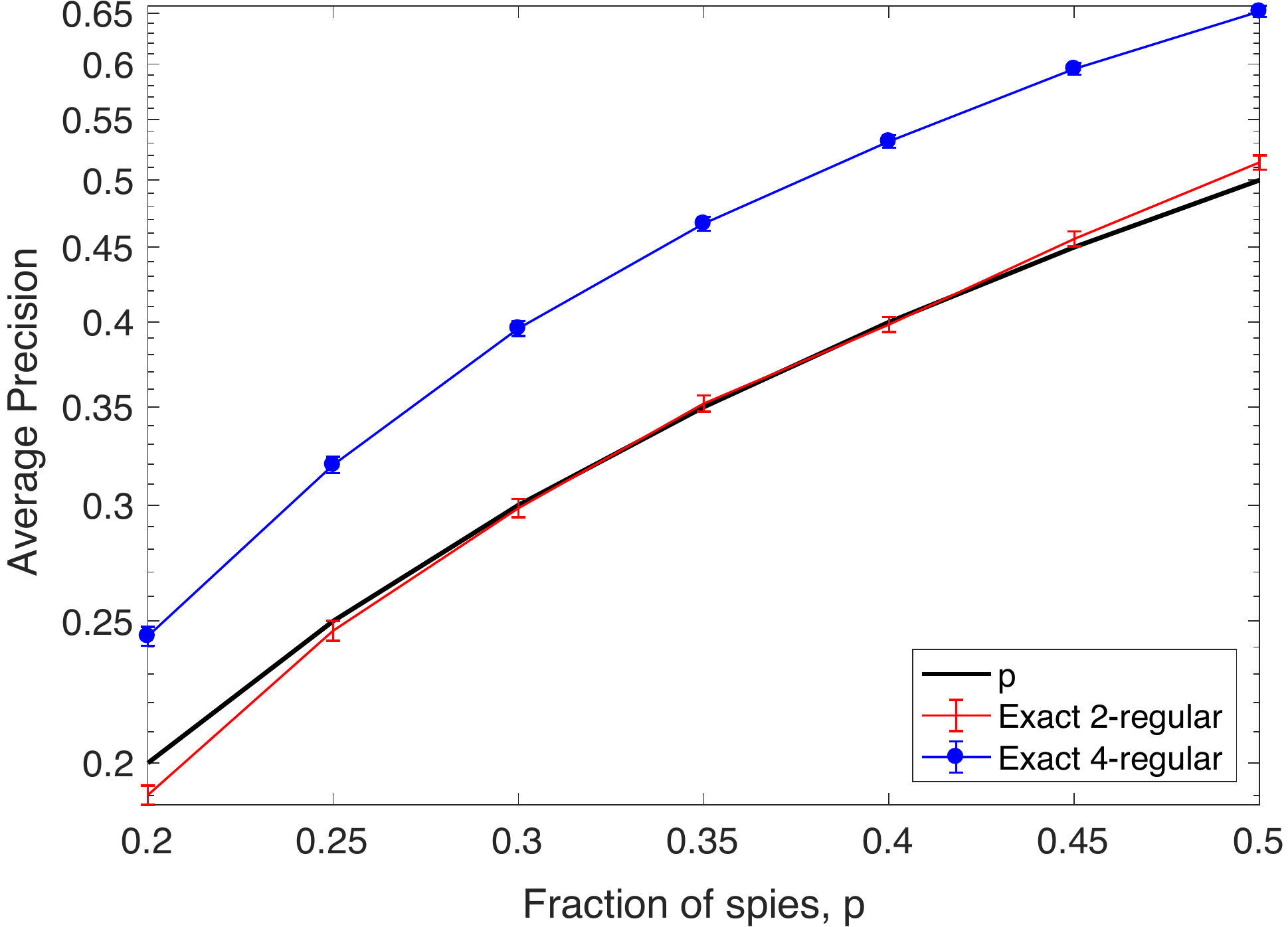}
  \caption{One-to-one transaction forwarding when the adversary knows the graph and knows the nodes' internal routing decisions.} \label{fig:known}
  \end{minipage} 
\end{figure*}

\section{Tradeoffs in graph topology} \label{app:tradeoffs}
As discussed in Section \ref{sec:intersection}, there are tradeoffs between using 4-regular graphs with one-to-one routing and line graphs (recall that on line graphs, there is only one possible routing scheme, which can be viewed as a special case of one-to-one routing). In this section, we give some more explanation and empirical results detailing these tradeoffs.
In our experiments, we have considered three different adversarial settings:
\begin{enumerate}
\item[(a)] An adversary who does not know the underlying graph
\item[(b)] An adversary who knows the graph, but does not know any routing decisions
\item[(c)] An adversary who knows the graph and also knows each node's \emph{relay} routing decisions. That is, it knows how node $v$ will route incoming transactions from edge $e$, but it does not know how $v$ will route $v$'s own transactions.
\end{enumerate}
Adversarial model (c) is motivated by the fact than an adversary can send probe transactions that get relayed through honest nodes, but it cannot force honest nodes to produce their own transactions. Hence learning a node's relay routing choices is feasible, whereas learning a node's own transaction routing choices is significantly more difficult. 
We also considered two different graph settings:
\begin{enumerate}
\item[(a)] Exact regular graphs (idealized setting)
\item[(b)] Approximate regular graphs (proposed construction mechanism detailed in Section \ref{sec: const graph})
\end{enumerate}

Since the adversary with no graph knowledge was already discussed in Section \ref{sec: topology}, we focus on adversarial models (b) and (c). 
We begin by considering an adversary that knows the graph, but not the internal routing decisions. 
Figure \ref{fig:unknown} shows the average precision of such an adversary on exact 4-regular and line graphs, each with 100 nodes. 
The error bars give the standard error for our experiments.
Our simulations show that 4-regular graphs have a lower average precision than line graphs. 
This trend is similar to what we observed for adversaries that do not know the graph, discussed in Section \ref{sec: topology}. 
However, Figure \ref{fig:known} illustrates the average precision for an adversary that knows both the graph and the internal routing decisions of each node. 
Here, the trend is reversed: line graphs have a precision that is upper bounded by $p$, whereas on 4-regular graphs, the precision can be higher than $p$.
This makes sense because on line graphs, each node has only one possible routing decision, so the additional routing knowledge of the adversary does not help.

These simulated observations suggest that if the adversary is strong enough to learn the internal, random routing decisions of each node, a line graph actually gives greater protections. 
However, we expect such learning to be expensive, since the only way to learn a node's routing decisions is to relay transactions through that node. To learn routing decisions for \emph{all} nodes in the graph will require at least a linear number of transactions in the number of nodes. 
Such an attack would quickly grow expensive, especially considering that the graph is changing periodically.

\end{document}